\documentclass[11pt,a4paper]{article}

\usepackage[utf8]{inputenc} 		
\usepackage[english]{babel}		
\usepackage{amsmath,amssymb,amsthm}	

\newtheorem{theorem}{Theorem}[section]
\newtheorem{lemma}[theorem]{Lemma}
\newtheorem{proposition}[theorem]{Proposition}
\newtheorem{corollary}[theorem]{Corollary}
\newtheorem{remark}[theorem]{Remark}

\usepackage{cancel}
\usepackage{tikz}
\usepackage{caption}
\usepackage{subcaption}
\usetikzlibrary{decorations.pathmorphing}
\usetikzlibrary{decorations.markings}
\usetikzlibrary{arrows.meta,bending}
\usepackage{graphicx}			
\usetikzlibrary{patterns}
\usetikzlibrary{arrows}
\tikzstyle directed=[postaction={decorate,decoration={markings,
		mark=at position .65 with {\arrow{latex}}}}]
\usepackage{graphicx}			
\usepackage{hyperref}
\hypersetup{colorlinks,%
	citecolor=purple,%
	filecolor=red,%
	linkcolor=cyan,%
	urlcolor=red}
\usepackage{xcolor}			
\usepackage[left=2.0cm,right=2.0cm,top=2.0cm,bottom=2.0cm]{geometry}
\usepackage{fancyhdr}
\usepackage{appendix}
\usepackage{dsfont}
\usepackage{mathrsfs}

\def\restriction#1#2{\mathchoice
	{\setbox1\hbox{${\displaystyle #1}_{\scriptstyle #2}$}
		\restrictionaux{#1}{#2}}
	{\setbox1\hbox{${\textstyle #1}_{\scriptstyle #2}$}
		\restrictionaux{#1}{#2}}
	{\setbox1\hbox{${\scriptstyle #1}_{\scriptscriptstyle #2}$}
		\restrictionaux{#1}{#2}}
	{\setbox1\hbox{${\scriptscriptstyle #1}_{\scriptscriptstyle #2}$}
		\restrictionaux{#1}{#2}}}
\def\restrictionaux#1#2{{#1\,\smash{\vrule height .6\ht1 depth .6\dp1}}_{\,#2}} 
\def\restrictionaux#1#2{{#1\,\smash{\vrule height .6\ht1 depth .6\dp1}}_{\,#2}}

\makeatletter
\def\@fnsymbol#1{\ensuremath{\ifcase#1\or \dagger\or \ddagger\or
		\mathsection\or \mathparagraph\or \|\or **\or \dagger\dagger
		\or \ddagger\ddagger \else\@ctrerr\fi}}
\makeatother

\date{}

\title{\textbf{Decay of the Local Energy for the Charged Klein-Gordon Equation in the Exterior \\De Sitter-Reissner-Nordstr\"{o}m Spacetime}}
\author{\textbf{Nicolas BESSET}\thanks{Institut Fourier - UMR 5582, Universit\'{e} Grenoble Alpes, CS 40700, France; e-mail: nicolas.besset@univ-grenoble-alpes.fr}}	

\begin{document}
\maketitle						


\begin{center}
\textbf{Abstract}
\end{center}
\begin{center}
\small
We show decay of the local energy of solutions of the charged Klein-Gordon equation in the exterior De Sitter-Reissner-Nordstr\"{o}m spacetime by means of a resonance expansion of the local propagator.
\end{center}
\section{Introduction}
%
There has been enormous progress in our understanding of scattering properties of solutions of hyperbolic equations on black hole type backgrounds over the last years. The aim of such studies is multifold. First of all these equations and thus their scattering properties are very important in their own right. Secondly, the understanding of dispersive properties of the solutions of these equations is a first step in the understanding of stability properties of the underlying spacetime. Eventually understanding the classical equation is also a first step in understanding the quantization of the field. The most important spacetime in this context is the (De Sitter-)Kerr spacetime, which is conjectured to be the unique solution of the Einstein equations describing a rotating black hole (for uniqueness results see \cite{AlIoKl2010} and references therein, and also \cite{Hi17} for the charged case). In the case of positive cosmological constant, nonlinear stability of the De Sitter-Kerr spacetime is now known for small angular velocity $|a|$ thanks to the seminal result of Hintz and Vasy \cite{HiVa16}. The case of zero cosmological constant is still open, but see \cite{KlSz17} for recent progress in this direction. Scattering theories for classical equations are also at the origin of many results in quantum field theory, see \textit{e.g.} the mathematically rigorous description of the Hawking effect in \cite{Ha09}. 

When studying linear waves on a black hole type spacetime, one encounters several difficulties. The first is linked to trapping and it is already present in the case of the (De Sitter-)Schwarzschild spacetime, which describes spherically symmetric black holes. The second is superradiance, which means that there is no positive conserved quantity for spin 1 equations  on the (De Sitter-)Kerr metric. Whereas this difficulty is not present for the wave equation on the (De Sitter-)Schwarzschild metric, it also appears when one considers a charged Klein-Gordon field on the (De Sitter-)Reissner-Nordstr\"om metric which describes a spherically symmetric charged black hole.  
In this context the phenomenon is linked to the charge of the black hole and the test particle and thus different from the Kerr case, where it is linked to the geometry of the spacetime. Superradiance already appears in flat spacetime when one considers a charged Klein-Gordon field which evolves in a strong electric field. In this context the natural setting seems to be the one of Krein spaces, see \textit{e.g.} \cite{GeoGerHa13}. This setting however is not available in the context of black holes, see \cite{GeoGerHA17}.   

In the present paper we show a resonance expansion for the solutions of the charged Klein-Gordon equation on the De Sitter-Reissner-Nordström metric. As a corollary we obtain exponential decay of the local energy for these solutions. We restrict our study to the case where the product of black hole charge and the scalar field charge is small. Such a resonance expansion for the solutions of the wave equation has been obtained first by Bony-H\"afner for the wave equation on the De Sitter-Schwarzschild metric \cite{BoHa08}. This result has been generalized to much more complicated situations which include perturbations of the De Sitter-Kerr metric by Vasy \cite{Va13}.  This last paper has developed new methods including a Fredholm theory for non elliptic problems. These methods could probably also be applied to the present case. In the present paper however we use the more elementary methods of Bony-H\"afner \cite{BoHa08} and Georgescu-G\'erard-H\"afner \cite{GeoGerHA17}.

The smallness of the charge product is non-quantitative as it is determined from a compactness
argument. It allows at many places of the paper to use perturbation arguments with respect to the non charged case. As far as we are aware the present system has not been studied for general charge products. We shall however stress out that exponentially finite energy growing modes appear when the charge product becomes large with respect to the mass of the field, see \cite{BeHa20}. The origin of this phenomenon is the existence of a resonance at 0 for the wave equation which moves to the upper half plane in the latter case; in contrast, the present paper shows that if the charge product is small enough with respect to the mass of the field then all resonances (including the resonance 0) for the wave equation lie in the lower half plane, yielding an exponential decay of the local energy as explained in the introductory example of Section \ref{Meromorphic extension and resonances} below. Let us mention that the resonance 0 still exists for the wave equation on the De Sitter-Kerr metric for small angular momentum of the black hole, see \cite{DyQNM}. In absence of cosmological constant, absence of growing modes is known for the wave equation on the Kerr metric for general angular momentum of the black hole, see \cite{Wh89}, but growing modes appear for the Klein-Gordon equation on the Kerr metric, see \cite{SR14}. The question of the existence or not of growing modes is therefore a very subtle question.

Let us also mention that it is crucial for our results that the cosmological constant is strictly positive in order to define resonances as the poles of the meromorphic extension of the weighted resolvent. The low frequency behaviour is more complicated in the zero cosmological constant case and only polynomial decay of the local energy is expected then. Using different techniques, Giorgi has recently shown a linear stability result for the sub-extremal Reissner-Nordstr\"{o}m spacetime, see \cite{Gi19}.

\paragraph{Organization of the paper.} The paper is organized as follows. In Section \ref{Functional framework} we give an introduction to the De Sitter-Reissner-Nordstr\"om metric and the charged Klein-Gordon equation on it. In Section \ref{Meromorphic extension and resonances}, a meromorphic extension result is shown for the cut-off resolvent and resonances are introduced. The resonance expansion is presented in Section \ref{Resonance expansion for the charged Klein-Gordon equation}. Suitable resolvent-type estimates are obtained in Section \ref{Estimates for the cut-off inverse of the quadratic pencil}. In section \ref{Proof of Theorem {Theoreme principal}} we prove the main theorem by a suitable contour deformation and using the resolvent-type estimates of Section \ref{Estimates for the cut-off inverse of the quadratic pencil}. The appendix contains a semiclassical limiting absorption principle for a class of generalized resolvents which might have some independent interest. 
\paragraph{Notations.} The set $\big\{z\in\mathbb{C}\ \big\vert\ \Im z\gtrless0\big\}$ will be denoted by $\mathbb{C}^{\pm}$. For any complex number $\lambda\in\mathbb{C}$, we will write $\left\langle\lambda\right\rangle:=\big(1+\left|\lambda\right|^{2}\big)^{1/2}$, $D(\lambda,R)$ will be the disc centered at $\lambda\in\mathbb{C}$ of radius $R>0$ and $D(\lambda,R)^{\complement}$ its complementary set. For all $\omega=|\omega|\mathrm{e}^{\mathrm{i}\theta}\in\mathbb{C}\setminus]-\infty,0]$, $\theta\in\mathbb{R}$, we will use the branch of the square root defined by $\sqrt{\omega}:=\sqrt{|\omega|}\mathrm{e}^{\mathrm{i}\theta/2}$.

The notation $\mathcal{C}^k_{\mathrm{c}}$ will be used to denote the space of compactly supported $\mathcal{C}^k$ functions. Also, the  Schwartz space on $\mathbb{R}$ will be denoted by $\mathscr{S}$. If $V,W$ are complex vector spaces, then $\mathcal{L}(V,W)$ will be the space of bounded linear operators $V\to W$. All the scalar products $\langle \cdot\,,\cdot\rangle$ will be antilinear with respect to their first component and linear with respect to their second component. For any function $f$, the support of $f$ will be denoted by $\mathrm{Supp\,}f$. If $A$ is an operator, we will denote by $\mathscr{D}\left(A\right)$ its domain, $\sigma\left(A\right)$ its spectrum and $\rho\left(A\right)$ its resolvent set. $A\geq 0$ will mean that $\langle Au,u\rangle\geq 0$ for all $u\in\mathscr{D}(A)$, and $A>0$ will mean that $A\geq 0$ and $\ker(A)=\{0\}$.

Now we define the symbol classes on $\mathbb{R}^{2d}$
\begin{align*}
S^{m,n}&:=\left\{a\in\mathcal{C}^{\infty}(\mathbb{R}^{2d},\mathbb{C})\ \bigg\vert\ \forall(\alpha,\beta)\in\mathbb{N}^{2d},\exists\, C_{\alpha,\beta}>0,|\partial_\xi^\alpha\partial_x^\beta a(x,\xi)|\leq C_{\alpha,\beta}\langle\xi\rangle^{m-|\alpha|}\langle x\rangle^{n-|\beta|}\right\}
\end{align*}
for any $(m,n)\in\mathbb{Z}^{2d}$ (here $\mathbb{N}$ and $\mathbb{Z}$ both include 0). We then define the semiclassical pseudodifferential operators classes
\begin{align*}
\Psi^{m,n}&:=\big\{a^{\mathrm{w}}(x,h\mathrm{D})\ \big\vert\ a\in S^{m,n}\big\},&\Psi^{-\infty,n}:=\bigcap_{m\in\mathbb{Z}}\Psi^{m,n}
\end{align*}
with $a^{\mathrm{w}}(x,h\mathrm{D})$ the Weyl quantization of the symbol $a$. For any $c>0$, the notation $P\in c\Psi^{m,n}$ means that $P\in\Psi^{m,n}$ and the norm of $P$ is bounded by a positive multiple of $c$. 
\paragraph{Acknowledgements.} I am very grateful to the anonymous referee for many suggestions to improve the manuscript, especially concerning the introduction of Section \ref{Meromorphic extension and resonances}.
%
%
%
%
\section{Functional framework}
\label{Functional framework}
%
%
%
%
\subsection{The charged Klein-Gordon equation on the De Sitter-Reissner-Nordstr\"{o}m metric}
\label{The charged Klein-Gordon equation on the De Sitter-Reissner-Nordstrom metric}
%
%
%
%
Let
\begin{align*}
	F\left(r\right):=1-\frac{2M}{r}+\frac{Q^{2}}{r^{2}}-\frac{\Lambda r^{2}}{3}
\end{align*}
with $M>0$ the mass of the black hole, $Q\in\mathbb{R}$ its electric charge and $\Lambda>0$ the cosmological constant. We assume that the parameters
\begin{align*}
	\Delta&:=1-4\Lambda Q^2,&m_k&:=\sqrt{\frac{1+(-1)^k\sqrt{\Delta}}{2\Lambda}},&&M_k:=m_k-\frac{2}{3}\Lambda m_k^3
\end{align*}
satisfy for any $k\in\{1,2\}$ the relations
\begin{align}
\label{Assumption on the main parameters of the problem}
	4\Lambda Q^2&<1,&M_1&<M<M_2
\end{align}
so that $F$ has four distinct zeros $-\infty<r_{n}<0<r_{c}<r_{-}<r_{+}<+\infty$ and is positive for all $r\in\left]r_{-},r_{+}\right[$ (see \cite[Prop. 1]{Mo17} with $\Lambda$ replaced by $\Lambda/3$ in our setting; see also \cite[Prop. 3.2]{HinlsKN} for another statement). We also assume that $9\Lambda M^2<1$ so that we can use the work of Bony-H\"{a}fner \cite{BoHa08} (the condition \eqref{Assumption on the main parameters of the problem} only ensures that $9\Lambda M^2<2$). The exterior De Sitter-Reissner-Nordstr\"{o}m spacetime is the Lorentzian manifold $(\mathcal{M},g)$ with
\begin{align*}
	&\mathcal{M}=\mathbb{R}_t\times\left]r_-,r_+\right[_r\times\mathbb{S}^2_\omega,&g=F\left(r\right)\mathrm{d}t^{2}-F\left(r\right)^{-1}\mathrm{d}r^{2}-r^{2}\mathrm{d}\omega^{2} 
\end{align*}
where $\mathrm{d}\omega^2$ is the standard metric on the unit sphere $\mathbb{S}^2$. 

Let $A:=\frac{Q}{r}\mathrm{d}t$. Then the charged wave operator on $(\mathcal{M},g)$ is
\begin{align*}
\label{Charged box operator}
	\cancel{\Box}_{g}&=\left(\nabla_\mu-\mathrm{i}qA_\mu\right)\left(\nabla^\mu-\mathrm{i}qA^\mu\right)
	=\frac{1}{F(r)}\left(\left(\partial_{t}-\mathrm{i}\frac{qQ}{r}\right)^{2}-\frac{F(r)}{r^{2}}\partial_{r}r^{2}F(r)\partial_{r}-\frac{F(r)}{r^{2}}\Delta_{\mathbb{S}^{2}}\right)
\end{align*}
with $q\in\mathbb{R}$ and the corresponding charged Klein-Gordon equation reads
\begin{align*}
	\cancel{\Box}_{g}u+m^{2}u&=0,\qquad\qquad m>0.
\end{align*}
We set $s:=qQ\in\mathbb{R}$ the charge product (which appears in the perturbation term of the standard wave operator), $X:=\left]r_-,r_+\right[_r\times\mathbb{S}^2_\omega$ and $V\left(r\right):=r^{-1}$ so that the above equation reads
\begin{equation}
\label{operators charged KG equation}
	\left(\partial_{t}-\mathrm{i}sV\right)^{2}u+\hat{P}u=0
\end{equation}
with
\begin{align}
\label{hat P}
	\hat{P}&=-\frac{F(r)}{r^{2}}\partial_{r}\left(r^{2}F(r)\partial_{r}\right)-\frac{F(r)}{r^{2}}\Delta_{\mathbb{S}{^2}}+m^{2}F(r)\nonumber\\
	&=-F(r)^{2}\partial_{r}^{2}-F\left(\frac{2F(r)}{r}+\frac{\partial F}{\partial r}(r)\right)\partial_{r}-\frac{F(r)}{r^{2}}\Delta_{\mathbb{S}^{2}}+m^{2}F(r)
\end{align}
defined on $\mathscr{D}(\hat{P}):=\big\{u\in L^{2}\left(X,F(r)^{-1}r^{2}\mathrm{d}r\mathrm{d}\omega\right)\ \big\vert\ \hat{P}u\in L^{2}\left(X,F(r)^{-1}r^{2}\mathrm{d}r\mathrm{d}\omega\right)\big\}$ (this is the spatial operator in \cite{BoHa08} with the additional mass term $m^2F(r)$). In the sequel, we will use the following notations:
\begin{align*}
	V_{\pm}&:=\lim_{r\to r_\pm}V(r)=r_{\pm}^{-1}.
\end{align*}
It turns out that the positive mass makes the study of the equation easier. Besides the fact that massless charged particles do not exist in physics, it is not excluded that the resonance 0 for the case $s=0$ (see \cite{BoHa08}) can move to $\mathbb{C}^+$ in the case $s\neq0$ and $m=0$.
%
%
%
\subsection{The Regge-Wheeler coordinate}
\label{The Regge-Wheeler coordinate}
%
%
%
%
We introduce the Regge-Wheeler coordinate $x\equiv x\left(r\right)$ defined by the differential relation
\begin{align}
\label{Regge-Wheeler coordinate}
	\frac{\mathrm{d}x}{\mathrm{d}r}&:=\frac{1}{F\left(r\right)}.
\end{align}
Using the four roots $r_{\alpha}$ of $F$, $\alpha\in I:=\left\{n,c,-,+\right\}$, we can write
\begin{align*}
	\frac{1}{F\left(r\right)}&=-\frac{3r^{2}}{\Lambda}\sum_{\alpha\in I}\frac{A_{\alpha}}{r-r_{\alpha}}
\end{align*}
where $A_{\alpha}=\prod_{\beta\in I\setminus\{\alpha\}}(r_\alpha-r_\beta)^{-1}$ for all $\alpha\in I$, and $\pm A_\pm>0$. Integrating \eqref{Regge-Wheeler coordinate} then yields
\begin{align}
\label{Regge-Wheeler x(r)}
	x\left(r\right)&=-\frac{3}{\Lambda}\sum_{\alpha\in I}A_{\alpha}r_{\alpha}^{2}\ln\left|\frac{r-r_{\alpha}}{\mathfrak{r}-r_{\alpha}}\right|
\end{align}
with $\mathfrak{r}:=\frac{1}{2}\left(3M+\sqrt{9M^{2}-8Q^{2}}\right)$ (we will explain this choice below); observe that $|Q|<\frac{3}{\sqrt{8}}M$ if \eqref{Assumption on the main parameters of the problem} holds (see the discussion below (17) in \cite{Mo17}). Therefore, we have
\begin{align*}
	\left|r-r_{\alpha}\right|&=\left|\mathfrak{r}-r_{\alpha}\right|\prod_{\beta\in I\setminus\left\{\alpha\right\}}\left|\frac{r-r_{\beta}}{\mathfrak{r}-r_{\beta}}\right|^{-A_{\beta}r_{\beta}^{2}/\left(A_{\alpha}r_{\alpha}^{2}\right)}\exp\left(-\frac{\Lambda}{3A_{\alpha}r_{\alpha}^{2}}x\right)&\forall\alpha\in I
\end{align*}
which entails the asymptotic behaviours
\begin{align}
\label{Exp convergence of r}
	F\left(r\left(x\right)\right)+\lvert r\left(x\right)-r_{\pm}\rvert&\lesssim\exp\left(-\frac{\Lambda}{3A_{\pm}r_{\pm}^{2}}x\right)&x\to\pm\infty.
\end{align}
Note here that
\begin{align}
\label{Surface gravities}
	-\frac{\Lambda}{3A_\pm r_\pm^2}&=F'(r_\pm)=2\kappa_\pm
\end{align}
where $\kappa_{-}>0$ is the surface gravity at the event horizon and $\kappa_{+}<0$ is the surface gravity at the cosmological horizon. Recall that $\kappa_\pm$ is defined by the relation
\begin{align*}
	X^\mu\nabla_\mu X^\nu&=-2\kappa_\pm X^\nu&X=\partial_t
\end{align*}
where the above equation is to be considered at the corresponding horizon.

In the appendix \ref{Proof of analytic extension}, we follow \cite[Prop. IV.2]{BaMo93} to show the extension result:
\begin{proposition}
	\label{Analytic extension of r}
	There exists a constant $\mathscr{A}>0$ such that the function $x\mapsto r(x)$ extends analytically to $\left\{\lambda\in\mathbb{C}\mid\left|\Re\lambda\right|>\mathscr{A}\right\}$.
\end{proposition}
On $L^{2}\left(X,\mathrm{d}x\mathrm{d}\omega\right)$, define the operator $P:=r\hat{P}r^{-1}$, given in the coordinates $(x,\omega)$ by the expression
\begin{align*}
	P&=-r^{-1}\partial_{x}r^{2}\partial_{x}r^{-1}-\frac{F\left(r\right)}{r^{2}}\Delta_{\mathbb{S}^{2}}+m^{2}F\left(r\right)=-\partial_{x}^{2}-W_{0}\Delta_{\mathbb{S}^{2}}+W_{1}
\end{align*}
where
\begin{align}
\label{Potentiels W0 et W1}
	W_{0}\left(x\right):=\frac{F\left(r\left(x\right)\right)}{r\left(x\right)^{2}},\quad W_{1}\left(x\right):=\frac{F\left(r\left(x\right)\right)}{r\left(x\right)}\frac{\partial F}{\partial r}\left(r\left(x\right)\right)+m^{2}F\left(r\left(x\right)\right).
\end{align}
It will happen in the sequel that we write $F\left(x\right)$ for $F\left(r\left(x\right)\right)$ and also $V\left(x\right)$ for $V\left(r\left(x\right)\right)$. Observe that the potentials $W_0$ and $W_1$ satisfy the same estimate as in \eqref{Exp convergence of r}.

As
\begin{align*}
	\frac{\mathrm{d}W_{0}}{\mathrm{d}x}&=F(r)\frac{\mathrm{d}W_{0}}{\mathrm{d}r}=\frac{2F(r)}{r^{5}}\left(3Mr-2Q^{2}-r^{2}\right),
\end{align*}
we see that the (unstable) maximum of $W_{0}$ occurs when $x=0$, \textit{i.e.} $r=\mathfrak{r}=\frac{1}{2}\left(3M+\sqrt{9M^{2}-8Q^{2}}\right)$\nolinebreak: this is the radius of the photon sphere. It is the only trapping set outside the black hole for null geodesics (see \cite{Mo17}). The trapping will have a consequence on some resolvent type estimates, see the paragraph \ref{Estimates in the zone III}.
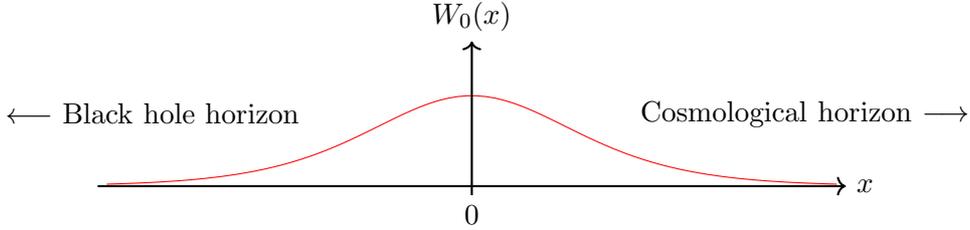
\begin{figure}[!h]
	\centering
	\captionsetup{justification=centering,margin=1.8cm}
	\begin{center}
		\begin{tikzpicture}[scale=1.2]
		\draw[red][domain=-4.0:4.0,samples=80]plot(\x,{4/(exp(0.65*\x)+exp(-0.65*\x))^2});
		\draw[->][thick](0,-0.1)--(0,1.6)node[above]{$W_{0}(x)$};
		\draw[->][thick](-4.1,0)--(4.1,0)node[right]{$x$};
		\draw(-3.5,0.8)node[]{$\longleftarrow$ Black hole horizon};
		\draw(3.65,0.8)node[]{Cosmological horizon $\longrightarrow$};
		\draw(0,-0.1)node[anchor=north]{$0$};
		\end{tikzpicture}
		\caption{\label{The potential W0}The potential $W_{0}$ in the Regge-Wheeler coordinates.}
	\end{center}
\end{figure}
%
%
%
%
\subsection{The charge Klein-Gordon operator}
\label{The charge Klein-Gordon operator}
%
%
%
%
Taking advantage of the spherical symmetry, we write
\begin{align*}
	&L^{2}\left(\mathbb{R}\times\mathbb{S}^{2},\mathrm{d}x\mathrm{d}\omega\right)\simeq\bigoplus_{\ell\in\mathbb{N}}\left(L^{2}\left(\mathbb{R},\mathrm{d}x\right)\otimes Y_{\ell}\right)=:\bigoplus_{\ell\in\mathbb{N}}\mathcal{V}_{\ell}
\end{align*}
where for all $\ell\in\mathbb{N}$, $Y_{\ell}$ is the $(2\ell+1)$-dimensional eigenspace of the operator $\left(-\Delta_{\mathbb{S}^{2}},H^{2}(\mathbb{S}^{2},\mathrm{d}\omega)\right)$ associated to the eigenvalue $\ell\left(\ell+1\right)$. On each $\mathcal{V}_{\ell}$, we define $P_{\ell}$ as the restriction of $P$ onto $\mathcal{V}_{\ell}$ which will be identified with an operator acting on $L^2(\mathbb{R},\mathrm{d}x)$, i.e.
\begin{equation}
\label{Expression of h ell}
	P_{\ell}=-\partial_{x}^{2}+\ell\left(\ell+1\right)W_{0}+W_{1}
\end{equation}
and we set $\mathscr{D}\left(P_{\ell}\right):=H^{2}(\mathbb{R},\mathrm{d}x)$ so that $P_{\ell}$ is self-adjoint. In the sequel, we will use the following (self-adjoint) realization of the total operator $P$:
\begin{align*}
	P&:=\bigoplus_{\ell\in\mathbb{N}}P_{\ell},&\mathscr{D}\left(P\right)&:=\Big\{u=(u_{\ell})_{\ell\in\mathbb{N}}\in\bigoplus_{\ell\in\mathbb{N}}\mathcal{V}_{\ell}\mid\forall\ell\in\mathbb{N},\ u_\ell\in\mathscr{D}(P_\ell)\Big\}.
\end{align*}
Now the charged Klein-Gordon equation reads
\begin{equation}
\label{operators charged KG equation with h}
	\left(\partial_{t}-\mathrm{i}sV\right)^{2}u+Pu=0.
\end{equation}
The point is to see that if $u$ is a solution of \eqref{operators charged KG equation with h}, then $v:=\left(u,-\mathrm{i}\partial_{t}u-sVu\right)$ solves the first order equation
\begin{align}
\label{First order equation with the charge Klein-Gordon operator}
	-\mathrm{i}\partial_{t}v=\hat{K}(s)v
\end{align}
where
\begin{align}
\label{Hamiltonien H}
	\hat{K}\left(s\right)&:=
	\begin{pmatrix}
	sV&\mathrm{Id} \\
	P&sV
	\end{pmatrix}
\end{align}
is the charge Klein-Gordon operator. Conversely, if $v=\left(v_{0},v_{1}\right)$ solves \eqref{First order equation with the charge Klein-Gordon operator}, then $v_0$ solves \eqref{operators charged KG equation with h}. We also define $\hat{K}_{\ell}\equiv\hat{K}_\ell(s)$\footnote{We will often drop the dependence in $s$.} with $P_{\ell}$ in place of $P$ for any $\ell\in\mathbb{N}$. Following \cite{GeoGerHA17}, we realize $\hat{K}_{\ell}$ with the domain
\begin{align*}
	\mathscr{D}(\hat{K}_{\ell})&:=\left\{u\in P_{\ell}^{-1/2}L^2(\mathbb{R},\mathrm{d}x)\oplus L^2(\mathbb{R},\mathrm{d}x)\mid\hat{K}_{\ell}u\in P_{\ell}^{-1/2}L^2(\mathbb{R},\mathrm{d}x)\oplus L^2(\mathbb{R},\mathrm{d}x)\right\}
\end{align*}
and realize the operator $\hat{K}$ as the direct sum on $\mathbb{N}\ni\ell$ of the $\hat{K}_{\ell}$.

Let $\dot{\mathcal{E}}_{\ell}$ be the completion of $P_{\ell}^{-1/2}L^2(\mathbb{R},\mathrm{d}x)\oplus L^2(\mathbb{R},\mathrm{d}x)$ for the norm\footnote{Note that the norm $\|.\|^2_{_{\dot{\mathcal{E}}_\ell}}$ is conserved if $[P_{\ell},sV]=0$; it is the case if $s=0$.}
\begin{align*}
	&\|u\|^2_{_{\dot{\mathcal{E}}_\ell}}:=\langle u_0,P_\ell u_0\rangle_{_{L^2(\mathbb{R},\mathrm{d}x)}}+\|u_1-sVu_0\|^2_{_{L^2(\mathbb{R},\mathrm{d}x)}}&u=(u_0,u_1)\in\dot{\mathcal{E}_\ell}
\end{align*}
and define $\big(\dot{\mathcal{E}},\|.\|_{\dot{\mathcal{E}}}\big)$ as the direct sum of the spaces $\dot{\mathcal{E}}_{\ell}$. \cite[Lem. 3.19]{GeoGerHA17} shows that $\hat{K}_\ell$ generates a continuous one-parameter group $(\mathrm{e}^{-\mathrm{i}t\hat{K}_\ell})_{t\in\mathbb{R}}$ on $(\dot{\mathcal{E}}_\ell, \|.\|_{_{\dot{\mathcal{E}}_\ell}})$. We similarly construct the spaces $\big(\mathcal{E}_\ell,\|.\|_{\mathcal{E}_\ell}\big)$ and $\big(\mathcal{E},\|.\|_{\mathcal{E}}\big)$ with $\langle P_\ell\rangle$ instead of $P_\ell$. Let us mention here that for any $n\in\mathbb{R}$ the quantity
\begin{align}
\label{Natural energies}
	\langle v\!\mid\!v\rangle_n:=\langle v_1-nv_0,v_1-nv_0\rangle_{_{L^2(\mathbb{R},\mathrm{d}x)}}+\langle(P-(sV-n)^2\,)v_0,v_0\rangle_{_{L^2(\mathbb{R},\mathrm{d}x)}}
\end{align}
is formally conserved if $v=(u,-\mathrm{i}\partial_t u)$ with $u$ solution of \eqref{operators charged KG equation with h} and is continuous with respect to the norm $\|.\|_{\mathcal{E}}$. However, it is in general  not positive nor continuous with respect to the norm $\|.\|_{\dot{\mathcal{E}}}$ (see \cite[§3.4.3]{GeoGerHA17} for more details): this is superradiance. When $\Lambda=0$ (that is, when the cosmological horizon is at infinity), the natural energy $\langle.\!\mid\!.\rangle_{sV_-}$ is positive for $s$ small enough and it can be used to define a Hilbert space framework.

An important observation is the fact that the norms $\|.\|_{_{\dot{\mathcal{E}}_\ell}}$ and $\|.\|_{_{\mathcal{E}_\ell}}$  are locally equivalent, meaning that for any $v\in\dot{\mathcal{E}}$ and any cut-off $\chi\in\mathcal{C}^{\infty}_{\mathrm{c}}\left(\mathbb{R},\mathbb{R}\right)$, we have
\begin{align}
\label{Local equivalence of norms}
	\|\chi v\|_{_{\dot{\mathcal{E}}}}&\lesssim\|\chi v\|_{_{\mathcal{E}}}\lesssim\|\chi v\|_{_{\dot{\mathcal{E}}}}.
\end{align}
The first inequality is obvious, and the second one is established with the Hardy type estimate $\|\chi v\|_{_{L^2}}\lesssim\|P^{1/2}v\|_{_{L^2}}$ (see \cite[Lem. 9.5]{GeoGerHA17}; the validity of this result in our setting is discussed in the subsection \ref{Meromorphic extension} below).
%
%
%
%
%
%
%
%
%
%
%
%
%
\subsection{The quadratic pencil}
\label{The quadratic pencil}
Let $u$ be a solution of \eqref{operators charged KG equation with h}. If we look for $u$ of the form $u=\mathrm{e}^{\mathrm{i}zt}v$ with $z\in\mathbb{C}$ for some $v$, then $v$ satisfies the equation $(P-(z-sV)^2)v=0$. We define the harmonic quadratic pencil
\begin{align*}
	p_{\ell}\left(z,s\right)&:=P_{\ell}-\left(z-sV\right)^{2},&\mathscr{D}(p_\ell(z,s))&:=\langle P_\ell\rangle^{-1}L^{2}(\mathbb{R},\mathrm{d}x)=H^{2}(\mathbb{R},\mathrm{d}x)
\end{align*}
and realize the total quadratic pencil as
\begin{align*}
	&p\left(z,s\right):=\bigoplus_{\ell\in\mathbb{N}}p_{\ell}\left(z,s\right),\\
	&\mathscr{D}(p(z,s)):=\Big\{u=(u_{\ell})_{\ell\in\mathbb{N}}\in\bigoplus_{\ell\in\mathbb{N}}\mathcal{V}_{\ell}\;\Big\vert\;\forall\ell\in\mathbb{N},\,u_\ell\in\mathscr{D}(p_\ell(z,s)),\,\sum_{\ell\in\mathbb{N}}\|p_{\ell}(z,s)u_\ell\|_{_{\mathcal{V}_{\ell}}}^{2}<+\infty\Big\}.
\end{align*}
\cite[Prop. 3.15]{GeoGerHA17} sets the useful relations
\begin{align}
\label{Relation between the domains of the quadratic pencil and the corresponding Hamiltonian}
	\rho(\hat{K}_\ell)\cap\mathbb{C}\setminus\mathbb{R}&=\big\{z\in\mathbb{C}\setminus\mathbb{R}\;\big\vert\;p_{\ell}\left(z,s\right):H^{2}(\mathbb{R},\mathrm{d}x)\to L^{2}(\mathbb{R},\mathrm{d}x)\text{ is bijective}\big\}
\end{align}
and
\begin{equation}
\label{Resolvent of H}
	\hat{R}_\ell\left(z,s\right):=(\hat{K}_\ell(s)-z)^{-1}=\begin{pmatrix}
	p_\ell\left(z,s\right)^{-1}(z-sV)&p_\ell\left(z,s\right)^{-1}\\\mathrm{Id}+(z-sV)\,p_\ell\left(z,s\right)^{-1}(z-sV)&(z-sV)\,p_\ell\left(z,s\right)^{-1}
	\end{pmatrix}
\end{equation}
for all $z\in\rho(\hat{K}_\ell)\cap\mathbb{C}\setminus\mathbb{R}$. In comparison, the relation (1.7) in \cite{BoHa08} involves the resolvent of $P_{\ell}$, which corresponds to the case $s=0$ for us. \cite[Prop. 3.12]{GeoGerHA17} shows that \eqref{Resolvent of H} is also valid for $z\in\rho(\hat{K}_{\ell})\cap\mathbb{R}$ when we work on $(\mathcal{E}_\ell,\|.\|_{_{\mathcal{E}_\ell}})$; by using the local equivalence \eqref{Local equivalence of norms} of the norms $\|.\|_{_{\dot{\mathcal{E}}_\ell}}$ and $\|.\|_{_{\mathcal{E}_\ell}}$, we can use \eqref{Resolvent of H} for $z\in\rho(\hat{K}_{\ell})\cap\mathbb{R}$ if we consider the cut-off resolvent $\chi\hat{R}_\ell\left(z,s\right)\chi$ with $\chi\in\mathcal{C}^{\infty}_{\mathrm{c}}\left(\mathbb{R},\mathbb{R}\right)$. In the sequel, we will simply call $p_{\ell}(z,s)$ the quadratic pencil when $\ell\in\mathbb{N}$ will be fixed.
%
%
%
%
%
%
%
%
\section{Meromorphic extension and resonances}
\label{Meromorphic extension and resonances}
We construct in this Section a meromorphic extension for the weighted resolvent of $\hat{K}(s)$. The main Theorem \ref{Theoreme principal}, which provides asymptotic decay (in time) for solutions of the charged Klein-Gordon equation \eqref{operators charged KG equation}, relies on such a construction. It will be established by means of the \textit{theory of resonances}. Let us briefly introduce the basic idea of this theory.
\paragraph{An introductory example to the theory of resonances.} An explicit and comprehensive example is given by the one-dimensional wave equation on $\mathbb{R}$:
\begin{align}
\label{Wave eq 1d}
	(\partial_{t}^{2}-\partial_{x}^{2})u&=f\in L^{2}(\mathbb{R},\mathrm{d}x).
\end{align}
Here $-\partial_{x}^{2}$ is the self-adjoint realization of the one-dimensional Laplacian on $H^{2}(\mathbb{R},\mathrm{d}x)$. Taking the time-dependent Fourier transform (denoted by the symbol $\hat{\ }\,$), we get
\begin{align*}
	(-\partial_{x}^{2}-z^{2})\hat{u}&=\hat{f}
\end{align*}
which is the one-dimensional Helmholtz equation. Call $R(z)$ the resolvent $(-\partial_{x}^{2}-z^{2})^{-1}$. Then $R(z)$ is well-defined for $z\in\mathbb{C}^{+}$ and we have for such spectral parameters
\begin{align*}
	\hat{u}&=R(z)\hat{f}.
\end{align*}
In this very simple example, we have an explicit representation formula:
\begin{align*}
	(R(z)\hat{f})(x)&=\frac{\mathrm{i}}{2z}\int_{-\infty}^{+\infty}\mathrm{e}^{\mathrm{i}z|x-y|}\hat{f}(y)\mathrm{d}y\qquad\qquad\forall z\in\mathbb{C}^{+},\forall x\in\mathbb{R}.
\end{align*}
Let us set $R(z;x,y):=\frac{\mathrm{i}}{2z}\mathrm{e}^{\mathrm{i}z|x-y|}$ the kernel of $R(z)$. By Schur's estimate,
\begin{align*}
	\|R(z)\|_{_{L^{2}(\mathbb{R},\mathrm{d}x)\to L^{2}(\mathbb{R},\mathrm{d}x)}}&\leq\left(\sup_{x\in\mathbb{R}}\int_{-\infty}^{+\infty}|R(z;x,y)|\mathrm{d}y\right)^{1/2}\left(\sup_{y\in\mathbb{R}}\int_{-\infty}^{+\infty}|R(z;x,y)|\mathrm{d}x\right)^{1/2}=\frac{1}{|z||\Im z|}.
\end{align*}
Now take an \textit{exponential} weight $w(x)=\mathcal{O}_{|x|\to+\infty}\big(\mathrm{e}^{-\kappa|x|}\big)$ for some $\kappa>0$. Then Schur's estimate gives
\begin{align*}
	\|w^{\delta}R(z)w^{\delta}\|_{_{L^{2}(\mathbb{R},\mathrm{d}x)\to L^{2}(\mathbb{R},\mathrm{d}x)}}&\lesssim\frac{1}{|z|(\kappa-|\Im z|)}
\end{align*}
for all $\delta>0$ and $z\in\overline{\mathbb{C}^{-}}\setminus\{0\}$ such that $\Im z>-\kappa\delta$. We can thus extend the weighted resolvent $w^{\delta}R(z)w^{\delta}$ from $\mathbb{C}^{+}$ to $\{\lambda\in\mathbb{C}\setminus\{0\}\mid\Im\lambda>-\kappa\delta\}$ as an analytic family of bounded operators acting on $L^{2}(\mathbb{R},\mathrm{d}x)$. We can easily check that
\begin{align*}
	\lim_{z\to 0}\|zw^{\delta}R(z)w^{\delta}\|_{_{L^{2}(\mathbb{R},\mathrm{d}x)\to L^{2}(\mathbb{R},\mathrm{d}x)}}&<+\infty
\end{align*}
for any $\delta>0$. Therefore $z=0$ is a pole of the extension: it is called \textit{resonance}. The operator $w^{\delta}R(z)w^{\delta}:L^{2}(\mathbb{R},\mathrm{d}x)\to L^{2}(\mathbb{R},\mathrm{d}x)$ is then the \textit{meromorphic extension} of the original weighted resolvent. Notice that we can replace any exponential weight by a cut-off $\chi\in\mathcal{C}^{\infty}_{\mathrm{c}}(\mathbb{R},\mathbb{R})$, for
\begin{align*}
	\|\chi R(z)\chi\|_{_{L^{2}(\mathbb{R},\mathrm{d}x)\to L^{2}(\mathbb{R},\mathrm{d}x)}}&\leq\|\chi w^{-\delta}\|_{_{L^{2}(\mathbb{R},\mathrm{d}x)\to L^{2}(\mathbb{R},\mathrm{d}x)}}^{2}\|w^{\delta}R(z)w^{\delta}\|_{_{L^{2}(\mathbb{R},\mathrm{d}x)\to L^{2}(\mathbb{R},\mathrm{d}x)}}.
\end{align*}
Collecting estimate for the cut-off resolvent is generally easier. Putting aside the regularity issue, we have the inversion formula
\begin{align*}
	u(t,x)&=\frac{1}{2\pi}\int_{-\infty+\mathrm{i}\nu}^{+\infty+\mathrm{i}\nu}\mathrm{e}^{-\mathrm{i}zt}R(z)\hat{f}(z,x)\mathrm{d}z
\end{align*}
for some $\nu>0$. If $f$ is compactly supported in $x\in\mathbb{R}$, then
\begin{align}
\label{Representation formula}
	\chi u(t,x)&=\frac{1}{2\pi}\int_{-\infty+\mathrm{i}\nu}^{+\infty+\mathrm{i}\nu}\mathrm{e}^{-\mathrm{i}zt}\chi R(z)\chi\hat{f}(z,x)\mathrm{d}z
\end{align}
for any cut-off $\chi$ such that $\chi\equiv 1$ on $\mathrm{Supp\,}f(t,\cdot)$. Provided that one has some nice estimates on $\chi R(z)\chi$, we can use a contour deformation to obtain integrals in $\mathbb{C}^{-}$ which provides exponentially decaying terms. In the meanwhile, the residue theorem makes appear the resonances; in the present example, this yields the following local energy estimate: for $t\gg 0$,
\begin{align}
\label{Decay example}
	\|\chi u(t,\cdot)\|_{_{H^{2}(\mathbb{R},\mathrm{d}x)}}&=\frac{1}{2}\chi\langle 1,\chi\hat{f}(0)\rangle_{_{L^{2}(\mathbb{R},\mathrm{d}x)}}+E(t),\qquad\qquad E(t)\lesssim\mathrm{e}^{-\nu t}\|f\|_{_{H^{2}(\mathbb{R},\mathrm{d}x)}}
\end{align}
for any $\nu>0$. For a more detailed and rigorous derivation of the above result in more general setting, we refer to the book of Dyatlov-Zworski \cite{DyZw}. Let us make some comments on \eqref{Decay example}:
\begin{enumerate}
	\item The non-vanishing term in \eqref{Decay example} is the projection on the \textit{resonant state} $x\mapsto 1$. This function indeed solves \eqref{Wave eq 1d} but is not integrable so it is not an eigenvalue of $-\partial_{x}^{2}$ in $H^{2}(\mathbb{R},\mathrm{d}x)$ (it is in some weighted Sobolev spaces). The resonance $0$ and the resonant state $1$ for the wave equation exist also on the De Sitter-Schwarzschild metric (see Bony-H\"{a}fner \cite{BoHa08}, formula (1.9); observe that the resonant state is $r$ therein because of the transformation $r\hat{P}r^{-1}$ below equation (1.3)) as well as on De Sitter-Kerr metric (see Dyatlov \cite{DyQNM}, formula (1.5)).
	\item One may notice that there is no loss of derivatives in \eqref{Decay example}, that is the norms are the same both on the left-hand and right-hand sides. We know that decay estimates without loss of derivatives become false in presence of an obstacle by the work of Ralston \cite{Ra69}. And indeed, for the wave equation on the De Sitter-Schwarzschild and De Sitter-Kerr metrics, where there exist trapping sets (the so-called photon sphere in the spherically symmetric case, as the one we have discussed about in Subsection \ref{The Regge-Wheeler coordinate}), there is a loss of angular derivatives in the estimates in \cite{BoHa08} and \cite{DyQNM}.
	\item We see in \eqref{Representation formula} the importance of the meromorphic extension of the cut-off resolvent as well as the localization of resonances in the complex plane. Any resonance in $\mathbb{C}^{+}$ gives exponentially growing terms (the corresponding resonant state is called \textit{growing mode}); conversely, any resonance in $\mathbb{C}^{-}$ gives exponentially decaying terms. The existence of real resonances has more subtle consequences: it leads to polynomially growing terms if the multiplicity of the resonance (as a pole of the meromorphic extension) is greater than 1, or to a stationary term with no growth or decay in time for a resonance of multiplicity equal to 1 (this happens in \cite{BoHa08} and \cite{DyQNM}).
	\item In the non-superradiant case \cite{BoHa08}, the meromorphic extension of the weighted resolvent $w^{\delta}(P-z^{2})^{-1}w^{\delta}$ with $w:=\sqrt{(r-r_-)(r_+-r)}$ is a direct consequence of the work of Mazzeo-Melrose \cite{MaMe87}. In the superradiant case however, as in \cite{DyQNM} or in the present paper, the resolvent $(P-z^{2})$ is replaced by the quadratic pencil $p(z,s)$, and \cite{MaMe87} can not be directly applied anymore.
\end{enumerate}
\paragraph{The charged Klein-Gordon equation on the De Sitter-Reissner-Nordstr\"{o}m metric.} In our setting, we know that the presence of the photon sphere will eventually lead to a loss of derivative in the estimates. Moreover, because of the superradiance, the homogeneous norm $\|\cdot\|_{_{\dot{\mathcal{E}}}}$ is not conserved and may grow (possibly exponentially fast) in time. Yet, it will turn out that no pole lies on and above the real axis provided that $s$ remains small, entailing an exponential decay of the solutions \linebreak%
of equation \eqref{operators charged KG equation}.

As motivated above, we first try to construct a meromorphic extension of the weighted resolvent of $\hat{K}(s)$. The presence of the mixed term $sV\partial_{t}$ in \eqref{operators charged KG equation} prevents us to directly use Mazzeo-Melrose result \cite{MaMe87}. However, $p(z,s)^{-1}$ formally tends as $s\to 0$ to $(P-z^{2})^{-1}$ for which \cite{MaMe87} applies; moreover, the case $s=0$ is very similar (even easier) to the case treated in \cite{BoHa08}. We will therefore try to obtain results for small $s$ using perturbation arguments. Our strategy is the following one:
\begin{enumerate}
	\item [$(i)$] Define first suitable "asymptotic" energy spaces by removing the troublesome negative contributions from the electromagnetic potential $sV$ near $r_\pm$ and define "asymptotic" selfadjoint Hamiltonians $\hat{H}_\pm(s)$ (see the paragraph \ref{Notations} below).
	\item [$(ii)$] For $s=0$, the situation is really similar to the Klein-Gordon equation on De Sitter-Schwarzschild metric: using the standard results \cite{BoHa08} and \cite{MaMe87}, we can meromorphically extend the weighted resolvent of $\hat{H}_\pm(0)$ from $\mathbb{C}^+$ to $\mathbb{C}$ with no poles on and above the real axis (see Lemma \ref{Meromorphic extension of asymptotic Hamiltonians} and Lemma \ref{Resonance 0 for P-z^2}).
	\item [$(iii)$] If $s$ remains small, we can use analytic Fredholm theory to get a meromorphic extension for the weighted resolvents of the asymptotic Hamiltonians $\hat{H}_\pm(s)$ into a strip in $\mathbb{C}^-$ (the perturbation argument entails a bound on the width of this strip which is directly linked to the rate of decay of the potentials $W_0$ and $W_1$ in $P$ near $r_\pm$). We will also get the absence of poles near the real axis (see Lemma \ref{Meromorphic extension of asymptotic Hamiltonians} and Lemma \ref{Extension of asymptotic Hamiltonians}).
	\item [$(iv)$] Finally, we construct a parametrix for the resolvent of an equivalent operator to $\hat{K}(s)$ by gluing together the resolvent of $\hat{H}_\pm(s)$ (see \eqref{Parametrix}). Using again the analytic Fredholm theory for $s$ sufficiently small, we show the existence of the weighted resolvent and also that the poles can only lie below the real axis (see Theorem \ref{No resonances in a strip near the real axis for s small}).
\end{enumerate}
The sequel of this Section is organized as follows: Subsection \ref{Notations} introduces notations and tools (operators, functional spaces) from \cite{GeoGerHA17} which will be used for the construction of the meromorphic extension of the weighted resolvent of $\hat{K}(s)$. Subsection \ref{Verification of the abstract assumptions (G)} aims to show that results obtained in \cite{GeoGerHA17} are available for us. Then Subsection \ref{Study of the asymptotic Hamiltonians} establishes the announced results for the asymptotic Hamiltonians $\hat{H}_\pm(s)$. Subsection \ref{Meromorphic extension} eventually gives the proof of the existence of the meromorphic extension of the weighted resolvent of $\hat{K}(s)$ and also shows that the poles in any compact neighbourhood of 0 lie below the real axis.
\subsection{Notations}
\label{Notations}
We introduce some notations following \cite[§2.1]{GeoGerHA17}. First observe that if $u$ solves \eqref{operators charged KG equation}, then $v:=\mathrm{e}^{-\mathrm{i}sV_{+}t}u$ satisfies
\begin{align*}
	(\partial_t^2-2\mathrm{i}s(V(r)-V_{+})\partial_t-s^2(V(r)-V_{+})^2+P)v&=0.
\end{align*}
We can therefore work with the potential\footnote{From a geometrical point of view, we are changing the gauge. Namely, $\frac{Q}{r}\mathrm{d}t$ is replaced by $\left(\frac{Q}{r}-\frac{Q}{r_+}\right)\mathrm{d}t$ which does not degenerate anymore at $r=r_+$. To see this, we use the standard Eddington-Finkelstein advanced and retarded coordinates $u=t-x$, $v=t+x$ to define the horizons: we have locally near the cosmological horizon $\mathrm{d}t=\mathrm{d}u+\mathrm{d}x$, $\mathrm{d}t=\mathrm{d}v-\mathrm{d}x$ and then $\frac{\mathrm{d}t}{\mathrm{d}r}=\pm F(r)^{-1}$. We eventually use that $\left(\frac{1}{r}-\frac{1}{r_+}\right)F(r)^{-1}$ remains bounded and does not vanish at $r=r_+$.} $\tilde{V}:= V-V_+=\mathcal{O}_{r\to r_+}(r_+-r)$ in this Section. In order not to overload notations, we will still denote $\tilde{V}$ by $V$ and $\lim_{r\to r_{\pm}}\tilde{V}(r)=V_{\pm}$.

Let us define $\mathcal{H}:=L^{2}\left(X,\mathrm{d}r\mathrm{d}\omega\right)$ and
\begin{align}
\label{Operator P rund}
	\mathscr{P}&:=rF(r)^{-1/2}\hat{P}r^{-1}F(r)^{1/2}=-r^{-1}F(r)^{1/2}\partial_{r}r^{2}F(r)\partial_{r}r^{-1}F(r)^{1/2}-\frac{F(r)}{r^{2}}\Delta_{\mathbb{S}{^2}}+m^{2}F(r)
\end{align}
with $\hat{P}$ given by \eqref{hat P}. Since $u\mapsto r^{-1}F^{1/2}u$ is an unitary isomorphism from $\mathcal{H}$ to $L^{2}\left(X,F^{-1}r^{2}\mathrm{d}r\mathrm{d}\omega\right)$, the results obtained below on $\mathscr{P}$ will also apply to $\hat{P}$ (and thus to $P$). Observe that the space $\dot{\mathcal{E}}$ has been defined in our setting with the operator $P$ which is $r\hat{P}r^{-1}$ expressed with the Regge-Wheeler coordinate, and $\hat{P}$ is equivalent to $\mathscr{P}$ as explained above; in the sequel, we will denote by $\dot{\mathscr{E}}$ the completion of $\mathscr{P}^{-1/2}\mathcal{H}\oplus\mathcal{H}$ for the norm $\|(u_0,u_1)\|_{_{\dot{\mathscr{E}}}}^2:=\langle u_0,\mathscr{P}u_0\rangle_{_{\mathcal{H}}}+\|u_1-sVu_0\|^2_{_{\mathcal{H}}}$. Let $i_\pm,j_\pm\in\mathcal{C}^{\infty}(\left]r_-,r_+\right[,\mathbb{R})$ such that
\begin{align*}
&i_\pm=j_\pm=0\text{ close to $r_\mp$},\qquad i_\pm=j_\pm=1\text{ close to $r_\pm$},\\
&i_-^2+i_+^2=1,\qquad i_\pm j_\pm=j_\pm,\qquad i_-j_+=i_+j_-=0.
\end{align*}
We then define the operators
\begin{align*}
&k_\pm:=s(V\mp j_\mp^2V_{-}),\qquad \mathscr{P}_\pm:=\mathscr{P}-k_\pm^2,\qquad \tilde{\mathscr{P}}_-:=\mathscr{P}-(sV_{-}-k_-)^2.
\end{align*}
We now define the isomorphism on $\dot{\mathscr{E}}$ (see comments above Lemma 3.13 in \cite{GeoGerHA17}) 
	\begin{align*}
	\Phi(sV)&:=\begin{pmatrix}
	\mathrm{Id}&0\\sV&\mathrm{Id}
	\end{pmatrix}
	\end{align*}
and we introduce the energy Klein-Gordon operator
\begin{align*}
	\hat{H}(s)&=\Phi(sV)\hat{K}(s)\Phi^{-1}(sV)=\begin{pmatrix}
	0&\mathrm{Id}\\
	\mathscr{P}-s^2V^2&2sV
	\end{pmatrix}
\end{align*}
with domain
\begin{align*}
	\mathscr{D}(\hat{H}(s))&=\left\{u\in\mathscr{P}^{-1/2}\mathcal{H}\oplus\mathcal{H}\mid \hat{H}(s)u\in\mathscr{P}^{-1/2}\mathcal{H}\oplus\mathcal{H}\right\}
\end{align*}
and the asymptotic Hamiltonians
\begin{align*}
	\hat{H}_\pm(s)&=\begin{pmatrix}
	0&\mathrm{Id}\\
	\mathscr{P}_\pm&2k_\pm
	\end{pmatrix}
\end{align*}
with domains
\begin{align*}
	\mathscr{D}(\hat{H}_+(s))&=\left(\mathscr{P}_+^{-1/2}\mathcal{H}\cap \mathscr{P}_+^{-1}\mathcal{H}\right)\oplus\langle \mathscr{P}_+\rangle^{-1/2}\mathcal{H},\\
	\mathscr{D}(\hat{H}_-(s))&=\Phi(sV_{-})\left(\tilde{\mathscr{P}}_-^{-1/2}\mathcal{H}\cap\tilde{\mathscr{P}}_-^{-1}\mathcal{H}\right)\oplus\langle\tilde{\mathscr{P}}_-\rangle^{-1/2}\mathcal{H}.
\end{align*}
These operators are self-adjoint on the following spaces (see the beginning of the paragraph 5.2 in \cite{GeoGerHA17}):
\begin{align*}
	\dot{\mathscr{E}}_{+}&:=\mathscr{P}_+^{-1/2}\mathcal{H}\oplus\mathcal{H},\\
	\dot{\mathscr{E}}_{-}&:=\Phi(sr^{-1}_-)\left(\tilde{\mathscr{P}}_-^{-1/2}\mathcal{H}\oplus\mathcal{H}\right).
\end{align*}
In the sequel, we will also use the spaces $\mathscr{E}_\pm$ defined as above but with the operators $\langle\mathscr{P}_\pm\rangle$ instead of $\mathscr{P}_\pm$. Finally, we define the weight $w(r):=\sqrt{(r-r_-)(r_+-r)}$.
\subsection{Abstract setting}
\label{Verification of the abstract assumptions (G)}
Meromorphic extensions in our setting follow from the works of Mazzeo-Melrose \cite{MaMe87} and Guillarmou \cite{Gu04}, as stated in \cite[Prop. 5.3]{GeoGerHA17}. The abstract setting in which this result can be used is recalled in this paragraph.

We first recall for the reader convenience the Abstract assumptions (A1)-(A3), the Meromorphic Extensions assumptions (ME1)-(ME2) as well as the "Two Ends" assumptions (TE1)-(TE3) of \cite{GeoGerHA17}:
\begin{align*}
	\label{A1}&\mathscr{P}>0,\tag{A1}\\
	\label{A2}&\!\!\begin{cases}
	sV\in\mathcal{B}(\mathscr{P}^{-1/2} L^2)>0,\\
	\text{if $z\neq\mathbb{R}$ then $(z-sV)^{-1}\in\mathcal{B}(\mathscr{P}^{-1/2} L^2)$ and there exists $n>0$}\\[-1mm]
	\text{such that $\|(z-sV)^{-1}\|_{_{\mathcal{B}(\mathscr{P}^{-1/2} L^2)}}\lesssim|\Im z|^{-n}$},\\
	\text{there exists $c>0$ such that $\|(z-sV)^{-1}\|_{_{\mathcal{B}(\mathscr{P}^{-1/2} L^2)}}\lesssim\big||z|-\|sV\|_{_{L^\infty}}\big|$}\\[-1mm]
	\text{if $|z|\geq c\|sV\|_{_{L^\infty}}$}
	\end{cases},\tag{A2}\\
	\label{ME1}&\!\!\begin{cases}
	\text{(a) }wVw\in L^\infty,\\
	\text{(b) }[V,w]=0\\
	\text{(c) }\mathscr{P}^{-1/2}[\mathscr{P},w^{-\epsilon}]w^{\epsilon/2}\in\mathcal{B}(L^2)\text{ for all $0<\epsilon\leq 1$},\\
	\text{(d) }\text{if $\epsilon>0$ then $\|w^{-\epsilon}u\|_{_{L^2}}\lesssim\|\mathscr{P}^{1/2}u\|_{_{L^2}}$ for all $u\in\mathscr{P}^{-1/2}L^2$},\\
	\text{(e) } w^{-1}\langle\mathscr{P}\rangle^{-1}\in\mathcal{B}(L^2)\text{ is compact}
	\end{cases},\tag{ME1}\\
	\label{ME2}&\!\!\begin{cases}
	\text{For all $\epsilon>0$ there exists $\delta_\epsilon>0$ such that $w^{-\epsilon}(\mathscr{P}-z^2)^{-1}w^{-\epsilon}$ extends from $\mathbb{C}^+$ }\\[-1mm]
	\text{to $\{z\in\mathbb{C}\mid\Im z>-\delta_\epsilon\}$ as a finite meromorphic function with values}\\[-1mm]
	\text{in compact operators acting on $L^2$}
	\end{cases},\tag{ME2}\\
	\label{TE1}&\!\!\begin{cases}
	[x,sV]=0,\\
	x\mapsto w(x)\in\mathcal{C}^\infty(\mathbb{R},\mathbb{R}),\\
	\chi_1(x)\mathscr{P}\chi_2(x)=0\text{ for all $\chi_1,\chi_2\in\mathcal{C}^\infty(\mathbb{R},\mathbb{R})$ bounded with all their derivatives}\\[-1mm]
	\text{and such that $\mathrm{Supp\,}\chi_1\cap\mathrm{Supp\,}\chi_2=\emptyset$}
	\end{cases},\tag{TE1}\\
	\label{TE2}&\text{There exists $\ell_-\in\mathbb{R}$ such that $(\mathscr{P}_+,k_+)$ and $(\tilde{\mathscr{P}_-},(k_--\ell_-))$ satisfy \eqref{A2}},\tag{TE2}\\
\end{align*}
\begin{align*}
	\label{TE3}&\!\!\begin{cases}
	\text{(a) }wi_+sVi_+w,wi_-(sV-\ell_-)i_-w\in L^\infty,\\
	\text{(b) }[\mathscr{P}-s^2V^2,i_\pm]=\tilde{i}[\mathscr{P}-s^2V^2,i_\pm]\tilde{i}\text{ for some $\tilde{i}\in\mathcal{C}^\infty_{\mathrm{c}}(\left]-2,2\right[,\mathbb{R})$ such that $\restriction{\tilde{i}}{[-1,1]}\equiv1$}\\
	\text{(c) }\text{$(\mathscr{P}_+,k_+,w)$ and $(\mathscr{P}_-,(k_--\ell_-),w)$ fulfill \eqref{ME1} and \eqref{ME2}},\\
	\text{(d) }\mathscr{P}_\pm^{1/2}i_\pm\mathscr{P}_\pm^{-1/2}, \mathscr{P}^{1/2}i_\pm\mathscr{P}^{-1/2}\in\mathcal{B}(L^2),\\
	\text{(e) }w[(\mathscr{P}-s^2V^2),i_\pm]w\mathscr{P}_\pm^{-1/2}, w[(\mathscr{P}-s^2V^2),i_\pm]w\mathscr{P}^{-1/2}, [(\mathscr{P}-s^2V^2),i_\pm]\mathscr{P}_\pm^{-1/2},\\[-0.75mm]
	\textcolor{white}{(e) \ \ \!}[(\mathscr{P}-s^2V^2),i_\pm]\mathscr{P}^{-1/2},\mathscr{P}^{-1/2}[w^{-1},\mathscr{P}]w\text{ are bounded operators on $L^2$},\\
	\text{(e) }\text{if $\epsilon>0$ then $\|w^{-\epsilon}u\|_{_{L^2}}\lesssim\|\mathscr{P}^{1/2}u\|_{_{L^2}}$ for all $u\in\mathscr{P}^{-1/2}L^2$}
	\end{cases}.\tag{TE3}
\end{align*}

\cite[§9]{GeoGerHA17} shows that all the above hypotheses actually follow from some geometric assumptions (the assumptions (G1)-(G7) in \cite[§2.1.1]{GeoGerHA17}). We show here that the charged Klein-Gordon equation in the exterior De Sitter-Reissner-Nordström spacetime can be dealt within this geometric setting:
\begin{itemize}
	\item [(G1)] The operator $P$ in \cite{GeoGerHA17} is $-\Delta_{\mathbb{S}^{2}}$ for us, and satisfies of course $[\Delta_{\mathbb{S}^{2}},\partial_{\phi}]=0$.%
	\item [(G2)] The operator $h_{0,s}$ in \cite{GeoGerHA17} is $\mathscr{P}$ for us, that is $\alpha_1(r)=\alpha_3(r)=r^{-1}F(r)^{1/2}$, $\alpha_2(r)=rF(r)^{1/2}$ and $\alpha_4(r)=mF(r)^{1/2}$. These last coefficients are clearly smooth in $r$. Furthermore, since we can write $F(r)=g(r)w(r)^2$ with $g(r)=\frac{\Lambda}{3r^2}(r-r_n)(r-r_c)\gtrsim1$ for all $r\in\left]r_-,r_+\right[$, it comes for all $j\in\{1,2,3,4\}$ as $r\to r_\pm$
	\begin{align*}
		&\alpha_j(r)-w(r)\left(i_-(r)\,\alpha_j^-+i_+(r)\,\alpha_j^+\right)=w(r)\left(g(r)^{1/2}-\alpha_j^{\pm}\right)=\mathcal{O}_{r\to r_\pm}\big(w(r)^{2}\big),\\
		&\alpha_1^\pm=\alpha_3^\pm=\frac{1}{r_\pm^2}\sqrt{\frac{\Lambda(r_\pm-r_n)(r_\pm-r_c)}{3}},\\
		&\alpha_2^\pm=\sqrt{\frac{\Lambda(r_\pm-r_n)(r_\pm-r_c)}{3}},\\
		&\alpha_4^\pm=\frac{m}{r_\pm}\sqrt{\frac{\Lambda(r_\pm-r_n)(r_\pm-r_c)}{3}}.
	\end{align*}
	Also, we clearly have $\alpha_j(r)\gtrsim w(r)$. Direct computations show that
	\begin{align*}
		\partial_r^m\partial_{\omega}^n\left(\alpha_j-w\left(i_-\,\alpha_j^-+i_+\,\alpha_j^+\right)\right)\!(r)=\mathcal{O}_{r\to r_\pm}\big(w(r)^{2-2m}\big)
	\end{align*}
	for all $m,n\in\mathbb{N}$.
	\item [(G3)] The operator $k_s$ in \cite{GeoGerHA17} is $sV(r)$ for us, so $k_s=k_{s,v}$ and $k_{s,r}=0$. We have $V(r)-V_{\pm}=\mathcal{O}_{r\to r_\pm}(|r_\pm-r|)=\mathcal{O}_{r\to r_\pm}\big(w(r)^2\big)$ (with $V_+=0$, recall the discussion at the beginning of Subsection \ref{Notations}) and $\partial_r^m\partial_{\omega}^nV(r)$ is bounded for any $m,n\in\mathbb{N}$.
	\item [(G4)] The perturbation $k$ in \cite{GeoGerHA17} is simply $k=k_s=sV$ for us, so that this assumption is trivially verified.
	\item [(G5)] The operator $h_0$ in \cite{GeoGerHA17} is simply $h_0=h_{0,s}=\mathscr{P}$ for us, and we have
	\begin{align*}
		\mathscr{P}&=-\alpha_1(r)\partial_r w(r)^2r^2g(r)\partial_r\alpha_1(r)-\alpha_1(r)^2\Delta_{\mathbb{S}^{2}}+\alpha_1(r)^2m^2r^2\\
		&=\alpha_1(r)\left(-\partial_r w(r)^2r^2g(r)\partial_r-\Delta_{\mathbb{S}^{2}}+m^2r^2\right)\alpha_1(r)\\
		&\gtrsim\alpha_1(r)\left(-\partial_r w(r)^2\partial_r-\Delta_{\mathbb{S}^{2}}+1\right)\alpha_1(r).
	\end{align*}
	\item [(G6)]This assumption is trivial in our setting.
	\item [(G7)] We check that $(\mathscr{P}_+,k_+)$ and $(\tilde{\mathscr{P}}_-,k_--sV_{-})$ satisfy (G5). Since $\alpha_1(r),k_+(r)=\mathcal{O}_{r\to r_\pm}(|r_\pm-r|)$, we can write for $|s|<mr_-$
	\begin{align*}
		\mathscr{P}_+&=-\alpha_1(r)\partial_r w(r)^2r^2g(r)\partial_r\alpha_1(r)-\alpha_1(r)^2\Delta_{\mathbb{S}^{2}}+\alpha_1(r)^2m^2r^2-k_+(r)^2\\
		&=\alpha_1(r)\left(-\partial_r w(r)^2r^2g(r)\partial_r-\Delta_{\mathbb{S}^{2}}+m^2r^2-\frac{k_+(r)^2}{\alpha_1(r)^2}\right)\alpha_1(r)\\
		&\gtrsim\alpha_1(r)\left(-\partial_r w(r)^2\partial_r-\Delta_{\mathbb{S}^{2}}+1\right)\alpha_1(r).
	\end{align*}
	As $k_-(r)-sV_{-}=\mathcal{O}_{r\to r_\pm}(|r_\pm-r|)$ too, we get the same conclusion with $\tilde{\mathscr{P}}_-$.
\end{itemize}

To end this Subsection, we recall from \cite[§9]{GeoGerHA17} that
\begin{align*}
	\text{(G3)}\implies\text{(A1)-(A3)},\qquad\text{(G3)}\implies\text{(ME1)},\qquad\text{(G3)-(G5)}\implies\text{(TE1)-(TE3)}
\end{align*}
and (ME2) is satisfied by assumptions (G1), (G2) and (G7) on the form of the operator $\mathscr{P}$ using Mazzeo-Melrose standard result (see \cite[§9.2.2]{GeoGerHA17} and also \cite{MaMe87} for the original work of Mazzeo-Melrose).
\subsection{Study of the asymptotic Hamiltonians}
\label{Study of the asymptotic Hamiltonians}
The aim of this paragraph is to show the existence of a meromorphic continuation of the weighted resolvent $w^{\delta}(\hat{H}_{\pm}(s)-z)^{-1}w^{\delta}$ from $\mathbb{C}^+$ into a strip in $\mathbb{C}^-$ which is analytic in $z$ in a tight box near 0. We start with the meromorphic extension.
\begin{lemma}
	\label{Meromorphic extension of asymptotic Hamiltonians}
	For all $\delta>\delta'>0$ and all $s\in\mathbb{R}$, $w^{\delta}(\hat{H}_{\pm}(s)-z)^{-1}w^{\delta}$ has a meromorphic extension from $\mathbb{C}^+$ to $\{\omega\in\mathbb{C}\mid\Im\omega>-\delta'\}$ with values in compact operators acting on $\dot{\mathscr{E}}_\pm$.
\end{lemma}
\begin{proof}
	Since hypotheses (G) are satisfied, we can apply \cite[Lem. 9.3]{GeoGerHA17} which shows that we can apply Mazzeo-Melrose result: the meromorphic extension of $w^{\delta}(\mathscr{P}_{\pm}-z^2)^{-1}w^{\delta}$ exists from $\mathbb{C}^+$ to a strip $\mathcal{O}_\delta$. This strip is explicitly given in the work of Guillarmou (\textit{cf.} \cite[Thm. 1.1]{Gu04}):
	\begin{align*}
		\mathcal{O}_\delta&=\left\{z\in\mathbb{C}\ \bigg\vert\ z^2=\lambda(3-\lambda),\ \Re\lambda>\frac{3}{2}-\delta\right\}.
	\end{align*}
	The absence of essential singularity is due to the fact that the metric $g$ is even (\textit{cf.} \cite[Def. 1.2 \& Thm. 1.4]{Gu04}). We have to check that the set $\mathcal{O}_\delta$ contains a strip in $\overline{\mathbb{C}^-}$. To see this, write $\lambda=\alpha+\mathrm{i}\beta$ and $z=a+\mathrm{i}b$ with $\alpha,\beta,a,b\in\mathbb{R}$, $b\leq0$ and $z^2=\lambda(3-\lambda)$. Solving for
	\begin{align}
	\label{Equation polynomial}
		\begin{cases}
			a^2-b^2~=~\alpha(3-\alpha)+\beta^2\\
			2ab~=~(3-2\alpha)\beta
		\end{cases}
	\end{align}
	we find
	\begin{align*}
		\begin{cases}
			\beta~=~\pm\sqrt{\frac{1}{2}(a^2-b^2-9/4)+\frac{1}{2}\sqrt{(a^2-b^2-9/4)^2+4a^2b^2}}\\
			\alpha~=~\frac{3}{2}-\frac{ab}{\beta}
		\end{cases}
	\end{align*}
and these expressions make sense since $\beta=0$ can happen only if $ab=0$, and
\begin{align*}
	\beta&=\pm\frac{|a||b|}{\sqrt{b^2+9/4}}+\mathcal{O}_{a\to0}(a),&\beta&=\mathcal{O}_{b\to0}(b).
\end{align*}
If $b=0$ then $\alpha=3/2$ and $\beta$ solves $a^2=9/4+\beta^2$, and conversely $\alpha=3/2$ implies $b=0$. Hence $\alpha=3/2$ allows all $z\in\mathbb{R}$. We may now assume $b<0$ (hence $\alpha\neq0$). The condition $\Re\lambda=\alpha>3/2-\delta$ reads $\frac{ab}{\beta}<\delta$, and this condition is trivially satisfied if $\alpha\geq3/2$ since \eqref{Equation polynomial} implies that $\frac{ab}{\beta}\leq0<\delta$. Otherwise, if $\alpha<3/2$ then \eqref{Equation polynomial} implies that $\frac{ab}{\beta}>0$ and
\begin{align*}
	b>-\left|\frac{\beta}{a}\right|\delta.
\end{align*}
We compute
\begin{align*}
	\left(\frac{\beta}{a}\right)'&=\frac{a\beta'-\beta}{a^2}
\end{align*}
where $'$ denotes here the derivative with respect to $a$, and
\begin{align*}
	\beta'&=\frac{a}{2\beta}\left(1+\frac{(a^2-b^2-9/4)+2b^2}{\sqrt{(a^2-b^2-9/4)^2+4a^2b^2}}\right)
\end{align*}
so that
\begin{align*}
	a\beta'-\beta=0&\iff a^2\left(1+\frac{(a^2-b^2-9/4)+2b^2}{\sqrt{(a^2-b^2-9/4)^2+4a^2b^2}}\right)=2\beta^2\\
	&\iff a^2(a^2-b^2-9/4)+2a^2b^2\\
	&\qquad\quad\!\!=-(b^2+9/4)\sqrt{(a^2-b^2-9/4)^2+4a^2b^2}+(a^2-b^2-9/4)^2+4a^2b^2\\
	&\iff(b^2+9/4)^2((a^2-b^2-9/4)^2+4a^2b^2)=(b^4+81/16+a^2b^2-9a^2/4+9b^2/2)^2.
\end{align*}
After some tedious simplifications, we obtain the very simple condition
\begin{align*}
	a\beta'-\beta=0&\iff 9a^4b^2=0.
\end{align*}
Thus $a=0$ is the only possible extremum of $\beta$ when $b<0$. One can check that $\beta\to1$ as $a\to\pm\infty$, whence
\begin{align*}
	\big\{z\in\mathbb{C}\ \big\vert\ 0\geq\Im z>-\delta\big\}&\subset\mathcal{O}_\delta.
\end{align*}

From there, we deduce the existence of the meromorphic extension of $w^{\delta}(\hat{H}_{\pm}(s)-z)^{-1}w^{\delta}$ for $z\in\{\omega\in\mathbb{C}\mid\Im\omega>-\delta'\}$ thanks to \cite[Lem. 4.3 \& Prop. 4.4]{GeoGerHA17} (the parameters $\epsilon$ and $\delta_\epsilon$ therein are identical in our situation, and $\delta_{\epsilon/2}$ can be replaced by any $\delta'<\delta_\epsilon$).
	%
	%
	%
	%
	%
	%
	%
	%
	%
	%
\end{proof}
%
%
%
Before proving the analyticity near 0 of the weighted resolvent, we need to prove the following result:
\begin{lemma}
\label{Resonance 0 for P-z^2}
	 For all $\delta>0$, $w^\delta(\mathscr{P}-z^2)^{-1}w^\delta$ has no pole in $\mathbb{R}$.
\end{lemma}
\begin{proof}
	We can work with the operator $P$ expressed in the Regge-Wheeler coordinate since $P\mapsto\mathscr{P}$ is an unitary transform (as explained at the beginning of Subsection \ref{Notations}).
	
	For all $\ell\in\mathbb{N}$, $P_{\ell}$ is selfadjoint and the potential $W_0+\ell(\ell+1)W_1$ (with $W_0$ and $W_1$ as in \eqref{Potentiels W0 et W1}) is bounded on $\mathscr{D}(P_{\ell})$ and tends to 0 at infinity exponentially fast; as a result, the Kato-Agmon-Simon theorem (\textit{cf.} \cite[Thm. XIII.57]{ReSi4}) implies that $P_{\ell}$ has no positive eigenvalue. As $P_{\ell}\geq 0$, we deduce that there is no eigenvalue on $\mathbb{R}\setminus\{0\}$. Furthermore, \cite[Prop. II.1]{BaMo93} shows that $0$ is not an eigenvalue for $P_{\ell}$ thanks to the exponential decay of $W_0+\ell(\ell+1)W_1$. Finally, $P_{\ell}$ verifies the limiting absorption principle
	\begin{align*}
		\sup_{\mu>0}\big\|\langle x\rangle^{-\alpha}(P_{\ell}-(\lambda+\mathrm{i}\mu))^{-1}\langle x\rangle^{-\alpha}\big\|_{_{L^{2}}}&<+\infty\qquad\qquad\forall\lambda\in\mathbb{R}\setminus\{0\},\,\forall\alpha>1,
	\end{align*}
	see Mourre \cite{Mou80}. The only issue then is $z=0$ which could be a pole.
	
	We introduce then the Jost solutions following \cite{Ba} (as $s=0$ in this Lemma, the potential $sV$ vanishes). Fix $\ell\in\mathbb{N}$ and set $\tilde{W}_{\ell}:=\ell(\ell+1)W_0+W_{1}$. Set
	\begin{align}
	\label{kappa}
		\kappa&:=\min\{\kappa_{-},|\kappa_{+}|\}
	\end{align}
	where $\kappa_{\pm}$ are the surface gravity at the event and cosmological horizons (cf. \eqref{Surface gravities}). For any $\alpha\in\left]0,2\kappa\right[$,
	\begin{align*}
		\int_{-\infty}^{+\infty}\big|\tilde{W}_{\ell}(x)\big|\mathrm{e}^{\alpha|x|}\mathrm{d}x<+\infty.
	\end{align*}
	The convergence of the above integral comes from the exponential decay of $\tilde{W}_{\ell}$ at infinity. For all $z\in\mathbb{C}$ such that $\Im z>-\kappa$, \cite[Prop. 2.1]{Ba} shows that there exist two unique $\mathcal{C}^{2}$ functions $x\mapsto e^{\pm}\left(x,z,\ell\right)$, that we will simply write $e_\pm(x)$ or $e_\pm(x,z)$, satisfying the Schrödinger equation
	\begin{align*}
		(\partial_{x}^{2}+z^{2}-\tilde{W}_\ell(x))e_{\pm}(x)&=0&\forall x\in\mathbb{R}
	\end{align*}
	with $\partial_{x}e_{\pm}\in L^{\infty}_{\mathrm{\ell\text{oc}}}(\mathbb{R}_x,\mathbb{C})$, and such that if $\Im z>-\kappa$, then
	$\partial^{j}_{x}e_{\pm}$ is analytic in $z$ for all $0\leq j\leq1$. Moreover, they satisfy
	\begin{align}
	\label{Jost solutions remainders}
		\lim_{x\to\pm\infty}\Big(\big|e_{\pm}(x)-\mathrm{e}^{\pm\mathrm{i}zx}\big|+\big|\partial_xe_{\pm}(x)\mp\mathrm{i}z\mathrm{e}^{\pm\mathrm{i}zx}\big|\Big)&=0.
	\end{align}
	By checking the formula on $\mathcal{C}_{\mathrm{c}}^2(\mathbb{R},\mathbb{C})$ first and then extending it on $H^2(\mathbb{R},\mathrm{d}x)$ by density, one easily shows that the kernel $K$ of $(P_\ell-(z-sV)^{2})^{-1}:H^{2}_{\mathrm{c}}(\mathbb{R},\mathrm{d}x)\to L^{2}_{\ell\mathrm{oc}}(\mathbb{R},\mathrm{d}x)$ for $\Im z>-\kappa$ is given by\footnote{Observe that in the example at the beginning of Section \ref{Meromorphic extension and resonances}, the kernel $R(z;x,y)$ is also given in terms of the Jost functions $e_{\pm}(x,z)=\mathrm{e}^{\pm\mathrm{i}z x}$.}
	\begin{align*}
		K(z;x,y)&=\frac{1}{\mathscr{W}(z)}\big(e_+(x,z)e_-(y,z)\mathds{1}_{x\geq y}(x,y)+e_+(y,z)e_-(x,z)\mathds{1}_{y\geq x}(x,y)\big)
	\end{align*}
	where $\mathscr{W}(z)=e_+(x)(e_-)'(x)-(e_+)'(x)e_-(x)$ is the Wronskian between $e_+$ and $e_-$. Since $\mathscr{W}$ is independent of $x\in\mathbb{R}$ (as shown at the very beginning of the proof of \cite[Prop. 2.1]{Ba}), we see using the non trivial limits for $e_{\pm}$ in \eqref{Jost solutions remainders} that a pole $z$ of order $n>0$ for $w^{\delta}(P_{\ell}-(z-sV)^2)^{-1}w^{\delta}$ with $\Im z>-\kappa$ is a zero of order $n$ of the Wronskian $\mathscr{W}$, and $e_+(\cdot,z)$ and $e_-(\cdot,z)$ are then collinear.
	
	We now reproduce the computation (2.14) in \cite{BoHa08}. Assume that $z=0$ is a pole; for all $\ell\in\mathbb{N}$, $e_+(\cdot,0,\ell)\in L^2_{\ell\mathrm{oc}}(\mathbb{R},\mathrm{d}x)$ and satisfies $P_{\ell}e_{+}=0$, so that
	\begin{align*}
		0&=\int_{-R_0}^{R_0}\left(P_{\ell}e_{+}\right)\overline{e_{+}}\mathrm{d}x\\
		&=\left[r\overline{e_{+}}\partial_{x}\left(r^{-1}e_{+}\right)\right]_{-R_0}^{R_0}+\int_{-R_0}^{R_0}\left|r\partial_{x}\left(r^{-1}e_{+}\right)\right|^{2}\mathrm{d}x+\ell\left(\ell+1\right)\int_{-R_0}^{R_0}F\left(r\right)\left|r^{-1}e_{+}\right|^{2}\mathrm{d}x\\
		&+m^{2}\int_{-R_0}^{R_0}F\left(r\right)\left|e_{+}\right|^{2}\mathrm{d}x.
	\end{align*}
	Letting $R_0\to+\infty$ and using the decay of the derivative of $e_+$ in \eqref{Jost solutions remainders} for $z=0$ show that $e_{+}=0$, a contradiction\footnote{We may notice here that the positive mass term $m^{2}$ allowed us to conclude that $z=0$ is not a pole. For the wave equation as in \cite{BoHa08}, we do not have any positivity for $\ell=0$ and $z=0$ is shown to be a pole.}.
\end{proof}
We are now ready to prove the analyticity. In the next result, we will consider $s$ as a complex number.
\begin{proposition}
	\label{Extension of asymptotic Hamiltonians}
	Let $0<\delta<\kappa$ and $R>0$. There exist $\varepsilon_0\equiv\varepsilon_0(\delta)>0$ and $\sigma\equiv\sigma(\mathscr{P},\tilde{k}_\pm)$ such that the extension of $w^{\delta}(\hat{H}_{\pm}(s)-z)^{-1}w^{\delta}$ is holomorphic in $(s,z)$ for $s\in D(0,\sigma)$ and $z\in\left]-R,R\right[+\mathrm{i}\left]-\varepsilon_0,\varepsilon_0\right[$.
\end{proposition}
The restriction $\delta<\kappa$ comes from the fact that the extension of $(\hat{H}_{\pm}(s)-z)^{-1}$ depends itself on $\delta$ (see formula \eqref{Restriction on delta} in the proof). One can increase the exponent of the weight $w$ but the width of the strip in $\mathbb{C}^{-}$ in which the result of Proposition \ref{Extension of asymptotic Hamiltonians} holds is bounded by $\kappa$.
\begin{proof}
	Observe that $(\hat{H}_-(s)-z)^{-1}=(\Phi(-sV_{-})\hat{H}_-(s)\Phi(sV_{-})-z)^{-1}$ on $\dot{\mathscr{E}}_-$, so it is sufficient to prove the announced results for $w^\delta(\hat{H}_+(s)-z)^{-1}w^\delta$ and $w^\delta(\widetilde{\hat{H}}_-(s)-z)^{-1}w^\delta$ with
	\begin{align*}
	\widetilde{\hat{H}}_-(s)&:=\Phi(-sV_{-})\hat{H}_-(s)\Phi(sV_{-})=\begin{pmatrix}
	sV_{-}&\mathrm{Id}\\
	\tilde{\mathscr{P}}_-&2k_--sV_{-}
	\end{pmatrix}.
	\end{align*}
	Proceeding as in the proof of \cite[Prop. 4.4]{GeoGerHA17}, we can work on the operators $\mathscr{P}-(z-\tilde{k}_\pm)^2$ with
	\begin{align*}
		\tilde{k}_+&:=k_+,&\tilde{k}_-&:=k_--sV_{-}
	\end{align*}
	so that $\tilde{k}_\pm$ are now exponentially decaying potentials at infinity (and are polynomial in $s$).
	
	We reproduce the perturbation argument of \cite[Lem. 4.3]{GeoGerHA17}. Choose $\varepsilon_0\in\left]0,\varepsilon_\delta\right[$, $\varepsilon_\delta$ as in Lemma \ref{Meromorphic extension of asymptotic Hamiltonians}, and pick $z\in\left]-R,R\right[+\mathrm{i}\left]-\varepsilon_0,\varepsilon_0\right[$ so that $w^{\delta}(\mathscr{P}-z^2)^{-1}w^{\delta}$ is holomorphic (it is possible since there is no pole in $\mathbb{R}$ by Lemma \ref{Resonance 0 for P-z^2}) and $w^{\delta}(\mathscr{P}-(z-\tilde{k}_\pm)^2)^{-1}w^{\delta}$ is meromorphic\footnote{Lemma \ref{Meromorphic extension of asymptotic Hamiltonians} of course applies if we replace $k_-$ by $\tilde{k}_-$.} in $z$. Write then in $\mathcal{H}$
	\begin{align}
	\label{Restriction on delta}
		w^{\delta}(\mathscr{P}-z^2)^{-1}w^{\delta}&=w^{\delta}(\mathscr{P}-(z-\tilde{k}_\pm)^2)^{-1}w^{\delta}\big(\mathrm{Id}+\mathcal{K}_\pm(s,z)\big)
	\end{align}
	with
	\begin{align*}
		\mathcal{K}_\pm(s,z)&:=w^{-\delta}\tilde{k}_\pm(2z-\tilde{k}_\pm)w^{-\delta}w^{\delta}(\mathscr{P}-z^2)^{-1}w^{\delta}.
	\end{align*}
	$\mathcal{K}_\pm(s,z)$ is clearly analytic in $s\in D(0,1)$ and in $z\in\left]-R,R\right[+\mathrm{i}\left]-\varepsilon_0,\varepsilon_0\right[$. %
	Since $\tilde{k}_\pm=\mathcal{O}_{r\to r_\pm}\big(w^{2\kappa}\big)$ by \eqref{Exp convergence of r} and \eqref{kappa}, $\delta<\kappa$ and $w^{\delta}(\mathscr{P}-z^2)^{-1}w^{\delta}$ is compact by Lemma \ref{Meromorphic extension of asymptotic Hamiltonians}, we see that $\mathcal{K}_\pm(s,z)$ is compact. %
	By two-dimensional analytic Fredholm theory, there exists a subvariety $S\subset D(0,1)\times\big(\left]-R,R\right[+\mathrm{i}\left]-\varepsilon_0,\varepsilon_0\right[\big)$ such that $(\mathrm{Id}+\mathcal{K}_\pm(s,z))^{-1}$ exists and is holomorphic in $(s,z)\in\Big(D(0,1)\times\big(\left]-R,R\right[+\mathrm{i}\left]-\varepsilon_0,\varepsilon_0\right[\big)\Big)\setminus S$. %
	We then get the representation formula for the extension:
	\begin{align}
	\label{Fredholm 1}
		w^{\delta}(\mathscr{P}-(z-\tilde{k}_\pm)^2)^{-1}w^{\delta}&=w^{\delta}(\mathscr{P}-z^2)^{-1}w^{\delta}\big(\mathrm{Id}+\mathcal{K}_\pm(s,z)\big)^{-1}.
	\end{align}
	We claim that for $\sigma>0$ sufficiently small, we have
	\begin{align}
	\label{Subvaritey 1}
	  \Big(D(0,\sigma)\times\big(\left]-R,R\right[+\mathrm{i}\left]-\varepsilon_0,\varepsilon_0\right[\big)\Big)\cap S&=\emptyset.
	\end{align}
	Otherwise, for every $n\in\mathbb{N}\setminus\{0\}$, there is a couple $(s_n,z_n)\in D(0,1/n)\times\big(\left]-R,R\right[+\mathrm{i}\left]-\varepsilon_0,\varepsilon_0\right[\big)$ such that $\mathrm{Id}+\mathcal{K}(s_n,z_n)$ is not invertible. By compactness, we can assume that $(s_n,z_n)\to(0,z_0)$ as $n\to+\infty$ for some $z_0\in[-R,R]+\mathrm{i}[-\varepsilon_0,\varepsilon_0]$. But $\mathrm{Id}+\mathcal{K}(0,z_0)=\mathrm{Id}$ is invertible for all $z\in\mathbb{C}$, so $\mathrm{Id}+\mathcal{K}_\pm(s,z)$ must be invertible too for all $(s,z)$ in a small neighbourhood of $(0,z_0)$, a contradiction.
	
	Assuming now that $s$ is sufficiently small so that \eqref{Subvaritey 1} is true, we deduce by \eqref{Fredholm 1} that the poles of $w^{\delta}(\mathscr{P}-(z-\tilde{k}_\pm)^2)^{-1}w^{\delta}$ are exactly the poles of $w^{\delta}(\mathscr{P}-z^2)^{-1}w^{\delta}$. Since for $z\in\left]-R,R\right[+\mathrm{i}\left]-\varepsilon_0,\varepsilon_0\right[$, $w^{\delta}(\mathscr{P}-z^2)^{-1}w^{\delta}$ has no pole, the same conclusion applies for $w^{\delta}(\mathscr{P}-(z-\tilde{k}_\pm)^2)^{-1}w^{\delta}$. 
\end{proof}
\subsection{Construction of the meromorphic extension of the weighted resolvent}
\label{Meromorphic extension}
The aim of this paragraph is to show the existence of a meromorphic extension for $w^\delta(\hat{K}(s)-z)^{-1}w^\delta$ in a strip near $0$ of width uniform in $s$. Since the operators $\hat{K}(s)$ and $\hat{H}(s)$ are equivalent modulo the isomorphism $\Phi(sV)$ (by (3.19) in \cite{GeoGerHA17}), we will work with the latter one. 

We first need some preliminary results. The starting point is the following result:
\begin{proposition}[Prop. 5.5 in \cite{GeoGerHA17}]
	\label{Spectrum of the Hamiltonian H}
	There is a finite set $Z\in\mathbb{C}\setminus\mathbb{R}$ with $\bar{Z}=Z$ such that the spectrum of $\hat{H}(s)$ is included in $\mathbb{R}\cup Z$ and the resolvent has a meromorphic extension to $\mathbb{C}\setminus\mathbb{R}$. Moreover, the set $Z$ consists of eigenvalues of finite multiplicity of $\hat{H}(s)$.
\end{proposition}
An important fact is that \cite[Prop. 3.6]{GeoGerHA17} shows that $Z\equiv Z(s)$ is contained in the disc $D(0,C|s|)$ for some constant $C>0$ (we can take $C=2\|V\|_{_{L^{\infty}}}$). We show below that $Z(s)\cap\mathbb{C}^+=\emptyset$ for $s$ sufficiently small.

We henceforth use the Regge-Wheeler coordinate $x$ introduced in Subsection \ref{The Regge-Wheeler coordinate}. We will still denote by $\mathscr{P}$ the operator defined in \eqref{Operator P rund} expressed in the coordinates $(x,\omega)$:
\begin{align*}
\mathscr{P}&=-F(r(x))^{-1/2}\,\partial_x^2\, F(r(x))^{-1/2}-W_0(x)\Delta_{\mathbb{S}^{2}}+W_1(x)
\end{align*}
Let also $\mathcal{H}=L^2(\mathbb{R}\times\mathbb{S}^{2}_{\omega},F(r(x))\mathrm{d}x\mathrm{d}\omega)$.
\begin{lemma}
	\label{Weights send E into itself and into E_pm}
	For all $\delta>0$, $w^{\delta}$ sends $\dot{\mathscr{E}}$ into $\dot{\mathscr{E}}_{\pm}$ and $\dot{\mathscr{E}}$ into itself.
\end{lemma}
\begin{proof}
	Let $u=(u_0,u_1)\in\dot{\mathscr{E}}$. We only show that $w^{\delta}\dot{\mathscr{E}}\subset\dot{\mathscr{E}}_{-}$, the proof of the other statements being slightly easier. We thus look for $v=(v_0,v_1)\in\tilde{\mathscr{P}}_{-}^{-1/2}\mathcal{H}\oplus\mathcal{H}$ such that $(w^{\delta}u_0,w^{\delta}u_1)=(v_0,sV_{-}v_0+v_1)$. 
	Since $w^{\delta}$ is bounded on $\mathbb{R}$, $w^{\delta}u_1\in\mathcal{H}$. Next, using the facts that $(w^{\delta})'u_0$, $V_+w^{\delta}u_0$ and $W_j^{1/2}w^{\delta}u_0$ are in $\mathcal{H}$ thanks to \eqref{ME1} (d) ($0\leq j\leq1$), we compute
	\begin{align*}
		\|\tilde{\mathscr{P}}_{-}^{1/2}w^{\delta}u_0\|^2_{_{\mathcal{H}}}	&=\big\langle\mathscr{P}w^{\delta}u_0,w^{\delta}u_0\big\rangle_{_{\mathcal{H}}}-\underbrace{\|s^2 (V_{-}-k_-)^2w^{\delta}u_0\|^2_{_{\mathcal{H}}}}_{\lesssim\,\|\mathscr{P}^{1/2}u_0\|^2_{_{\mathcal{H}}}}
	\end{align*}
	and working with the operators $\mathscr{P}_{\ell}$ defined as $P_{\ell}$ ($\ell\in\mathbb{N}$), we get
	\begin{align*}
		\big\langle\mathscr{P}_{\ell}w^{\delta}u_0,w^{\delta}u_0\big\rangle_{_{\mathcal{H}}}&=\|\partial_xw^{\delta}u_0\|^2_{_{\mathcal{H}}}+\underbrace{\big\langle(-\ell(\ell+1)W_0+W_1)w^{\delta}u_0,w^{\delta}u_0\big\rangle_{_{\mathcal{H}}}}_{\lesssim\,\|\mathscr{P}_{\ell}^{1/2}u_0\|^2_{_{\mathcal{H}}}},\\
		\|\partial_xw^{\delta}u_0\|^2_{_{\mathcal{H}}}&\lesssim\|(w^{\delta})'u_0\|^2_{_{\mathcal{H}}}+\|w^{\delta}u_0'\|^2_{_{\mathcal{H}}}\lesssim\|\mathscr{P}_{\ell}^{1/2}u_0\|^2_{_{\mathcal{H}}}.
	\end{align*}
	This proves that $w^{\delta}u_0\in \tilde{\mathscr{P}}_{-}^{-1/2}\mathcal{H}$. Hence $v_0:=w^{\delta}u_0\in \tilde{\mathscr{P}}_-^{1/2}\mathcal{H}$, and the problem boils down to show that $v_1:=w^{\delta}u_1-sV_{-}v_0=w^{\delta}u_1-sV_{-}w^{\delta}u_0$ is in $\mathcal{H}$; this is a consequence of \eqref{ME1} (d) which implies that $sV_{-}w^{\delta}u_0\in\mathcal{H}$.
\end{proof}
For all $z\in\mathbb{C}^+$ and $s\in\mathbb{R}$, we introduce the operator
\begin{align*}
	Q(s,z)&:=i_-^{2}(\hat{H}_-(s)-z)^{-1}+i_+^{2}(\hat{H}_+(s)-z)^{-1}=\sum_{\pm}i_{\pm}^{2}(\hat{H}_{\pm}(s)-z)^{-1}.
\end{align*}
In Subsection \ref{Study of the asymptotic Hamiltonians} above, we have studied the resolvents of the asymptotic Hamiltonians. In particular, we know that $Q(s,z)$ meromorphically extends into a strip in $\mathbb{C}^-$ and is analytic in a small neighbourhood of $\mathbb{R}$. We wish to show that $(\hat{H}(s)-z)^{-1}$ has the same properties. To do this, we show that $Q(s,z)$ is a parametrix for the resolvent on the energy space.

By Lemma \ref{Weights send E into itself and into E_pm}, $Q(s,z)w^{\delta}$ is well-defined in $\dot{\mathscr{E}}$. Using the potentials $k_{\pm}=s(V\mp j_{\mp}^{2}V_{-})$ introduced in Subsection \ref{Notations} as well as the relations $i_{\pm}j_{\mp}=0$, we compute:
\begin{align*}
	i_{\pm}^{2}(\hat{H}(s)-z)(\hat{H}_{\pm}(s)-z)^{-1}&=i_{\pm}^{2}\mathrm{Id}-i_{\pm}^{2}\begin{pmatrix}
	0&0\\s^{2}(V^{2}-k_{\pm}^{2})&2s(V-k_{\pm})
	\end{pmatrix}(\hat{H}_{\pm}(s)-z)^{-1}=i_{\pm}^{2}\mathrm{Id}.
\end{align*}
Since $i_+^{2}+i_{-}^2=1$, we obtain in $\dot{\mathscr{E}}$:
\begin{align}
\label{Parametrix}
	(\hat{H}(s)-z)Q(s,z)w^{\delta}&=\left(\mathrm{Id}+\sum_{\pm}[\hat{H}(s),i_{\pm}^{2}](\hat{H}_{\pm}(s)-z)^{-1}\right)w^{\delta}.
\end{align}
It follows that for all $z\notin Z$ (see Proposition \ref{Spectrum of the Hamiltonian H})
\begin{align}
\label{1+Fredholm}
  w^{\delta}Q(s,z)w^{\delta}&=w^{\delta}(\hat{H}(s)-z)^{-1}w^{\delta}\big(\mathrm{Id}+\hat{\mathcal{K}}_\pm(s,z)\big)
\end{align}
with
\begin{align*}
	\hat{\mathcal{K}}_\pm(s,z)&:=w^{-\delta}\sum_{\pm}[\hat{H}(s),i_{\pm}^{2}](\hat{H}_{\pm}(s)-z)^{-1}w^{\delta}.
\end{align*}
%
%
\begin{lemma}
	\label{Well-definedness on E}
	The operators on the left and right-hand sides of \eqref{1+Fredholm} send $\dot{\mathscr{E}}$ into itself.
\end{lemma}
\begin{proof}
	For the left-hand side of \eqref{1+Fredholm}, we successively use Lemma \ref{Weights send E into itself and into E_pm}, the facts that $(\hat{H}_{\pm}(s)-z)^{-1}$ sends $\dot{\mathscr{E}}_{\pm}$ into $\mathscr{D}(\hat{H}_{\pm}(s))\subset\dot{\mathscr{E}}_{\pm}$ and $i_{\pm}$ sends $\dot{\mathscr{E}}_{\pm}$ into $\dot{\mathscr{E}}$ (\textit{cf.} \cite[Lem. 5.4]{GeoGerHA17}), and again Lemma \ref{Weights send E into itself and into E_pm}.

	We now deal with the right-hand side of \eqref{1+Fredholm}. By Lemma \ref{Weights send E into itself and into E_pm}, we only have to show that $w^{-\delta}[\hat{H}(s),i_{\pm}](\hat{H}_{\pm}(s)-z)^{-1}$ sends $\dot{\mathscr{E}}_{\pm}$ into $\dot{\mathscr{E}}$. Let $u\in\dot{\mathscr{E}}_{\pm}$ and write $v=(v_0,v_1):(\hat{H}_{\pm}(s)-z)^{-1}u\in\dot{\mathscr{E}}_{\pm}$. We have
	\begin{align*}
	w^{-\delta}[H(s),i_{\pm}](H_{\pm}(s)-z)^{-1}u&=w^{-\delta}\begin{pmatrix}
	0&0\\\left[\mathscr{P},i_{\pm}\right]&0
	\end{pmatrix}(v_0,v_1)=\begin{pmatrix}
	0\\w^{-\delta}\left[\mathscr{P},i_{\pm}\right]v_0
	\end{pmatrix}=\begin{pmatrix}
	0\\w^{-\delta}\left[\mathscr{P},i_{\pm}\right]w^{-\delta}w^{\delta}v_0
	\end{pmatrix}.
	\end{align*}
	Since $w^{\delta}\dot{\mathscr{E}}_{\pm}\subset\dot{\mathscr{E}}_{\pm}$, we can use \eqref{TE3} (e) to conclude that the second component is in $\mathcal{H}$, whence $w^{-\delta}[\hat{H}(s),i_{\pm}](\hat{H}_{\pm}(s)-z)^{-1}\dot{\mathscr{E}}_{\pm}\subset\dot{\mathscr{E}}$ (when $\pm=-$, we use that $\mathscr{P}^{1/2}\tilde{\mathscr{P}}_-^{-1/2}$ is bounded on $\mathcal{H}$).
\end{proof}
In the following, we will consider $s$ as a complex number lying in a small neighbourhood of 0.
\begin{lemma}
	\label{Fredholm}
	Let $0<\delta<\kappa$ and $R>R_0$. $\mathrm{Id}+\hat{\mathcal{K}}_\pm(s,z)$ is a holomorphic family of Fredholm operators acting on $\dot{\mathscr{E}}$ for $(s,z)\in D(0,\sigma)\times\big(\left]-R,R\right[+\mathrm{i}\left]-\varepsilon_0,\varepsilon_0\right[\big)$, with $\sigma>0$ sufficiently small and $\varepsilon_0>0$ as in Proposition \ref{Extension of asymptotic Hamiltonians}.
\end{lemma}
\begin{proof}
	Write
	\begin{align*}
		\hat{\mathcal{K}}_\pm(s,z)&=\sum_{\pm}w^{-\delta}[\hat{H}(s),i_{\pm}^{2}]w^{-\delta}w^{\delta}(\hat{H}_{\pm}(s)-z)^{-1}w^{\delta}.
	\end{align*}
	By Lemma \ref{Meromorphic extension of asymptotic Hamiltonians}, $w^{\delta}(\hat{H}_{\pm}(s)-z)^{-1}w^{\delta}$ is compact on $\dot{\mathscr{E}}_{\pm}$ and Proposition \ref{Extension of asymptotic Hamiltonians} shows that the extension is holomorphic in $(s,z)$. Furthermore,
	\begin{align*}
		w^{-\delta}[\hat{H}(s),i_{\pm}^{2}]w^{-\delta}&=\begin{pmatrix}
		0&0\\w^{-\delta}\left[\mathscr{P},i_{\pm}^{2}\right]w^{-\delta}&0
		\end{pmatrix}
	\end{align*}
	is bounded on $\dot{\mathscr{E}}_{\pm}$ (as a consequence of \eqref{TE3} (e), see the end of the proof of Lemma \ref{Well-definedness on E}). Hence $\hat{\mathcal{K}}_\pm(s,z)$ is compact and thus $\mathrm{Id}+\hat{\mathcal{K}}_\pm(s,z)$ is Fredholm.
\end{proof}
We are now ready to construct the meromorphic extension of the weighted resolvent. For all $s_0>0$, define $R_0:=2Cs_0$. Proposition \ref{Spectrum of the Hamiltonian H} and the remark below then show that $Z(s)\subset D(0,R/2)$ for all $s\in\left]-s_0,s_0\right[$ and all $R>R_0$.
\begin{theorem}
	\label{No resonances in a strip near the real axis for s small}
	Let $0<\delta<\kappa$ and $s\in\left]-s_0,s_0\right[$.
	\begin{itemize}
		\item [1.]  For $s$ small enough, $w^{\delta}(\hat{H}(s)-z)^{-1}w^{\delta}$ has a meromorphic extension from $\mathbb{C}^+\setminus Z$ to $\big\{\omega\in\mathbb{C}\ \big\vert\ \Im\omega>-\delta'\big\}$ for all $0<\delta'<\delta$ with values in compact operators acting on $\dot{\mathscr{E}}$.
		\item [2.] For all $R>R_0$, there exists $0<s_1<s_0$ such that for all $s\in\left]-s_1,s_1\right[$, the extension of $w^{\delta}(\hat{H}(s)-z)^{-1}w^{\delta}$ is analytic in $z\in\left]-R,R\right[+\mathrm{i}\left]-\varepsilon_0,\varepsilon_0\right[$ with $\varepsilon_0>0$ as in Proposition \ref{Extension of asymptotic Hamiltonians}.
	\end{itemize}
	%
\end{theorem}
\begin{proof}
	We first show Part 1. Let $s\in\mathbb{C}$ small enough and let $z\in\left]-R,R\right[+\mathrm{i}\left]-\delta',\delta'\right[$. Since $\hat{H}_{\pm}(0)=\hat{H}(0)$, we observe that $\hat{\mathcal{K}}_\pm(0,z)=0$ and $Q(0,z)=(\hat{H}(0)-z)^{-1}$. Hence the operator $\mathrm{Id}+\hat{\mathcal{K}}_\pm(0,z)=\mathrm{Id}$ is invertible for all $z\in\mathbb{C}$. Finally, Lemma \ref{Meromorphic extension of asymptotic Hamiltonians} shows that $w^{\delta}(\hat{H}_{\pm}(s)-z)^{-1}w^{\delta}$ is meromorphic in $z$. We can therefore use the meromorphic Fredholm theory to invert $\mathrm{Id}+\hat{\mathcal{K}}_\pm(s,z)$ on $\dot{\mathscr{E}}$. Using \eqref{1+Fredholm}, we have the representation formula
	\begin{align}
	\label{Representation formula 2}
		w^{\delta}(\hat{H}(s)-z)^{-1}w^{\delta}&=w^{\delta}Q(s,z)w^{\delta}\big(\mathrm{Id}+\hat{\mathcal{K}}_\pm(s,z)\big)^{-1}
	\end{align}
	which is valid for $z\in\left]-R,R\right[+\mathrm{i}\left]-\delta',\delta'\right[$. This shows that $w^{\delta}(\hat{H}(s)-z)^{-1}w^{\delta}$ has a meromorphic extension in this strip and Part 1 is settled.

	Let us show Part 2. of the theorem. %
	We pick this time $(s,z)\in D(0,\sigma)\times\big(\left]-R,R\right[+\mathrm{i}\left]-\varepsilon_0,\varepsilon_0\right[\big)$ with $\sigma, \varepsilon_0>0$. Lemma \ref{Fredholm} shows that, if $\sigma$ is very small, $\mathrm{Id}+\hat{\mathcal{K}}_\pm(s,z)$ is a holomorphic family of Fredholm operators acting on $\dot{\mathscr{E}}$. We can thus use the two-dimensional analytic Fredholm theory which implies that there is a meromorphic extension $D(0,\sigma)\times\big(\left]-R,R\right[+\mathrm{i}\left]-\varepsilon_0,\varepsilon_0\right[\big)\ni(s,z)\mapsto\big(\mathrm{Id}+\hat{\mathcal{K}}_\pm(s,z)\big)^{-1}$, and \eqref{Representation formula 2} is valid for $(s,z)\in D(0,\sigma)\times\big(\left]-R,R\right[+\mathrm{i}\left]-\varepsilon_0,\varepsilon_0\right[\big)$ with $\sigma$ small. This shows that the poles of $w^{\delta}(\hat{H}(s)-z)^{-1}w^{\delta}$ are the poles of $(\mathrm{Id}+\hat{\mathcal{K}}_\pm(s,z))^{-1}$ and $w^{\delta}Q(s,z)w^{\delta}$, the last ones being the poles of $w^\delta(\hat{H}_{\pm}(s)-z)^{-1}w^\delta$.

	The multidimensional analytic Fredholm theory also implies that there exists a (possibly empty) subvariety $S\subset D(0,\sigma)\times\big(\left]-R,R\right[+\mathrm{i}\left]-\varepsilon_0,\varepsilon_0\right[\big)$ such that $\mathrm{Id}+\hat{\mathcal{K}}(s,z)$ is invertible for $(s,z)\notin S$. We claim that we can take $\sigma>0$ small enough so that
	\begin{align*}
	  \Big(D(0,\sigma)\times\big(\left]-R,R\right[+\mathrm{i}\left]-\varepsilon_0,\varepsilon_0\right[\big)\Big)\cap S&=\emptyset.
	\end{align*}
	Otherwise, for every $n\in\mathbb{N}\setminus\{0\}$, there is a couple $(s_n,z_n)\in D(0,1/n)\times\big(\left]-R,R\right[+\mathrm{i}\left]-\varepsilon_0,\varepsilon_0\right[\big)$ such that $\mathrm{Id}+\hat{\mathcal{K}}(s_n,z_n)$ is not invertible. By compactness, we can assume that $(s_n,z_n)\to(0,z_0)$ as $n\to+\infty$ for some $z_0\in[-R,R]+\mathrm{i}[-\varepsilon_0,\varepsilon_0]$. But $\mathrm{Id}+\hat{\mathcal{K}}(0,z_0)=\mathrm{Id}$ is invertible for all $z\in\mathbb{C}$, so $\mathrm{Id}+\hat{\mathcal{K}}_\pm(s,z)$ must be invertible too for all $(s,z)$ in a small neighbourhood of $(0,z_0)$, a contradiction.
	
	We now assume $|s|<s_1$ where $s_1$ is so small that $\mathrm{Id}+\hat{\mathcal{K}}_\pm(s,z)$ is invertible on $\dot{\mathscr{E}}$ for $z\in \left]-R,R\right[+\mathrm{i}\left]-\varepsilon_0,\varepsilon_0\right[$. Using then the formula \eqref{Representation formula 2}, we conclude that the poles of $w^{\delta}(\hat{H}(s)-z)^{-1}w^{\delta}$ are precisely the poles of $w^{\delta}Q(s,z)w^{\delta}$, which are the poles of $w^{\delta}(\hat{H}_{\pm}(s)-z)^{-1}w^{\delta}$. We then use Proposition \ref{Extension of asymptotic Hamiltonians} to conclude that there is no pole for $z\in\left]-R,R\right[+\mathrm{i}\left]-\varepsilon_0,\varepsilon_0\right[$. This completes the proof. 
\end{proof}
As a first consequence, we deduce a holomorphy result for the resolvent.
\begin{corollary}
\label{Analyticity in C+}
	Let $\varepsilon_0>0$ as in Proposition \ref{Extension of asymptotic Hamiltonians}. Then for all $s\in\mathbb{R}$ such that $|s|<s_1$ with $Cs_1<\varepsilon_0$, the resolvent $(\hat{H}(s)-z)^{-1}$ is holomorphic in $z\in\mathbb{C}^+$. Furthermore, the spectrum of $\hat{H}(s)$ is contained in $\mathbb{R}$.
\end{corollary}
\begin{proof}
	We know by Theorem \ref{No resonances in a strip near the real axis for s small} that the weighted resolvent $w^{\delta}(\hat{H}(s)-z)^{-1}w^{\delta}$ is holomorphic\linebreak
	in $z\in\left(\left]-R,R\right[+\mathrm{i}\left]-\varepsilon_0,\varepsilon_0\right[\right)\cap\mathbb{C}^+\subset D(0,R/2)\subset\mathbb{C}^+$ if we assume $s<s_1$ by Part 2. of Theorem \ref{No resonances in a strip near the real axis for s small}.\linebreak
	By Proposition \ref{Spectrum of the Hamiltonian H}, $(\hat{H}(s)-z)^{-1}$ is holomorphic in $\mathbb{C}^+\setminus Z$. Assume then that $z_0\in\mathbb{C}^+\cap Z$ is a pole of order $m_0\in\mathbb{N}$: there exist some finite rank operators $A_{1},\ldots,A_{m_0}:\mathcal{H}\to\mathcal{H}$ such that
	\begin{align*}
		(\hat{H}(s)-z)^{-1}&=\sum_{j=1}^{m_0}\frac{A_{j}}{(z-z_0)^{j}}+\text{\,holomorphic term}&\forall z&\in\mathbb{C}^+\text{ near $z_0$}.
	\end{align*}
	Since $R>R_0$, $Z\cap\mathbb{C}^+\subset D(0,R/2)\cap\mathbb{C}^+$ and then
	\begin{align*}
		w^{\delta}(\hat{H}(s)-z)^{-1}w^{\delta}&=\sum_{j=1}^{m_0}\frac{w^{\delta}A_{j}w^{\delta}}{(z-z_0)^{j}}+\text{\,holomorphic term}
	\end{align*}
	is holomorphic in $z$ near $z_j$, so that $A_{1}=\ldots=A_{m_0}=0$ and $(\hat{H}(s)-z)^{-1}$ is holomorphic in $z\in\mathbb{C}^+$. By Proposition \ref{Spectrum of the Hamiltonian H}, this implies that the spectrum of $\hat{H}(s)$ in $Z\cap\mathbb{C}^{+}$ is empty; by symmetry, we deduce that $\bar{Z}\cap\mathbb{C}^{-}=\emptyset$ too.
\end{proof}
\begin{remark}
	Theorem \ref{No resonances in a strip near the real axis for s small} and Corollary \ref{Analyticity in C+} answer Bachelot's open question of the nature of the sets $\sigma_p$ (the eigenvalues in $\mathbb{C}^{+}$) and $\sigma_{ss}$ (the real resonances, also called hyperradiant modes) defined in \cite{Ba} by equations (2.35) and (2.36), when the charge product $s$ is sufficiently small: both are empty as he conjectured at the end of his paper.
\end{remark}
We finally deduce the existence of the cut-off inverse of the quadratic pencil and define resonances.
\begin{corollary}
	\label{Resonances def}
	Let $s\in\left]-s_0,s_0\right[$ small enough. The operator $\chi p(z,s)^{-1}\chi:L^2(\mathbb{R},\mathrm{d}x)\to H^2(\mathbb{R},\mathrm{d}x)$ defines for any $\chi\in\mathcal{C}^{\infty}_{\mathrm{c}}\left(\mathbb{R},\mathbb{R}\right)$ a meromorphic function of $z\in\{\omega\in\mathbb{C}\mid\Im\omega>-\kappa\}$ and analytic if $\Im z>-\varepsilon_0$ with $\varepsilon_0>0$ given by Proposition \ref{Extension of asymptotic Hamiltonians}.\\
	If $\chi$ is not identically 0, then the poles $z$ of this extension are exactly the poles of the cut-off resolvent $\chi(\hat{H}(s)-z)^{-1}\chi$ and are independent of the choice of $\chi$. We call them \textup{resonances} of $p$ and write $z\in\mathrm{Res}(p)$. Similarly, we define $\mathrm{Res}(p_{\ell})$ as the poles of $\chi p_\ell(z,s)^{-1}\chi$ for all $\ell\in\mathbb{N}$.
\end{corollary}
\begin{proof}
	Let $R>0$ and let $z\in\mathbb{C}$ with $-R\leq\Re z\leq R$ and $\Im z >-\kappa$. The meromorphic extension $w^{\delta}(\hat{H}(s)-z)^{-1}w^{\delta}:\dot{\mathscr{E}}\to\dot{\mathscr{E}}$ (with $0<\delta<\kappa$) entails the meromorphic extension $w^{\delta}(\hat{K}(s)-z)^{-1}w^{\delta}:\dot{\mathcal{E}}\to\dot{\mathcal{E}}$ since $\hat{H}(s)$ and $\hat{K}(s)$ are equivalent on $\dot{\mathcal{E}}$ modulo the isomorphism $\Phi(sV)$ introduced in Subsection \ref{Notations}. Since $w(x)$ is exponentially decaying by \eqref{Exp convergence of r}, we can write for any cut-off $\chi\in\mathcal{C}^{\infty}_{\mathrm{c}}\left(\mathbb{R},\mathbb{R}\right)$
	\begin{align*}
	\chi(\hat{K}(s)-z)^{-1}\chi&=(\chi w^{-\delta})w^{\delta}(\hat{K}(s)-z)^{-1}w^{\delta}(w^{-\delta}\chi):L^2(\mathbb{R},\mathrm{d}x)\to H^2(\mathbb{R},\mathrm{d}x)
	\end{align*}
	using that $\chi w^{-\delta}\in\mathcal{C}^{\infty}_{\mathrm{c}}\left(\mathbb{R},\mathbb{R}\right)$. In particular, if $\chi$ is not identically 0, the poles of $\chi(\hat{K}(s)-z)^{-1}\chi$ and $w^{\delta}(\hat{K}(s)-z)^{-1}w^{\delta}$ coincide.
	
	By formula \eqref{Resolvent of H} and the discussion below, we see that we can define the operator $\chi p(z,s)^{-1}\chi:L^2(\mathbb{R},\mathrm{d}x)\to H^2(\mathbb{R},\mathrm{d}x)$ for any $\chi\in\mathcal{C}^{\infty}_{\mathrm{c}}\left(\mathbb{R},\mathbb{R}\right)$ as a meromorphic function of $z$, and its poles are precisely the poles of $\chi(\hat{K}(s)-z)^{-1}\chi$.
	
	To conclude the proof, it remains to prove the analyticity in the whole strip $\big\{z\in\mathbb{C}\,\,\big\vert\,\,\Im z>-\varepsilon_0\big\}$ (which excludes a possible accumulation of resonances to $\mathbb{R}$ at infinity): this follows from Theorem \ref{Microlocalization of high frequency resonances} below.
\end{proof}
%
%
%
%
%
%
%
%
%
\section{Resonance expansion for the charged Klein-Gordon equation}
\label{Resonance expansion for the charged Klein-Gordon equation}
We present in this section the main result of this chapter which is an extension of \cite[Theorem 1.3]{BoHa08}. By using the formula \eqref{Relation between the domains of the quadratic pencil and the corresponding Hamiltonian} and \eqref{Resolvent of H} as well as (3.21) in \cite{GeoGerHA17} and the local equivalence \eqref{Local equivalence of norms} of the norms $\|.\|_{_{\dot{\mathcal{E}}_{\ell}}}$ and $\|.\|_{_{\mathcal{E}_{\ell}}}$ if $z\in\mathbb{R}$, we can define for $\Im z>-\kappa$ the meromorphic extension of the cut-off resolvent $\hat{R}_{\chi,\ell}(z):=\chi(\hat{K}_\ell-z)^{-1}\chi$. For all resonance $z_0\in\mathrm{Res}(p_{\ell})$, denote by $m(z_0)\in\mathbb{N}$ its multiplicity and set
\begin{align*}
	\Pi^{\chi,\ell}_{j,k}:=\frac{1}{2\pi\mathrm{i}}\oint_{\partial\gamma}\frac{(-\mathrm{i})^{k}}{k!}\hat{R}_{\chi,\ell}(z)(z-z_0)^{k}\mathrm{d}z
\end{align*}
defined for all integer $k\geq-(m(z_{0})+1)$ with $\gamma$ a small positively oriented circle enclosing $z_{0}$ and no other resonance. We will denote by $\hat{R}_{\chi}(z)$ and $\Pi^{\chi}_{j,k}$ the cut-off resolvent of the full operator $\hat{K}$ and the corresponding generalized projector, respectively. Recall that $\mathrm{Res}(p)$ is introduced in Corollary \ref{Resonances def}.

We first introduce the set of pseudo-poles of $P$ whose points approximate high frequency resonances. The proof is given in Appendix \ref{Microlocalization of high frequency resonances section}.
\begin{theorem}
	\label{Microlocalization of high frequency resonances}
	There exist $K>0$ and $\theta>0$ such that, for any
	$C>0$, there exists an injective map $\tilde{b}:\Gamma\to\mathrm{Res}(p)$ with
	\begin{align*}
	\Gamma&=\frac{\sqrt{F(\mathfrak{r})}}{\mathfrak{r}}\left(\pm\mathbb{N}\setminus\{0\}\pm\frac{1}{2}\pm\frac{qQ}{\sqrt{F(\mathfrak{r})}}-\frac{\mathrm{i}}{2}\sqrt{\left|3-\frac{12M}{\mathfrak{r}}+\frac{10Q^2}{\mathfrak{r}^2}\right|}\left(\mathbb{N}+\frac{1}{2}\right)\right)
	\end{align*}
	the set of pseudo-poles, such that all the poles in
	\begin{align*}
	\Omega_{C}&=\big\{\lambda\in\mathbb{C}\,\,\big\vert\,\,|\lambda|>K,\Im\lambda>-\max\{C,\theta|\Re\lambda|\}\big\}
	\end{align*}
	are in the image of $\tilde{b}$. Furthermore, if $\mu\in\Gamma$ and $\tilde{b}(\mu)\in\Omega_{C}$, then
	\begin{align*}
	\lim_{|\mu|\to+\infty}(\tilde{b}(\mu)-\mu)&=0.
	\end{align*}
	If $\Re\mu=\frac{\sqrt{F(\mathfrak{r})}}{\mathfrak{r}}\left(\pm\ell\pm\frac{1}{2}\pm\frac{qQ}{\sqrt{F(\mathfrak{r})}}\right)$ for $\ell\in\mathbb{N}\setminus\{0\}$, then the corresponding pole	$\tilde{b}(\mu)$ has multiplicity $2\ell+1$.
\end{theorem}
We can now state our main result (the proof is given in Subsection \ref{Proof of Theorem {Theoreme principal}}):
\begin{theorem}[Decay of the local energy]
\label{Theoreme principal}
	Let $\chi\in\mathcal{C}^{\infty}_{\mathrm{c}}\left(\mathbb{R},\mathbb{R}\right)$.
	\begin{itemize}
	\item [(i)] Let $\nu>0$ such that $\nu\notin\Gamma$ ($\Gamma$ is the set of pseudo-poles as in Theorem \ref{Microlocalization of high frequency resonances}), $\nu<\kappa$ and $\mathrm{Res}(p)\cap\{\lambda\in\mathbb{C}\mid\Im\lambda=-\nu\}=\emptyset$. There exists $N>0$ such that, for all $u\in\dot{\mathcal{E}}$ with $\langle-\Delta_{\mathbb{S}^2}\rangle^{N}u\in\dot{\mathcal{E}}$ and $s$ small enough, we have
	\begin{align}
	\label{Expansion du propagateur}
		\chi\mathrm{e}^{-\mathrm{i}t\hat{K}}\chi u&=\sum_{\substack{z_j\in\mathrm{Res}(p)\\\Im z_j>-\nu}}\sum_{k=0}^{m(z_j)}\mathrm{e}^{-\mathrm{i}z_jt}t^k\Pi_{j,k}^\chi u+E(t)u
	\end{align}
	for $t>0$ sufficiently large, with
	\begin{align*}
		\|E(t)u\|_{_{\dot{\mathcal{E}}}}\lesssim\mathrm{e}^{-\nu t}\|\langle-\Delta_{\mathbb{S}^2}\rangle^{N}u\|_{_{\dot{\mathcal{E}}}}
	\end{align*}
	and the sum is absolutely convergent in the sense that
	\begin{align*}
		\sum_{\substack{z_j\in\mathrm{Res}(p)\\\Im z_j>-\nu}}\sum_{k=0}^{m(z_j)}\|\Pi_{j,k}^\chi\langle-\Delta_{\mathbb{S}^2}\rangle^{-N}\|_{_{\dot{\mathcal{E}}\to\dot{\mathcal{E}}}}<+\infty.
	\end{align*}
	\item [(ii)] There exists $\varepsilon>0$ such that, for any increasing positive function $g$ with $\lim_{x\to+\infty}g(x)=+\infty$ and $g(x)\leq x$ for $x\gg0$, for all $u\in\dot{\mathcal{E}}$ with $g(-\Delta_{\mathbb{S}^{2}})u\in\dot{\mathcal{E}}$ and $s$ small enough, we have
	\begin{align*}
		\|\chi\mathrm{e}^{-\mathrm{i}t\hat{K}}\chi u\|_{\dot{\mathcal{E}}}&\lesssim (g(\mathrm{e}^{\varepsilon t}))^{-1}\|g(-\Delta_{\mathbb{S}^{2}})u\|_{\dot{\mathcal{E}}}
	\end{align*}
	for $t>0$ sufficiently large.
	\end{itemize}
\end{theorem}
\begin{remark}
	\begin{enumerate}
	\item Formula \eqref{Expansion du propagateur} provides a physical interpretation of resonances: they are the frequencies and dumping rates of charged Klein-Gordon field in presence of the charged black hole (see Chapter 4.35 in \cite{Ch92} for a discussion on the interpretation of resonances).
	\item Part (ii) of Theorem \ref{Theoreme principal} shows that a logarithmic derivative loss in the angular direction $(\ln\langle-\Delta_{\mathbb{S}^{2}}\rangle)^{\alpha}u\in\dot{\mathcal{E}}$ with $\alpha>1$ entails the integrability of the local energy:
	\begin{align*}
		\left\|\int_{0}^{+\infty}\chi\mathrm{e}^{-\mathrm{i}t(\hat{K}-z)}\chi\,u\,\mathrm{d}t\right\|_{_{\dot{\mathcal{E}}}}&\lesssim\|(\ln\langle-\Delta_{\mathbb{S}^{2}}\rangle)^\alpha u\|_{_{\dot{\mathcal{E}}}}.
	\end{align*}
	\item In the limits $9\Lambda M^2\to 1^-$ and $Q\to 0$, the expansion in part (i) of Theorem \ref{Theoreme principal} is not empty has (infinitely many) pseudo-poles of $\Gamma$ (introduced in Theorem \ref{Microlocalization of high frequency resonances}) lie in the strip $\big\{z\in\mathbb{C}\mid\Im(z)>-\kappa\big\}$. To see this, it suffices to consider the case $Q=0$. Then $\mathfrak{r}=3M$ and
	\begin{align*}
		\min\big\{|\Im\lambda|\mid\lambda\in\Gamma\big\}=\frac{\sqrt{F(\mathfrak{r})}}{4\mathfrak{r}}\sqrt{\left|3-\frac{12M}{\mathfrak{r}}\right|}=\frac{\sqrt{1-9\Lambda M^2}}{12\sqrt{3}M}.
	\end{align*}
	We show that
	\begin{align}
	\label{eq res exp}
	\frac{\sqrt{1-9\Lambda M^2}}{12\sqrt{3}M}&<|\kappa_\pm|.
	\end{align}
	Observe that for $Q=0$, we have
	\begin{align*}
	F'(r)&=\frac{2M}{r^2}-\frac{2\Lambda r}{3}=\frac{1-F(r)}{r}-\Lambda r
	\end{align*}
	so that
	\begin{align*}
	F'(r_\pm)&=\frac{1}{r_\pm}-\Lambda r_\pm.
	\end{align*}
	Thus \eqref{eq res exp} becomes
	\begin{align}
	\label{1}
	\frac{\sqrt{1-9\Lambda M^2}}{6\sqrt{3}M}&<\frac{|1-\Lambda r_\pm^2|}{r_\pm}.
	\end{align}
	Set $\alpha:=3\sqrt{\Lambda}M<1$. The footnote page 6 in \cite{SaZw97} shows that
	\begin{align}
	\label{roots}
	r_\pm&=\frac{1}{\sqrt{\Lambda}}\,\Im\!\left(\left(\mp\sqrt{1-\alpha^2}+\mathrm{i}\alpha\right)^{1/3}\right).
	\end{align}
	As $\mp\sqrt{1-\alpha^2}+\mathrm{i}\alpha$ has modulus one, we can write $r_\pm=\frac{\sin\theta_\pm}{\sqrt{\Lambda}}$ for some $\theta_\pm\in\left]0,\pi\right[$ (the roots are positive) and thus \eqref{1} reads
	\begin{align*}
	\frac{\sqrt{1-\alpha^2}}{6\sqrt{3}M}&<\frac{\cos^2\theta\pm}{\sin\theta_\pm}\sqrt{\Lambda}.
	\end{align*}
	We eventually show that
	\begin{align*}
	\frac{\sqrt{1-\alpha^2}}{2\sqrt{3}\alpha}&<\frac{\cos^2\theta\pm}{\sin\theta_\pm}.
	\end{align*}
	When $\alpha\to 1^-$, the left-hand side above goes to 0 whereas the right-hand side remains positive: this last assertion can be checked in \eqref{roots} as then
	\begin{align*}
	\Im\!\left(\left(\mp\sqrt{1-\alpha^2}+\mathrm{i}\alpha\right)^{1/3}\right)&=\Im\!\left(\mathrm{i}^{1/3}\right)=\Im\!\left(\mathrm{e}^{\mathrm{i}\pi/6}\right)=\frac{1}{2}\neq 0.
	\end{align*}
	%
	%
	%
	%
	%
	%
	%
	%
	\end{enumerate}
\end{remark}
%
%
%
%
%
%
%
%
%
%
\section{Estimates for the cut-off inverse of the quadratic pencil}
\label{Estimates for the cut-off inverse of the quadratic pencil}
In this section, we show some estimates on the cut-off inverse of the quadratic pencil. We can work with $\ell\in\mathbb{N}$ fixed but our estimates have to be uniform in $\ell$. Since $\big(\chi p(-\bar{z}+2sV,s)\chi\big)^*=\chi p(z,s)\chi$, we can restrict ourselves to consider $z\in\mathbb{C}$ with $\Re z>-2s_0\|V\|_{_{L^\infty}}$ for some fixed $s_0>0$ such that $0<|s|<s_0$. In the following, we are simply denoting by $L^2$ the space $L^2(\mathbb{R},\mathrm{d}x)$. For some real numbers $R,C_{0},C_{1}>0$ (determined by Theorem \ref{Main results in the four zones for the harmonic quadratic pencil} below), we define the
\begin{itemize}
	\item zone I as $\left[-R,R\right]+\mathrm{i}\left[-C_{0},C_{0}\right]$,
	\item zone II as $\left[R,\ell/R\right]+\mathrm{i}\left[-C_{0},C_{0}\right]$,
	\item zone III as $\left[\ell/R,R\ell \right]+\mathrm{i}\left[-C_{0},C_{0}\right]$,
	\item zone IV as $\left(\left[R\ell ,+\infty\right[+\mathrm{i}\left]-\infty,C_{0}\right]\right)\cap\left\{\lambda\in\mathbb{C}\mid \Im\lambda\geq-C_{0}-C_{1}\ln\langle\lambda\rangle\right\}\cap\Omega_{\kappa}$
\end{itemize}
with $\Omega_{\kappa}:=\big\{\omega\in\mathbb{C}\mid\Im\omega>-\kappa\big\}$ (recall that $\kappa:=\min\{\kappa_{-},|\kappa_{+}|\}$).

\begin{figure}[!h]
	\centering
	\captionsetup{justification=centering,margin=1.8cm}
	\begin{center}
		\begin{tikzpicture}[scale=1.0]
		\fill[gray!15](7.8,3.8)--(7.8,0)--plot[domain=-3.937:-3.2]({sqrt(exp(-(0.1+\x))-1-\x^2)},4+\x)--(5.59,3.8)--cycle;
		\filldraw[fill=gray!15,dashed](0.2,3.8)--(1.8,3.8)--(1.8,2.2)--(0.2,2.2)--cycle;
		\filldraw[fill=gray!15,dashed](1.8,3.8)--(3.2,3.8)--(3.2,2.2)--(1.8,2.2)--cycle;
		\filldraw[fill=gray!15,dashed](3.2,3.8)--(5.6,3.8)--(5.6,2.2)--(3.2,2.2)--cycle;
		\draw[dashed](5.6,2.2)--(5.6,0.81);
		\draw[dashed](5.6,3.8)--(7.8,3.8);
		\draw[dotted][domain=-3.2:-2.49]plot({sqrt(exp(-(0.1+\x))-1-\x^2)},4+\x);
		\draw[dashed][domain=-3.937:-3.2]plot({sqrt(exp(-(0.1+\x))-1-\x^2)},4+\x);
		\draw(3.75,0.5)node[]{$\Im\lambda=-C_{0}-C_{1}\ln\left\langle\lambda\right\rangle$};
		\draw[->][thick](1,0)--(1,4)node[above]{$\mathrm{i}\mathbb{R}$};
		\draw[->][thick](-0.5,3)--(8,3)node[right]{$\mathbb{R}$};
		\draw(1,3)node[anchor=north west]{$0$};
		
		\draw[thick](0.9,2.2)--(1.1,2.2);
		\draw(1,2.2)node[anchor=north east]{$-\mathrm{i}C_{0}$};
		\draw[thick](0.2,3.1)--(0.2,2.9);
		\draw(0.2,3)node[anchor=north east]{$-R$};
		\draw[thick](1.8,3.1)--(1.8,2.9);
		\draw(1.8,3)node[anchor=north west]{$R$};
		\draw[thick](3.2,3.1)--(3.2,2.9);
		\draw(3.2,3)node[anchor=north west]{$\ell/R$};
		\draw[thick](5.6,3.1)--(5.6,2.9);
		\draw(5.6,3)node[anchor=north west]{$R\ell$};
		
		\draw(0.5,3.5)--(-0.6,4)node[fill=white]{zone I};
		\draw(2.5,3.4)node{zone II};
		\draw(4.5,3.4)node{zone III};
		\draw(6.7,3.4)node{zone IV};
		\end{tikzpicture}
		\caption{\label{The different zones}The four zones.}
	\end{center}
\end{figure}
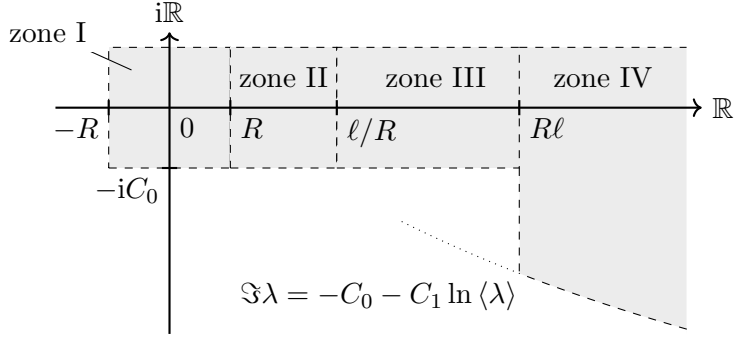
\noindent We quote here all the estimates that we are going to show in this Section in the following theorem (which is similar to \cite[Thm. 2.1]{BoHa08}):

\begin{theorem}
	\label{Main results in the four zones for the harmonic quadratic pencil}
	Let $\chi\in\mathcal{C}_{\mathrm{c}}^{\infty}\left(\mathbb{R},\mathbb{R}\right)$, $s\in\mathbb{R}$ and $\varepsilon,\,\Omega_{\kappa}$ as above. If $s$ is small enough, then the following estimates hold uniformly in $\ell\in\mathbb{N}$:
	\begin{enumerate}
		\item For all $R>0$, $C>0$ and $0<C_0<\varepsilon$, $\mathrm{Res}(p)\cap\left(\left[-R,R\right]+\mathrm{i}\left[-C_{0},C\right]\right)=\emptyset$ and the operator
		\begin{align}
		\chi p\left(z,s\right)^{-1}\chi:L^{2}\to L^{2}
		\end{align}
		exists and is bounded uniformly in $z\in\left[-R,R\right]+\mathrm{i}\left[-C_{0},C\right]$. Moreover, we have
		\begin{align}
		\label{Other bound in the zone I}
		\|\chi p_{\ell}\left(z,s\right)^{-1}\chi\|_{_{L^2\to L^2}}&\leq	\|\chi p\left(z,s\right)^{-1}\chi\|_{_{L^2\to L^2}}\lesssim\prod_{\substack{z_j\in\mathrm{Res}(p)\\|z_j|<2R}}\frac{1}{|z-z_j|}.
		\end{align}
		\item There exist $R>0$ and $0<C_ {0}<\varepsilon$ such that there is no resonance in $\left[R,\ell/R\right]+\mathrm{i}\left[-C_{0},C_{0}\right]$. Furthermore, for all $z\in\left[R,\ell/R\right]+\mathrm{i}\left[-C_{0},C_{0}\right]$, we have
		\begin{align}
		\|\chi p_\ell(z,s)^{-1}\chi\|_{_{L^{2}\to L^{2}}}&\lesssim\frac{1}{\langle z\rangle^{2}}.
		\end{align}
		\item Let $R>0$ and $0<C_0<\varepsilon$ be fixed and suppose that $\ell\gg0$. The number of resonances of $p_\ell$ in $[\ell/R,R\ell]+\mathrm{i}[-C_0,C_0]$ is bounded uniformly in $\ell$ and there exists $C>0$ such that, for all $z\in[\ell/R,R\ell]+\mathrm{i}[-C_0,C_0]$,
		\begin{align}
		\label{final bound in zone III 001}
		\|\chi p_\ell(z,s)^{-1}\chi\|_{_{L^2\to L^2}}&\lesssim\langle z\rangle^C\prod_{\substack{z_j\in\mathrm{Res}(p_\ell)\\|z-z_j|<1}}\frac{1}{|z-z_j|}.
		\end{align}
		Furthermore, there exists $\varepsilon>0$ such that there is no resonance in $[\ell/R,R\ell]+\mathrm{i}[-\varepsilon,0]$ and we have for all $z\in[\ell/R,R\ell]+\mathrm{i}[-\varepsilon,0]$
		\begin{align}
		\label{final bound in zone III 002}
		\|\chi p_\ell(z,s)^{-1}\chi\|_{_{L^2\to L^2}}&\lesssim\frac{\ln\langle z\rangle}{\langle z\rangle}\,\mathrm{e}^{|\Im z|\ln\langle z\rangle}.
		\end{align}
		\item Let $R\gg0$, $C_0>0$ and $C_1>0$. Set
		\begin{align*}
		\tilde{\Omega}_{\ell}&:=\big([R\ell ,+\infty[\,+\,\mathrm{i}\,]-\infty,C_0]\big)\cap\big\{\lambda\in\mathbb{C}\mid\Im\lambda\geq-C_0-C_1\ln\langle\lambda\rangle\big\}\cap\Omega_{\kappa}.
		\end{align*}
		There is no resonance in $\tilde{\Omega}_{\ell}$ and there exists $C>0$ such that for all $z$ in this set,
		\begin{align}
		\label{Estimate on the cut-off inverse of the quadratic pencil in the zone IV formula}
		\|\chi p_\ell(z,s)^{-1}\chi\|_{_{L^2\to L^2}}&\leq C\langle z\rangle^{-1}\mathrm{e}^{C|\Im z|}.
		\end{align}
	\end{enumerate}
\end{theorem}
\begin{remark}
	High frequency resonances of the zone III (i.e. resonances whose real part are of order $\ell\gg0$) are localized in Theorem \ref{Microlocalization of high frequency resonances}.
\end{remark}
The announced estimate in the zone I is a direct application of results of Section \ref{Meromorphic extension and resonances} (see Theorem \ref{No resonances in a strip near the real axis for s small}). We thus show the estimates for the zones II, III and IV.
%
%
\subsection{Estimates in the zone II}
\label{Estimates in the zone II}
%
We prove part 2. of Theorem \ref{Main results in the four zones for the harmonic quadratic pencil} using the complex scaling introduced in \cite[§4]{Zworski_paper}. Observe that the zone II does not exist if $\ell=0$, so that we can assume that $\ell\geq1$. Let $z\in\left[R,\ell/R\right]+\mathrm{i}\left[-C_{0},C_{0}\right]$ and choose $N\in\left[R,\ell/R\right]$ such that $z\in[N,2N]+\mathrm{i}\left[-C_{0},C_{0}\right]$. We introduce the semiclassical parameter
\begin{align*}
h&:=N^{-1}
\end{align*}
and the new spectral parameter
\begin{align*}
\label{New spectral parameter in zone II}
\lambda&:=h^2z^2\in\left[1/4,4\right]+\mathrm{i}\left[-4C_0h,4C_0h\right].
\end{align*}
In this setting, we define the operator
\begin{equation}
\label{Definition of tilde p h 001}
\tilde{p}_{h}(\sqrt{\lambda},s):=h^{2}p_{\ell}\left(z,s\right)=\underbrace{-h^{2}\partial_{x}^{2}+\alpha^2W_{0}\left(x\right)}_{=:Q_{h}}-\lambda+\underbrace{h^{2}W_{1}\left(x\right)+2h\sqrt{\lambda} sV\left(x\right)-h^{2}s^{2}V\left(x\right)^{2}}_{=:R_{h}\left(\lambda\right)}
\end{equation}
where $\alpha:=h(\ell(\ell+1))^{1/2}\gg2\mathscr{A}>0$, $\mathscr{A}$ as in Proposition \ref{Analytic extension of r}.

We now use the ($\alpha$-dependent) contour $\Gamma_{\theta}:=\Gamma_{\theta}^-\cup\Gamma_{\theta}^+$ for $0<\theta<\pi/2$, with\footnote{The factor $1/\kappa_\pm$ in the second argument of $f^\pm_\theta$ comes from the fact that $\kappa_{\pm} x$ corresponds to Zworski's variable $r$.}
\begin{align*}
\Gamma_{\theta}^\pm&:=\left\{x+\mathrm{i}f_{\theta}^\pm(x,\ln(g_\infty^\pm)/\kappa_\pm)\mid x\in\mathbb{R}_\pm\right\}
\end{align*}
where (using estimate \eqref{Exp convergence of r} for $W_0$)
\begin{align*}
	g_\infty^\pm&:=\lim_{x\to\pm\infty}\mathrm{e}^{2\kappa_\pm x}W_0(x)
\end{align*}
and
\begin{align*}
	f^{\pm}_\theta(x,\beta)&:=\begin{cases}
	0&\text{if $|x|\leq\beta/2-C_1$}\\
	\theta(x-\beta/2)&\text{if $|x|\geq\beta/2+C_2$}
	\end{cases}
\end{align*}
with constants $C_1,C_2>0$ as in (4.4) in \cite[§4]{Zworski_paper} (see Figure \ref{The contour Gamma_theta1} for the behaviour of $\Gamma_{\theta}$). Define next $L^{2}(\Gamma_{\theta})$ and $H^{2}(\Gamma_{\theta})$ as the associated Lebesgue and Sobolev spaces. Using the analytic extension of $x\mapsto r(x)$ on the set $\Sigma:=\left\{\eta\in\mathbb{C}\mid\left|\Re\eta\right|>\mathscr{A}\right\}$, we extend $V$, $W_{0}$ and $W_{1}$ on $\Sigma$ (and still denote them $V$, $W_{0}$ and $W_{1}$). We then define the distorted operators
\begin{align*}
\tilde{p}_{h,\theta}(\sqrt{\lambda},s)&=\tilde{p}_{h}(\sqrt{\lambda},s)\upharpoonright_{\Gamma_{\theta}},&Q_{h,\theta}&:=Q_{h}\upharpoonright_{\Gamma_{\theta}},&R_{h,\theta}\left(\lambda\right)&=R_{h}\left(\lambda\right)\upharpoonright_{\Gamma_{\theta}}.
\end{align*}
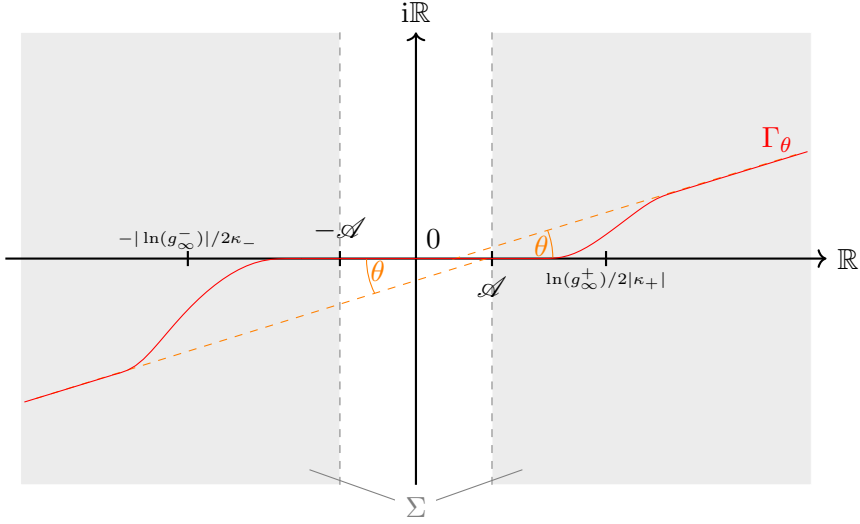
\begin{figure}[!h]
	\centering
	\captionsetup{justification=centering,margin=1.8cm}
	\begin{center}
		\begin{tikzpicture}[scale=1.0]
		\filldraw[fill=gray!15,draw=white](3,-3)--(3,3)--(-1.2,3)--(-1.2,-3)--cycle;
		\filldraw[fill=gray!15,draw=white](5,-3)--(5,3)--(9.2,3)--(9.2,-3)--cycle;
		\draw[gray,dashed](3,-3)--(3,3);
		\draw[gray,dashed](5,-3)--(5,3);
		\draw[->][thick](4,-3)--(4,3)node[above]{$\mathrm{i}\mathbb{R}$};
		\draw[->][thick](-1.4,0)--(9.4,0)node[right]{$\mathbb{R}$};
		\draw(4,0)node[anchor=south west]{$0$};
		%
		\draw[thick](3,-0.1)--(3,0.1)node[anchor=south]{$-\mathscr{A}$};
		\draw[thick](5,-0.1)--(5,0.1);
		\draw[](5,-0.1)node[anchor=north]{$\mathscr{A}$};
		\draw[gray](5.4,-2.8)--(4,-3.3);
		\draw[gray](2.6,-2.8)--(4,-3.3)node[fill=white]{$\Sigma$};
		\draw[orange,dashed](-1,-1.856)--(4.96103,0.01);
		\draw[orange,dashed](4.48,-0.01)--(9.0,1.38);
		\draw[red]plot[domain=7.5:9.15](\x,{0.3080*\x-1.4});
		\draw[red]plot[domain=0.0:-1.15](\x,{0.3080*\x-1.548});
		\draw[orange](3.35,-0.01)arc(0:26.625:-1);
		\draw[orange](3.5,-0.170)node{$\theta$};
		\draw[orange](5.80,0.01)arc(0:21.6:1);
		\draw[orange](5.65,0.170)node{$\theta$};
		\draw[red](2.5,0)--(5.5,0);
		\draw[red]plot[domain=5.5:7.5](\x,{0.15*(1.6130+exp(-1.2*(\x-6.0)^2))*(\x-5.5)*( 1+0.8065*(exp(1/(5.4-\x)+1/(7.6-\x)))/(exp(1/(5.4-\x)+1/(7.6-\x))+exp(-(1/(5.4-\x)+1/(7.6-\x))))-(exp(-(1/(5.4-\x)+1/(7.6-\x))))/(exp(1/(5.4-\x)+1/(7.6-\x))+exp(-(1/(5.4-\x)+1/(7.6-\x)))) )});
		\draw[red]plot[domain=0.0:2.5](\x,{-0.15*(1.6130+exp(-0.4*(\x-1.0)^2))*(2.5-\x)*( 1+0.8065*(exp(-(1/(-0.1-\x)+1/(2.6-\x))))/(exp((1/(-0.1-\x)+1/(2.6-\x)))+exp(-((1/(-0.1-\x)+1/(2.6-\x)))))-(exp((1/(-0.1-\x)+1/(2.6-\x))))/(exp((1/(-0.1-\x)+1/(2.6-\x)))+exp(-(1/(-0.1-\x)+1/(2.6-\x)))) )});
		\draw[red](8.75,1.575)node{$\Gamma_{\theta}$};
		\draw[thick](6.5,-0.1)--(6.5,0.1);
		\draw[thick](6.5,0.015)node[anchor=north]{\tiny{$\ln(g_\infty^+)/2|\kappa_{+}|$}};
		\draw[thick](1.0,-0.1)--(1.0,0.1);
		\draw[thick](1.0,-0.015)node[anchor=south]{\tiny{$-|\ln(g_\infty^-)|/2\kappa_{-}$}};
		\end{tikzpicture}
		\caption{\label{The contour Gamma_theta1}The contour $\Gamma_{\theta}$.}
	\end{center}
\end{figure}
If $q_{h,\theta}$ denotes the symbol of $Q_{h,\theta}$, then \cite[Lem. 4.3]{Zworski_paper} shows that there exists $0<c<1$ and $\theta_0>0$ such that
\begin{align*}
	\left|q_{h,\theta}\left(x,\xi\right)-\lambda\right|&\gtrsim\theta(\left\langle\xi\right\rangle^{2}+\mathrm{e}^{\kappa_{\pm}x}\langle\alpha\rangle^2)&\pm x\geq0
\end{align*}
provided that $\Re(\lambda)>0$, $\Im(\lambda)<c$ and $0<\theta<\theta_0$ (recall that $\pm\kappa_\pm<0$). For $h$ small enough (that is $R$ large enough), we can apply \cite[Prop. 4.1]{Zworski_paper} to get
\begin{align*}
\|({Q}_{h,\theta}-\lambda)^{-1}\|_{_{L^{2}(\Gamma_{\theta})\to H^{2}(\Gamma_{\theta})}}&=\mathcal{O}(\theta^{-1}).
\end{align*}

In order to invert the distorted quadratic pencil, we use a Neumann series argument by showing that $R_{h,\theta}(\lambda)=\mathcal{O}_{L^{2}(\Gamma_{\theta})}(h)$ (as a multiplication operator). In view of the form of $R_{h,\theta}$ in \eqref{Definition of tilde p h 001} and because the extension of $r$ is analytic, it is enough to bound $x\mapsto r(x+\mathrm{i}f^{\pm}_\theta(x))$ from below and above for $|x|>2\mathscr{A}$. By Lagrange inversion formula \eqref{Lagrange inversion series}, we can write
\begin{align*}
	|r(x+\mathrm{i}f^{\pm}_\theta(x))-r_\pm|&\leq\sum_{k=1}^{+\infty}c_k\frac{\mathrm{e}^{2k\kappa_{\pm}x}}{k!}
\end{align*}
for some coefficients $c_k>0$ (recall from Subsection \ref{The Regge-Wheeler coordinate} that $2\kappa_{\pm}=\frac{\Lambda}{3A_\pm r_\pm^2}$), and the series converges (because \eqref{Lagrange inversion series} converges uniformly when $|\Re x|>\mathscr{A}$). Since the sum is decreasing with respect to $x$, we deduce $r<+\infty$ on $\Gamma_{\theta}$. On the other hand,
\begin{align*}
	\left|r-r_{\pm}\right|&=C(r-r_{n})^{-\frac{A_{n}r_{n}^{2}}{A_{\pm}r_{\pm}^{2}}}(r-r_{c})^{-\frac{A_{c}r_{c}^{2}}{A_{\pm}r_{\pm}^{2}}}|r-r_{\mp}|^{-\frac{A_{\mp}r_{\mp}^{2}}{A_{\pm}r_{\pm}^{2}}}\mathrm{e}^{2\kappa_{\pm}x}
\end{align*}
with $C\in\mathbb{R}$. Since no terms on the right-hand side can blow up when restricted on $\Gamma_{\theta}$ and since the exponential goes to zero when $|x|\to+\infty$, it follows that $r\to r_\pm>0$ as $x\to\pm\infty$. We therefore conclude that the restriction of $r$ on $\Gamma_\theta\cap D(0,R_0)^\complement$ is bounded from below and above for $R_0\gg0$, giving $R_{h,\theta}(\lambda)=\mathcal{O}_{L^{2}(\Gamma_{\theta})}(h)$. Thus,
\begin{align*}
\tilde{p}_{h,\theta}(\sqrt{\lambda},s)^{-1}&=\big(1+(Q_{h,\theta}-\lambda)^{-1}R_{h,\theta}(\lambda)\big)^{-1}(Q_{h,\theta}-\lambda)^{-1}.
\end{align*}

We finally choose $\chi\in\mathcal{C}_{\mathrm{c}}^{\infty}\left(\mathbb{R},\mathbb{R}\right)$ and increase if necessary the value of the number $\mathscr{A}$ so that $\mathrm{Supp}\chi\subset\left[-\mathscr{A},\mathscr{A}\right]$. From \cite[Lem. 3.5]{SjZw91}, we have in the $L^{2}$ sense
\begin{align*}
\chi \tilde{p}_{h,\theta}(\sqrt{\lambda},s)^{-1}\chi&=\chi \tilde{p}_{h}(\sqrt{\lambda},s)^{-1}\chi
\end{align*}
whence
\begin{align*}
\|\chi p_{\ell}(z,s)^{-1}\chi\|_{_{L^{2}\to L^{2}}}&=h^{2}\|\chi \tilde{p}_{h}(\sqrt{\lambda},s)^{-1}\chi\|_{_{L^{2}\to L^{2}}}\lesssim\langle z\rangle^{-2}.
\end{align*}
%
%
%
\subsection{Estimates in the zone III}
\label{Estimates in the zone III}
%
We turn to the proof of part 3. of Theorem \ref{Main results in the four zones for the harmonic quadratic pencil}. We define the semiclassical parameter
\begin{align*}
h:=(\ell(\ell+1))^{-1/2}
\end{align*}
with again $\ell>0$ since the zone does not exist for $\ell=0$. For $z\in\left[\ell/R,R\ell \right]+\mathrm{i}\left[-C_{0},C_{0}\right]$, we define a new spectral parameter
\begin{align*}
\lambda:=h^2z^2\in\left[\frac{1}{3R^2},R^2\right]+\mathrm{i}[-\sqrt{2}\,C_0Rh,\sqrt{2}\,C_0Rh]\subset[a,b]+\mathrm{i}[-ch,ch]
\end{align*}
for some $0<a<b$ and $c>0$. Finally, we set
\begin{align*}
\tilde{P}_h&:=h^2P_\ell=-h^2\partial_x^2+W_0+h^2W_1,&\tilde{p}_h(\sqrt{\lambda},s)&:=\tilde{P}_h-(\sqrt{\lambda}-hsV)^2
\end{align*}
and write $\tilde{P}_\theta$ and $\tilde{p}_\theta$ for the corresponding distorted operators on the contour $\Gamma_{\theta}$ as we did in the paragraph \ref{Estimates in the zone II}. We are still using the subscript $L^2$ when we work with the distorted operators.

As $W_0$ admits a non-degenerate maximum at $x=0$ (see Section \ref{Functional framework}, Figure \ref{The potential W0}), $(x,\xi)=(0,0)$ is a trivial solution of the Hamilton equations associated to the principal symbol of $\tilde{P}_{h}$:
\begin{align*}
\begin{cases}
\dot{x}&=2\xi\\
\dot{\xi}&=-W_0'(x)
\end{cases}.
\end{align*}
Therefore the energy level $\mathfrak{E}_0:=W_0(0)$ is trapping. For this reason, the zone III is called the trapping zone.

We first show an adaptation of \cite[Lem. 6.5]{BoMi04} to our setting.

\begin{proposition}
	\label{Adaptation BoMi04 Theorem}
	For $\theta=Nh$ with $N>0$ large enough and $s\in\mathbb{R}$ sufficiently small, there exist $C\equiv C(N)>0$ and $\varepsilon>0$ such that, for all $E\in[\mathfrak{E}_0-\varepsilon,\mathfrak{E}_0+\varepsilon]$ and $|\lambda-E|\leq\varepsilon\theta/2$, it holds
	\begin{align*}
	\|(\tilde{P}_\theta-(\sqrt{\lambda}-hsV)^2)^{-1}\|_{_{L^2\to L^2}}&=\mathcal{O}(h^{-C})\prod_{\substack{\lambda_j\in\mathrm{Res}(\tilde{p})\\|\lambda-\lambda_j|<\varepsilon\theta}}\frac{h}{|\lambda-\lambda_j|}.
	\end{align*}
\end{proposition}

\begin{proof}
	The announced estimate is known for the resolvent $(\tilde{P}_\theta-\lambda)^{-1}=\tilde{p}_\theta(\sqrt{\lambda},0)^{-1}$ corresponding to the case $s=0$. The argument can be found in \cite{TaZw98} which uses techniques developed in \cite{Sj96}, and the authors of \cite{BoMi04} adapted it for the one dimensional case of a non degenerate trapping energy level $\mathfrak{E}_0$. More precisely, for $\theta=Nh$ with $N\gg0$ large enough, one can construct a bounded operator $\tilde{K}\in\mathcal{L}\left(L^2,L^2\right)$ (see (6.15) in \cite{BoMi04}) satisfying the following properties:
	\begin{itemize}
		\item [$(i)$] $\|\tilde{K}\|_{_{L^2\to L^2}}=\mathcal{O}(1)$,
		\item [$(ii)$] $r:=\mathrm{rank}\,\tilde{K}\leq\mathcal{O}(\theta h^{-1}\ln(1/\theta))$,
		\item [$(iii)$] for $h$ small enough, there exists $\varepsilon>0$ such that, for all $E\in[\mathfrak{E}_0-\varepsilon,\mathfrak{E}_0+\varepsilon]$ and $\lambda\in[E-\varepsilon\theta,E+\varepsilon\theta]$, $$\|(\tilde{P}_\theta-\mathrm{i}\theta\tilde{K}-\lambda)^{-1}\|_{_{L^2\to\mathscr{D}}}\leq\mathcal{O}(\theta^{-1}),\qquad\qquad\mathscr{D}:=\mathscr{D}(\tilde{P}_\theta).$$
	\end{itemize}
	%
	%
	In \cite[Lem. 3.2]{SjZw91}, it is  shown that $\tilde{P}_\theta-\lambda$ is a Fredholm operator from its domain $\mathscr{D}$ to $L^2$, so we can construct a well-posed Grushin problem
	\begin{align}
	\label{Grushin problem in zone III}
		\mathcal{P}(\lambda)&:=\begin{pmatrix}
		\tilde{P}_\theta-\lambda&R_-\\
		R_+&0_{_{\mathbb{C}^r\to\mathbb{C}^r}}
		\end{pmatrix}:\mathscr{D}\oplus\mathbb{C}^r\to L^2\oplus\mathbb{C}^r
	\end{align}
	where $R_-$ and $R_+$ are constructed with $\tilde{P}_\theta-\mathrm{i}\theta\tilde{K}-\lambda$ (see \cite{Sj96}, page 401, below (6.12) for the construction).
	
	Now consider $s\neq0$. If $s$ is small enough, $(\tilde{P}_\theta-\mathrm{i}\theta\tilde{K}-(\sqrt{\lambda}-hsV)^2)^{-1}$ is invertible by pseudodifferential calculus\footnote{See the definition of the operator $K$ at the beginning of the proof of \cite[Lem. 6.5]{BoMi04}.} as for the case $s=0$. By the resolvent identity, one can show that
	\begin{align*}
		\|(\tilde{P}_\theta-\mathrm{i}\theta\tilde{K}-(\sqrt{\lambda}-hsV)^2)^{-1}\|_{_{L^2\to\mathscr{D}}}&\leq \mathcal{O}(\theta^{-1})+\mathcal{O}(h|s|)\|(\tilde{P}_\theta-\mathrm{i}\theta\tilde{K}-(\sqrt{\lambda}-hsV)^2)^{-1}\|_{_{L^2\to\mathscr{D}}}\mathcal{O}(\theta^{-1})
	\end{align*}
	since $|\lambda|\leq h(|\Re z|+|\Im z|)\leq \mathcal{O}(1)+\mathcal{O}(h)$. Hence for $s$ sufficiently small, we have
	\begin{align*}
	\|(\tilde{P}_\theta-\mathrm{i}\theta\tilde{K}-(\sqrt{\lambda}-hsV)^2)^{-1}\|_{_{L^2\to\mathscr{D}}}&\leq\mathcal{O}(\theta^{-1})&\lambda\in[E-\varepsilon\theta,E+\varepsilon\theta].
	\end{align*}
	%
	%
	%
	Because the quadratic pencil remains a Fredholm operator provided that $\|hsV\|_{_{L^\infty}}$ is sufficiently small\footnote{Recall that the set of Fredholm operators in $\mathcal{L}(\mathscr{D},L^2)$ is open for the norm topology.}, we can write a new well-posed Grushin problem 
	\begin{align*}
		\mathcal{P}(\lambda)&:=\begin{pmatrix}
		\tilde{p}_{\theta}(\sqrt{\lambda},s)&R_-\\
		R_+&0_{_{\mathbb{C}^r\to\mathbb{C}^r}}
		\end{pmatrix}:\mathscr{D}\oplus\mathbb{C}^r\to L^2\oplus\mathbb{C}^r
	\end{align*}
	where this time $R_+$ and $R_-$ are constructed with $\tilde{P}_\theta-\mathrm{i}\theta\tilde{K}-(\sqrt{\lambda}-hsV)^2$. 
	If we set
	\begin{align*}
		\mathcal{E}(\lambda):=\mathcal{P}(\lambda)^{-1}&=\begin{pmatrix}
		E(\lambda)&E_+(\lambda)\\
		E_-(\lambda)&E_0(\lambda)
		\end{pmatrix}:L^2\oplus\mathbb{C}^r\to\mathscr{D}\oplus\mathbb{C}^r,
	\end{align*}
	then the relations $\mathcal{E}(\lambda)\mathcal{P}(\lambda)=\mathcal{P}(\lambda)\mathcal{E}(\lambda)=\mathrm{Id}$ as well as the following estimate (which is a consequence of properties $(i)$ and $(iii)$ above)
	\begin{align*}
		\big\|(\tilde{P}_\theta-\mathrm{i}\theta\tilde{K}-(\sqrt{\lambda}-hsV)^2)^{-1}(\tilde{P}_\theta-(\sqrt{\lambda}-hsV)^2)\big\|_{_{L^{2}\to L^{2}}}=\mathcal{O}(1)
	\end{align*}
	imply as in \cite{BoMi04} that $\|E(\lambda)\|_{_{L^2\to\mathscr{D}}},\|E_-(\lambda)\|_{_{L^2\to\mathbb{C}^r}}=\mathcal{O}(\theta^{-1})$ and $\|E_+(\lambda)\|_{_{\mathbb{C}^r\to\mathscr{D}}},\|E_0(\lambda)\|_{_{\mathbb{C}^r\to\mathbb{C}^r}}=\mathcal{O}(1)$.
	Applying formula (8.11) in \cite{Sj96}, we obtain
	\begin{align*}
		(\tilde{P}_\theta-(\sqrt{\lambda}-hsV)^2)^{-1}&=E(\lambda)-E_+(\lambda)E_0(\lambda)^{-1}E_-(\lambda)
	\end{align*}
	which implies
	\begin{align*}
		\|(\tilde{P}_\theta-(\sqrt{\lambda}-hsV)^2)^{-1}\|_{_{L^2\to\mathscr{D}}}&=\mathcal{O}(\theta^{-1})\big(1+\|E_0(\lambda)^{-1}\|_{_{\mathbb{C}^r\to\mathbb{C}^r}}\big)
	\end{align*}
	as in \cite[Lem. 6.5]{BoMi04}, and we then follow the end of its proof to conclude.
\end{proof}
We can now follow the arguments below \cite[Lem. 2.2]{BoHa08}. The set of pseudo-poles (2.28) and the injective map (2.29) in this reference exist in our setting by Theorem \ref{Microlocalization of high frequency resonances} (but are quite different). This implies that there is no resonance in $\Omega(h):=[a/2,2b]+\mathrm{i}[-\varepsilon h,ch]$ provided that $h$ and $s$ are small enough. As a result, \eqref{final bound in zone III 001} holds true. As for the estimate \eqref{final bound in zone III 002}, we use Burq's Lemma:
\begin{lemma}[Lemma 2.3 in \cite{BoHa08}]
  Suppose that $f(\lambda,h)$ is a family of holomorphic functions defined for $0<h<1$ in a neighbourhood of $\Omega(h):=[a/2,2b]+\mathrm{i}[-ch,ch]$ with $0<a<b$ and $c>0$, such that
  \begin{align*}
    |f(\lambda,h)|&\lesssim\begin{cases}
    h^{-C'}&\text{in }\Omega(h)\\
    \displaystyle\frac{1}{|\Im\lambda|}&\text{in }\Omega(h)\cap\mathbb{C}^+
    \end{cases}
  \end{align*}
  for some $C'>0$. Then there exists $h_0$, $C>0$ such that, for any $0<h<h_0$ and any $\lambda\in[a,b]+\mathrm{i}[-ch,0]$,
  \begin{align*}
    |f(\lambda,h)|&\leq C\frac{|\ln h|}{h}\mathrm{e}^{C|\Im\lambda||\!\ln h|/h}.
  \end{align*}
\end{lemma}
We apply this result to the function $f(\lambda,h):=\|\chi(\tilde{P}_h-(\sqrt{\lambda}-hsV)^2)^{-1}\chi\|$, observing that for all $\lambda\in\Omega(h)\cap\mathbb{C}^+$ the resolvent identity gives
\begin{align*}
\|(\tilde{P}_h-(\sqrt{\lambda}-hsV)^2)^{-1}\|&\leq\frac{1}{|\Im \lambda|}+\|(\tilde{P}_h-(\sqrt{\lambda}-hsV)^2)^{-1}\|\frac{\mathcal{O}(h|s|)}{|\Im \lambda|}\lesssim\frac{1}{|\Im \lambda|}
\end{align*}
because $\|(\tilde{P}_h-(\sqrt{\lambda}-hsV)^2)^{-1}\|$ is uniformly bounded on this set for $h$ and $s$ small enough.
%
%
%
\subsection{Estimates in the zone IV}
\label{Estimates in the zone IV}
%
This last paragraph is devoted to the proof of part 4. of Theorem \ref{Main results in the four zones for the harmonic quadratic pencil}. 
For $z\in\left(\left[R\ell ,+\infty\right[+\mathrm{i}\left]-\infty,C_{0}\right]\right)\cap\left\{\lambda\in\mathbb{C}\mid \Im\lambda\geq-C_{0}-C_{1}\ln\langle\lambda\rangle\right\}$, there exists a number $N>R\ell>0$ such that $z\in\left[N,2N\right]+\mathrm{i}\left[-C\ln N,C_{0}\right]$. We introduce the semiclassical parameters
\begin{align*}
&h:=\frac{1}{N},&\mu:=\ell\left(\ell+1\right)h^{2},&&\nu:=h^{2}.
\end{align*}
Observe that these parameters are very small when $N\gg0$. Moreover, we can consider that $h\leq1$ even if $\ell=0$, simply by taking $R\geq1$ in the zone I if it was not the case ($R$ as in Theorem \ref{Main results in the four zones for the harmonic quadratic pencil}). We then define a new spectral parameter
\begin{align*}
\lambda&:=z^2h^2\in\left[1,2\right]+\mathrm{i}\left[Ch\ln h,C_{0}h\right]\subset[a,b]+\mathrm{i}\left[-ch\left|\ln h\right|,ch\right]
\end{align*}
where $0<a\leq1<2\leq b<+\infty$ and $\max\left\{C,C_{0}\right\}<c<+\infty$ (observe that $a$ and $b$ do not depend on $h$). Let $J:=[a,b]$ and set
\begin{align*}
J^{+}:=\{\eta\in\mathbb{C}^{+}\mid\Re(\eta)\in J\}.
\end{align*}
Define then
\begin{align*}
\tilde{P}_{h}&:=h^{2}P_{\ell}=-h^{2}\partial_{x}^{2}+\mu W_{0}+\nu W_{1}, &
\tilde{p}_{h}(\sqrt{\lambda},s)&:=h^2p_{\ell}\left(z,s\right)=\tilde{P}_{h}-(\sqrt{\lambda}-hsV)^{2}.
\end{align*}
\paragraph{Semiclassical limiting absorption principle for the quadratic pencil.} As in \cite{BoHa08}, we first get a control until the real line by using a semiclassical limiting absorption principle for the semiclassical quadratic pencil. The appendix \ref{Limiting absorption principle for the quadratic pencil} provides a proof, close to the idea developed by G\'{e}rard \cite{Ger08}, of such a result for a class of perturbed resolvents, so we only have to check if the required abstract assumptions are satisfied.

Introduce the generator of dilations $A:=-\mathrm{i}h\left(x\partial_{x}+\partial_{x}x\right)$ with domain $\mathscr{D}\left(A\right):=\left\{u\in L^{2}\mid Au\in L^{2}\right\}$. We then pick $\rho\in\mathcal{C}_{\mathrm{c}}^{\infty}\left(\mathbb{R},[0,1]\right)$ such that $\mathrm{Supp}\ \rho\subset[a/3,3b]$ and $\rho\equiv1$ on $I:=[a/2,2b]$, and we define $\mathcal{A}$ as the closure of the operator $\rho(P)A\rho(P)$. In this setting, $\rho(P)A\rho(P)$ is well-defined on $\mathscr{D}(A)$, $\mathcal{A}$ is self-adjoint and we have $P\in\mathcal{C}^{2}\left(\mathcal{A}\right)$ (\textit{cf.} \cite[§2.4]{BoHa08}) so that \eqref{Assumption P} holds. A direct computation shows that
\begin{align*}
	\mathrm{i}h^{-1}[P,A]&=4P-4\mu W_{0}-4\nu W_{1}-2\mu xW_{0}'-2\nu xW_{1}'
\end{align*}
so that, for $\mu$ and $\nu$ sufficiently small, we get the Mourre estimate \eqref{Assumption M} (uniform in $\mu$, $\nu$)
\begin{align*}
	\mathds{1}_{I}(P)[P,\mathrm{i}\mathcal{A}]\mathds{1}_{I}(P)&\geq ah\mathds{1}_{I}(P).
\end{align*}
Since $V\in\mathcal{B}(\mathscr{D}(P),L^2)$ it is clear that $V\in L^\infty_{\ell\mathrm{oc}}(\tilde{P}_{h})$. Moreover, assumption \eqref{Assumption H} is fulfilled for $f(z,B):=(\sqrt{z}-sB)^2$.

It remains to show that assumption \eqref{Assumption A} is satisfied for $B=hV$. Observe that this abstract assumption is particularly well adapted to semiclassical pseudodifferential calculus framework, especially the commutator estimate which provides the supplementary term $h$. In \cite{Ha01}, it is shown that $\mathcal{A}\in\Psi^{-\infty,1}$ ($\mathcal{A}$ is the operator $c_{\tilde{\chi}}$ in \cite{Ha01}, see above Lemma 3.3). We will use it to show the following result:
\begin{lemma}
	\label{AVA bounded}
	Let $\sigma\in[0,1]$. Then $V\in\mathcal{B}(\mathscr{D}(\langle\mathcal{A}\rangle^{\sigma}))$ and $[V,\chi(\tilde{P}_h)]\in h\mathcal{B}(\mathscr{D}(\langle\mathcal{A}\rangle^{\sigma}))$.
\end{lemma}
\begin{proof}
Let $\Omega:=[0,1]+\mathrm{i}\mathbb{R}$ and let $z\in\Omega$. On $\mathscr{D}(\langle\mathcal{A}\rangle^{2})\times\mathscr{D}(\langle\mathcal{A}\rangle^{2})$, we define the sesquilinear form
\begin{align*}
	Q_{z}(\varphi,\psi)&:=\langle V\langle\mathcal{A}\rangle^{-2z}\varphi,\langle\mathcal{A}\rangle^{2z}\psi\rangle\qquad\forall\varphi,\psi\in\mathscr{D}(\langle\mathcal{A}\rangle^{2}).
\end{align*}
By functional calculus, $Q_{z}$ is well-defined and analytic in $z\in\Omega$. When $z\in\{0\}+\mathrm{i}\mathbb{R}$, $|(1+\lambda^2)^{z/2}|=1$ for all $\lambda\in\mathbb{R}$ so that functional calculus first applied to $\langle\mathcal{A}\rangle^{2z}$ and then to $\langle\mathcal{A}\rangle^{-2z}$ gives
\begin{align*}
	|Q_{z}(\varphi,\psi)|&\leq|\langle V\langle\mathcal{A}\rangle^{-2z}\varphi,\psi\rangle|\\
	&=|\langle \langle\mathcal{A}\rangle^{-2z}\varphi,V\psi\rangle|\\
	&\leq|\langle\varphi,V\psi\rangle|\\
	&\leq\|V\|_{_{L^{\infty}}}\|\varphi\|_{_{L^{2}}}\|\psi\|_{_{L^{2}}}.
\end{align*}
When $z=1$, pseudodifferential calculus shows that $\langle\mathcal{A}\rangle^{2}V\langle\mathcal{A}\rangle^{-2}\in\Psi^{0,0}$, so that for all $z\in\{1\}+\mathrm{i}\mathbb{R}$ (using again functional calculus for $\langle\mathcal{A}\rangle^{\pm 2\mathrm{i}\Im z}$),
\begin{align*}
	|Q_{z}(\varphi,\psi)|&\leq|\langle \langle\mathcal{A}\rangle^{2}V\langle\mathcal{A}\rangle^{-2}\langle\mathcal{A}\rangle^{-2\mathrm{i}\Im z}\varphi,\psi\rangle|\\
	&\leq\|\langle\mathcal{A}\rangle^{2}V\langle\mathcal{A}\rangle^{-2}\|_{_{L^2\to L^2}}\|\langle\mathcal{A}\rangle^{-2\mathrm{i}\Im z}\varphi\|_{_{L^{2}}}\|\psi\|_{_{L^{2}}}\\
	&\leq\|\langle\mathcal{A}\rangle^{2}V\langle\mathcal{A}\rangle^{-2}\|_{_{L^2\to L^2}}\|\varphi\|_{_{L^{2}}}\|\psi\|_{_{L^{2}}}.
\end{align*}
By the maximum principle, there exists a constant $C>0$ such that $Q_{z}$ is bounded by $C$ for all $0\leq\Re z\leq1$. In particular, we can extend $Q_{\sigma/2}$ on $L^{2}\times\mathscr{D}(\langle\mathcal{A}\rangle^{2})$ as a bounded sesquilinear form and for $\sigma\in[0,1]$ and $\varphi\in L^2$, we have
\begin{align*}
	|Q_{\sigma/2}(\varphi,\psi)|&\leq C\|\varphi\|_{_{L^2}}\|\psi\|_{_{L^2}}.
\end{align*}
This means that the map $\mathscr{D}(\langle\mathcal{A}\rangle^{\sigma})\ni\psi\mapsto\langle V\langle\mathcal{A}\rangle^{-\sigma}\varphi,\langle\mathcal{A}\rangle^{\sigma}\psi\rangle$ is continuous. By definition of the adjoint operator and because $\langle\mathcal{A}\rangle^{\sigma}$ is self-adjoint, this implies that $V\langle\mathcal{A}\rangle^{-\sigma}\varphi\in\mathscr{D}(\langle\mathcal{A}\rangle^{\sigma})$ for all $\varphi\in L^2$.
	
Consider now the sesquilinear form
\begin{align*}
	\tilde{Q}_{z}(\varphi,\psi)&:=\langle[V,\chi(\tilde{P}_h)]\langle\mathcal{A}\rangle^{-2z}\varphi,\langle\mathcal{A}\rangle^{2z}\psi\rangle\qquad\forall\varphi,\psi\in\mathscr{D}(\langle\mathcal{A}\rangle^{2}).
\end{align*}
By semiclassical pseudodifferential calculus, we have (see \textit{e.g.} (4.4.19) in \cite{Zw})
\begin{align*}
	[V,\chi(\tilde{P}_h)]&=\frac{h}{\mathrm{i}}\{V(x),\chi(\xi^2+\mu W_0(x)+\nu W_1(x))\}^{\mathrm{w}}+h^3\Psi^{-\infty,0}\\
	&=h\Psi^{-\infty,-\infty}+h^3\Psi^{-\infty,0}
\end{align*}
because $V(x)\in\Psi^{0,0}$, $V'(x)\in\Psi^{0,-\infty}$ and $\chi(\tilde{P}_h)\in\Psi^{-\infty,0}$. Despite the fact that the error term above looks less regular than the main term, it is in fact more regular as it can be shown using expansion (4.4.15) in \cite{Zw} (but we will not need such a regularity). Now we can proceed as above with $Q_z$ and $V$ to conclude.
\end{proof}
Now that all assumptions in appendix \ref{Limiting absorption principle for the quadratic pencil} have been checked, we can use Theorem \ref{ASALP} as well as the fact that $\|\langle x\rangle^{-\sigma}\langle\mathcal{A}\rangle^{\sigma}\|\lesssim1$ for all $\sigma\leq1$\footnote{We show it using the sesquilinear form $(\varphi,\psi)\mapsto\langle\langle x\rangle^{-\sigma}\langle\mathcal{A}\rangle^{\sigma}\varphi,\psi\rangle$ first well-defined on $\mathscr{D}(\mathcal{A}^2)\times\mathscr{D}(\mathcal{A}^2)$ because $\langle x\rangle^{-2}\in\Psi^{0,-2}$, and then extended to $L^2\times L^2$ by maximum principle.}: for $\sigma\in\left]1/2,1\right]$ and $h$ small enough, we have uniformly in $\mu,\nu$
\begin{align*}
	\sup_{\lambda\in J^{+}}\left\|\langle x\rangle^{-\sigma}\tilde{p}_h(\sqrt{\lambda},s)^{-1}\langle x\rangle^{-\sigma}\right\|&\leq\left\|\langle x\rangle^{-\sigma}\langle\mathcal{A}\rangle^{\sigma}\right\|\left(\,\sup_{\lambda\in J^{+}}\left\|\langle\mathcal{A}\rangle^{-\sigma}\tilde{p}_h(\sqrt{\lambda},s)^{-1}\langle\mathcal{A}\rangle^{-\sigma}\right\|\right)\left\|\langle\mathcal{A}\rangle^{\sigma}\langle x\rangle^{-\sigma}\right\|\lesssim h^{-1}.
\end{align*}
\paragraph{Estimates below the real axis.} Next, we can use the work of Martinez \cite{Ma} to get a bound under the real line. Indeed, Section 4 of the last reference applies in our setting because $\tilde{p}_h(\sqrt{\lambda},s)$ is a differential operator (so that \cite[Prop. 3.1 \& Cor. 3.2]{Ma} apply) and because $(\lambda-hsV(x))^2\in[\lambda-\delta,\lambda+\delta]+\mathrm{i}[ch\ln h,0]$ for all $\lambda$ in the zone IV and all $x\in\mathbb{R}$ if $s$ is small enough (so that the estimate (4.6) in \cite{Ma} still holds). It follows that equation (4.13) holds with $\tilde{p}_h(\sqrt{\lambda},s)$ instead of $P_{\theta}-\rho$\footnote{We can in fact insert any pseudodifferential operator here provided that hypotheses in \cite[§2]{Ma} are verified.}. In our setting, this reads
\begin{align}
\label{Martinez bound}
	\|\chi\tilde{p}_h(\sqrt{\lambda},s)^{-1}\chi\|&\leq Ch^{-C}
\end{align}
for some $C>0$.

To get \eqref{Estimate on the cut-off inverse of the quadratic pencil in the zone IV formula}, we reproduce the argument at the end of the proof of \cite[Lem. 2.4]{BoHa08}. Choose $f$ holomorphic satisfying the following conditions:
\begin{align*}
	\begin{cases}
		|f|<1&\text{for $\lambda\in[a/2,2b]+\mathrm{i}[ch\ln h,0]$},\\
		|f|\geq 1&\text{for $\lambda\in[a,b]+\mathrm{i}[ch\ln h,0]$},\\
		|f|\leq h^C&\text{for $\lambda\in[a/2,2b]\setminus[2a/3,3b/2]+\mathrm{i}[ch\ln h,0]$}
	\end{cases}
\end{align*}
where $C>0$ is the constant in \eqref{Martinez bound}. Since $f$ is holomorphic, the function
\begin{align*}
	g(\lambda)&:=\ln\|\chi\tilde{p}(\sqrt{\lambda},s)^{-1}\chi\|_{_{L^2\to L^2}}+\ln|f(\lambda)|+\frac{C}{ch}\Im\lambda
\end{align*}
is subharmonic. We can check that $g(\lambda)\lesssim\ln(h^{-1})$ on the boundary of $[a/2,2b]+\mathrm{i}[ch\ln h,0]$. By the maximum principle, this estimate holds for all $\lambda\in[a/2,2b]+\mathrm{i}[ch\ln h,0]$, whence
\begin{align*}
	\|\chi\tilde{p}_h(\sqrt{\lambda},s)^{-1}\chi\|_{_{L^2\to L^2}}&\lesssim h^{-1}\mathrm{e}^{\frac{C}{ch}|\Im\lambda|}.
\end{align*}
The desired estimate \eqref{Estimate on the cut-off inverse of the quadratic pencil in the zone IV formula} then follows.
%
%
%
%
%
%
%
%
\section{Proof of Theorem \ref{Theoreme principal}}
\label{Proof of Theorem {Theoreme principal}}
We prove in this section Theorem \ref{Theoreme principal}. The resonance expansion \eqref{Expansion du propagateur} follows from the theory of resonances as presented in \cite[§3]{BoHa08}, and we can follow the proof of this paper. We only have to adapt \cite[Prop. 3.1]{BoHa08} to get an estimate for the resolvent $\hat{R}_{\chi,\ell}(z)$:
\begin{proposition}
	\label{Estimee pour la resolvente tronquee de Lell dans Emodell}
	Let $\ell\in\mathbb{N}$ and let $\chi\in\mathcal{C}^{\infty}_{\mathrm{c}}\left(\mathbb{R},\mathbb{R}\right)$. There exists $\tilde{\chi}\in\mathcal{C}^{\infty}_{\mathrm{c}}\left(\mathbb{R},\mathbb{R}\right)$ satisfying $\tilde{\chi}\chi=\chi$ such that for all $z\in\mathbb{C}\setminus\mathrm{Res}(p_\ell)$, the cut-off resolvent $\chi(\hat{K}_{\ell}-z)^{-1}\chi$ is a bounded operator on $\dot{\mathcal{E}}_{\ell}$ and satisfies uniformly in $\ell$
	\begin{align*}
		\|\hat{R}_{\chi,\ell}(z)\|_{_{\dot{\mathcal{E}}_{\ell}\to\dot{\mathcal{E}}_{\ell}}}\lesssim\langle z\rangle\|\tilde{\chi}p_\ell(z,s)^{-1}\tilde{\chi}\|_{_{L^{2}\to L^{2}}}.
	\end{align*}
\end{proposition}

\begin{proof}
	Since the norms $\|.\|_{_{\mathcal{E}_{\ell}}}$ and $\|.\|_{_{\dot{\mathcal{E}}_{\ell}}}$ are locally equivalent thanks to the Hardy type estimate $\|\chi.\|_{_{L^2}}\lesssim\|P_\ell^{1/2}.\|_{_{L^2}}$ uniformly in $\ell$ (\textit{cf.} \cite[Lem. 9.5]{GeoGerHA17}), we can work on $(\mathcal{E}_{\ell},\|.\|_{_{\mathcal{E}_{\ell}}})$.
	For $\left(u_{0},u_{1}\right)\in\mathcal{E}_{\ell}$, we have
	\begin{align}
	\label{Application of the cut-off charge Klein-Gordon operator to elements of mathcal{E}}
		\hat{R}_{\chi,\ell}(z)\begin{pmatrix}
		u_0\\u_1
		\end{pmatrix}&=\begin{pmatrix}
		\chi p_\ell(z,s)^{-1}\chi((z-sV)u_0+u_1)\\
		\chi(1+(z-sV)p_\ell(z,s)^{-1}(z-sV))\chi u_0+(z-sV)\chi p_\ell(z,s)^{-1}\chi u_1
		\end{pmatrix}
	\end{align}
	and since it holds
	\begin{align*}
		\|(z-sV)\chi p_\ell(z,s)^{-1}\chi u_1\|_{_{L^2}}\leq(1+|s|\|V\|_{_{L^\infty}})\langle z\rangle\|\chi p_\ell(z,s)^{-1}\chi\|_{_{L^{2}\to L^{2}}}\|u_1\|_{_{L^2}},
	\end{align*}
	the $\mathcal{E}_{\ell}$-norm of \eqref{Application of the cut-off charge Klein-Gordon operator to elements of mathcal{E}} can be bounded if we show the following estimates:
	\begin{subequations}
	\begin{align}
	\label{Estimate C_a}
		\|P_\ell^{1/2}\chi p_\ell(z,s)^{-1}\chi(z-sV)u_0\|_{_{L^2}}&\leq C_\mathrm{a}\langle z\rangle\|\tilde{\chi}p_\ell(z,s)^{-1}\tilde{\chi}\|_{_{L^2\to L^2}}\|P_\ell^{1/2}u_0\|_{_{L^2}},
	\end{align}
	\begin{align}
	\label{Estimate C_b}
		\|P_\ell^{1/2}\chi p_\ell(z,s)^{-1}\chi u_1\|_{_{L^2}}&\leq C_\mathrm{b}\langle z\rangle\|\tilde{\chi}p_\ell(z,s)^{-1}\tilde{\chi}\|_{_{L^2\to L^2}}\|u_1\|_{_{L^2}},
	\end{align}
	\begin{align}
	\label{Estimate C_c}
		\|\chi(1+(z-sV)p_\ell(z,s)^{-1}(z-sV))\chi u_0\|_{_{L^2}}&\leq C_\mathrm{c}\langle z\rangle\|\tilde{\chi}p_\ell(z,s)^{-1}\tilde{\chi}\|_{_{L^2\to L^2}}\|P_\ell^{1/2}u_0\|_{_{L^2}}.
	\end{align}
	\end{subequations}
	We use complex interpolation.
	\paragraph{Estimate \eqref{Estimate C_a}.} Let us define
	\begin{align*}
		\Lambda_\mathrm{a}(\theta)&:=\langle z\rangle^{-2\theta}P_\ell^{\theta}\chi p_\ell(z,s)^{-1}\chi P_\ell^{-\theta}.
	\end{align*}
	By functional calculus, $\Lambda_\mathrm{a}$ is analytic from $[0,1]+\mathrm{i}\mathbb{R}$ to $\mathcal{L}(L^2,L^2)$ because $P_\ell>0$ and $\langle z\rangle>0$. We want to show that
	\begin{align*}
		\|\Lambda_\mathrm{a}(1/2)u\|_{_{L^2}}&\leq C_\mathrm{a}\|\tilde{\chi}p_\ell(z,s)^{-1}\tilde{\chi}\|_{_{L^2\to L^2}}\|u\|_{_{L^2}}&\forall u\in L^2
	\end{align*}
	for some $C_\mathrm{a}>0$. By the maximum principle, it is sufficient to bound $\Lambda_\mathrm{a}(\theta)$ for $\theta\in\{0,1\}+\mathrm{i}\mathbb{R}$, and since $P_\ell$ is self-adjoint, it is sufficient by functional calculus to restrict ourselves to $\Im\theta=0$. If $\theta=0$, there is nothing to do. Now for $\theta=1$, we put $u=(z-sV)u_0$ and try to show that
	\begin{align}
	\label{Estimate to show for C_a}
		\|P_\ell\chi p_\ell(z,s)^{-1}\chi u\|_{_{L^2\to L^2}}&\leq C_\mathrm{a}\langle z\rangle\|\tilde{\chi}p_\ell(z,s)^{-1}\tilde{\chi}\|_{_{L^2\to L^2}}\|P_\ell u\|_{_{L^2}}.
	\end{align}
	Write
	\begin{align}
	\label{Useful identity}
		P_\ell\chi p_\ell(z,s)^{-1}\chi&=\underbrace{[P_\ell,\chi] p_\ell(z,s)^{-1}\chi}_{=:A}+\underbrace{\chi P_\ell p_\ell(z,s)^{-1}\chi}_{=:B}.
	\end{align}
	We first deal with $A$. Pick $z_0\in\rho(\hat{K}_{\ell})\cap \mathbb{C}^+$ so that $p_{\ell}(z_0,s)^{-1}$ exists (\textit{cf.} \eqref{Relation between the domains of the quadratic pencil and the corresponding Hamiltonian}). Then
	\begin{align*}
		\big[P_\ell,\chi\big]p_\ell(z,s)^{-1}&=p_\ell(z_0,s)^{-1}\big[p_\ell(z_0,s),\big[P_\ell,\chi\big]\big]p_\ell(z,s)^{-1}+p_\ell(z_0,s)^{-1}\big[P_\ell,\chi\big]p_\ell(z_0,s)p_\ell(z,s)^{-1}
	\end{align*}
	with 
	\begin{align*}
		\big[P_\ell,\chi\big]&=-\chi\partial_x-\chi',\\
		\big[p_\ell(z_0,s),\big[P_\ell,\chi\big]\big]&=2\chi'\partial_x^2+(\chi'+\chi'')\partial_x+2z_0sV'\chi-2s^2VV'\chi.
	\end{align*}
	By pseudodifferential calculus, we get:
	\begin{align*}
		p_\ell(z_0,s)^{-1}&\in\Psi^{-2,0},&\big[P_\ell,\chi\big]&\in\Psi^{1,-\infty},&\big[p_\ell(z_0,s),\big[P_\ell,\chi\big]\big]&\in\Psi^{2,-\infty}.
	\end{align*}
	On the other hand, we have
	\begin{align}
	\label{Commutator P_ell and p_ell(z,s)}
		p_\ell(z_0,s)p_\ell(z,s)^{-1}&=\big(p_\ell(z,s)+(z^2-z_0^2)-2(z-z_0)sV)\big)p_\ell(z,s)^{-1}\nonumber\\
		&=1+p_\ell(z,s)^{-1}(z^2-z_0^2)-2(z-z_0)sVp_\ell(z,s)^{-1}\nonumber\\
		&=p_\ell(z,s)^{-1}(P_{\ell}-(z_0^2-2zsV+s^2V^2))-2(z-z_0)sVp_\ell(z,s)^{-1}.
	\end{align}
	Using the identity
	\begin{align*}
		\chi p_\ell(z,s)^{-1}P_{\ell}\chi&=\chi p_\ell(z,s)^{-1}\big[P_{\ell},\chi\big]+\chi p_\ell(z,s)^{-1}\chi P_{\ell}
	\end{align*}
	and the uniform bound in $\ell$
	\begin{align}
	\label{chi derivative}
		\|\chi'u'\|_{_{L^2}}&\lesssim\|\chi_1u\|_{_{L^2}}+\|\chi_2u''\|_{_{L^2}}&\chi_j\in\mathcal{C}^{\infty}_{\mathrm{c}}(\mathbb{R},\mathbb{R}),\ \mathrm{Supp\,} \chi_j=\mathrm{Supp\,}\chi,
	\end{align}
	we obtain from \eqref{Useful identity}
	\begin{align}
	\label{Estimate on A}
		\|Au\|_{_{L^2}}&\leq \tilde{C}_\mathrm{a}\langle z\rangle\|\tilde{\chi}p_\ell(z,s)^{-1}\tilde{\chi}\|_{_{L^2\to L^2}}\|u\|_{_{L^2}}
	\end{align}
	where the constant $\tilde{C}_\mathrm{a}$ only depends on $z_0,s,V,V',\chi,\chi',\chi'',\chi_1$ and $\chi_2$.
	
	We now turn to $B$. Using again \eqref{Commutator P_ell and p_ell(z,s)}, we see that
	\begin{align*}
		\|\chi P_\ell p_\ell(z,s)^{-1}\chi u\|_{_{L^2}}&\leq\|\chi p_\ell(z_0,s)p_\ell(z,s)^{-1}\chi u\|_{_{L^2}}\\
		&+\|\chi(z_0^2-2z_0sV+s^2V^2)p_\ell(z,s)^{-1}\chi u\|_{_{L^2}}\\
		&\leq\|\chi p_\ell(z,s)^{-1}(P_{\ell}-(z_0^2-2zsV+s^2V^2))\chi u\|_{_{L^2}}+2\big(|z|+|z_0|\big)|s|\|V\|_{_{L^\infty}}\|\chi p_\ell(z,s)^{-1}\chi u\|_{_{L^2}}\\
		&+\big(|z_0|^2+2|z_0||s|\|V\|_{_{L^\infty}}+s^2\|V\|_{_{L^\infty}}^2\big)\|\chi p_\ell(z,s)^{-1}\chi u\|_{_{L^2}}\\
		&\leq\|\chi p_\ell(z,s)^{-1}P_{\ell}\chi v\|_{_{L^2}}\\
		&+\underbrace{2\big\langle|z_0|+2|s|\|V\|_{_{L^\infty}}\big\rangle^2}_{=:\tilde{\tilde{C}}_\mathrm{a}}\langle z\rangle\|\chi p_\ell(z,s)^{-1}\chi\|_{_{L^2\to L^2}} \|u\|_{_{L^2}}.
	\end{align*}
	Commuting $P_{\ell}$ with $\chi$ and using \eqref{chi derivative}, we get \eqref{Estimate to show for C_a} with $C_\mathrm{a}=\max\big\{\tilde{C}_\mathrm{a},1+\tilde{\tilde{C}}_\mathrm{a}\big\}$.
	\paragraph{Estimate \eqref{Estimate C_b}.} Let us define
	\begin{align*}
		\Lambda_\mathrm{b}(\theta)&=\langle z\rangle^{-2\theta}P_\ell^{\theta}\chi p_\ell(z,s)^{-1}\chi&\theta\in[0,1]+\mathrm{i}\mathbb{R}.
	\end{align*}
	$\Lambda_\mathrm{b}$ is analytic from $[0,1]+\mathrm{i}\mathbb{R}$ to $\mathcal{L}(L^2,L^2)$. As the above estimate, it is sufficient to show a bound on $\Lambda_\mathrm{b}(1)$, the imaginary part of $\theta$ playing no role and the case $\theta=0$ being trivial. We get \eqref{Estimate C_b} if we show that
	\begin{align}
	\label{Estimate to show for C_b}
		\|P_\ell\chi p_\ell(z,s)^{-1}\chi u\|_{_{L^2}}&\leq C_\mathrm{b}\langle z\rangle^{2}\|\tilde{\chi}p_\ell(z,s)^{-1}\tilde{\chi}\|_{_{L^2}}\|u\|_{_{L^2}}\forall u\in L^2
	\end{align}
	for some $C_\mathrm{b}>0$. Using the identity \eqref{Useful identity} and the estimate \eqref{Estimate on A}, we obtain
	\begin{align*}
		\|P_\ell\chi p_\ell(z,s)^{-1}\chi u\|_{_{L^2\to L^2}}&\leq \tilde{C}_\mathrm{a}\langle z\rangle\|\tilde{\chi}p_\ell(z,s)^{-1}\tilde{\chi}\|_{_{L^2\to L^2}}\|u\|_{_{L^2}}+\|\chi p_\ell(z,s)^{-1}P_{\ell}\chi u\|_{_{L^2}}
	\end{align*}
	but this time we ask for the $L^2$ norm of $u$. Hence, we use that
	\begin{align*}
		p_\ell(z,s)^{-1}P_{\ell}&=1+p_\ell(z,s)^{-1}(z-sV)^2
	\end{align*}
	which yields \eqref{Estimate to show for C_b} with $C_\mathrm{b}=\max\big\{\tilde{C}_\mathrm{a},2\big\langle s\|V\|_{_{L^\infty}}\big\rangle^2\big\}$.
	\paragraph{Estimate \eqref{Estimate C_c}.} Let us define
	\begin{align*}
		\Lambda_\mathrm{c}(\theta)&:=\langle z\rangle^{2(\theta-1)}\chi(1+(z-(sV)^{2(1-\theta)})p_\ell(z,s)^{-1}(z-2^{2\theta-1}sV))\chi P_\ell^{-\theta}.
	\end{align*}
	Once again, $\Lambda_\mathrm{c}$ is analytic from $[0,1]+\mathrm{i}\mathbb{R}$ to $\mathcal{L}(L^2,L^2)$ and (dropping the imaginary part)
	\begin{align*}
		\|\Lambda_\mathrm{c}(0)\|_{_{L^2\to L^2}}&\leq(2+|s|\|V\|_{_{L^\infty}})^3\|\chi p_{\ell}(z,s)^{-1}\chi\|_{_{L^2\to L^2}}.
	\end{align*}
	We then get a bound on $\Lambda_{c}(1)$: we prove
	\begin{align*}
		\|\chi(1+zp_\ell(z,s)^{-1}(z-2sV))\chi u\|_{_{L^2}}&\leq C_\mathrm{c}\|\tilde{\chi}p_\ell(z,s)^{-1}\tilde{\chi}\|_{_{L^2\to L^2}}\|P_\ell u\|_{_{L^2}}&\forall u\in L^2.
	\end{align*}
	We have
	\begin{align*}
		\|\chi(1+p_\ell(z,s)^{-1}z(z-2sV))\chi u\|_{_{L^2}}&\leq\|\chi(1+p_\ell(z,s)^{-1}(z-sV)^2\chi)u\|_{_{L^2}}+\|\chi p_\ell(z,s)^{-1}s^2V^2\chi u\|_{_{L^2}}
	\end{align*}
	and
	\begin{align*}
		\chi(1+p_\ell(z,s)^{-1}(z-sV)^2)\chi&=\chi p_\ell(z,s)^{-1}P_{\ell}\chi.
	\end{align*}
	Commuting $P_{\ell}$ with $\chi$ and using \eqref{chi derivative} gives us
	\begin{align*}
		\|\chi p_\ell(z,s)^{-1}z(z-2sV)\chi v\|_{_{L^2}}&\leq C_\mathrm{c}\|\tilde{\chi}p_\ell(z,s)^{-1}\tilde{\chi}\|_{_{L^2\to L^2}}\|\langle P_\ell\rangle v\|_{_{L^2}}
	\end{align*}
	with $C_\mathrm{c}=\max\big\{(1+|s|\|V\|_{_{L^\infty}})^3,\|\chi\|_{_{L^\infty}},\|\chi_1\|_{_{L^\infty}}+\|\chi_2\|_{_{L^\infty}}\big\}$.
\end{proof}
The proof is now the same as in \cite[§3.2]{BoHa08}. For $\nu>0$ fixed and for $\ell\in\mathbb{N}$, we define $L_{\nu}^{2}(\mathbb{R},\dot{\mathcal{E}}_\ell)$ as the class of functions $t\mapsto v\left(t\right)$ with values in $\dot{\mathcal{E}}_\ell$ such that $t\mapsto\mathrm{e}^{-\nu t}v(t)\in L^{2}(\mathbb{R},\dot{\mathcal{E}}_\ell)$. For $u\in\dot{\mathcal{E}}_\ell$, the componentwise defined function
	\begin{align*}
	v\left(t\right)&=
	\begin{cases}
	\mathrm{e}^{-\mathrm{i}t\hat{K}_\ell}u&\text{for $t\geq0$}\\
	0&\text{for $t<0$}
	\end{cases}
	\end{align*}
	is in $L_{\nu}^{2}(\mathbb{R},\dot{\mathcal{E}}_\ell)$ if $\nu$ is sufficiently large and thus
	\begin{align}
	\label{Laplace}
		\tilde{v}(z)&=\int_{0}^{+\infty}\mathrm{e}^{\mathrm{i}zt}v(t)\mathrm{d}t
	\end{align}
	is well-defined as soon as $\Im z\geq\nu$. For all $t\geq0$, we have the inversion formula
	\begin{align}
	\label{Inversion Laplace}
		v(t)&=\frac{1}{2\pi}\int_{-\infty+\mathrm{i}\nu}^{+\infty+\mathrm{i}\nu}\mathrm{e}^{-\mathrm{i}zt}\tilde{v}(z)\mathrm{d}z
	\end{align}
	that is
	\begin{align*}
		\mathrm{e}^{-\mathrm{i}t\hat{K}_{\ell}}u=\frac{1}{2\pi\mathrm{i}}\int_{-\infty+\mathrm{i}\nu}^{+\infty+\mathrm{i}\nu}\mathrm{e}^{-\mathrm{i}zt}(\hat{K}_\ell-z)^{-1}u\mathrm{d}z
	\end{align*}
	in the $L^{2}_{\nu}(\mathbb{R},\dot{\mathcal{E}}_\ell)$ sense. We then use the following result:
	\begin{lemma}[Lemma 3.2 in \cite{BoHa08}]
	\label{Lemma 3.2 in BoHa}
		Let $N\in\mathbb{N}$, $\chi\in\mathcal{C}^\infty_{\mathrm{c}}(\mathbb{R},\mathbb{R})$ and define for all $j\in\mathbb{N}$ the spaces $\dot{\mathcal{E}_\ell}^{-j}:=(\hat{K}_{\ell}-\mathrm{i})^{j}\dot{\mathcal{E}}_{\ell}$. Then for all $k\in\{0,\ldots,N\}$, there exist bounded operators $B_j\in\mathcal{L}(\dot{\mathcal{E}}_{\ell}^{-k},\dot{\mathcal{E}}_{\ell}^{-k-j})$ and $B\in\mathcal{L}(\dot{\mathcal{E}}_{\ell}^{-k},\dot{\mathcal{E}}_{\ell}^{-k-N-1})$ such that
		\begin{align*}
			\hat{R}_{\chi,\ell}(z)&=\sum_{j=0}^{N}\frac{B_j}{(z-\mathrm{i}(\nu+1))^{j+1}}+\frac{B\hat{R}_{\tilde{\chi},\ell}(z)\chi}{(z-\mathrm{i}(\nu+1))^{N+1}}
		\end{align*}
		for some $\tilde{\chi}\in\mathcal{C}^\infty_{\mathrm{c}}(\mathbb{R},\mathbb{R})$ satisfying $\tilde{\chi}\chi=\chi$.
	\end{lemma}
	Now define
	\begin{align}
	\label{tildeR}
		\tilde{R}_{\chi,\ell}(z)&:=\hat{R}_{\chi,\ell}(z)-\sum_{0\leq j\leq 1}\frac{B_j}{(z-\mathrm{i}(\nu+1))^{j+1}}
	\end{align}
	with $B_j\in\mathcal{L}(\dot{\mathcal{E}}_{\ell},\dot{\mathcal{E}}_{\ell}^{-j})$ as in Lemma \ref{Lemma 3.2 in BoHa}; we thus have\footnote{The purpose of Lemma \ref{Lemma 3.2 in BoHa} is to provide us with integrability in $z$ at the prize of using the weaker spaces $\dot{\mathcal{E}}_{\ell}^{-2}$. The task then consists in showing that all the terms in $\dot{\mathcal{E}}_{\ell}^{-2}$ vanish after deformation of contours and the remaining terms are in $\dot{\mathcal{E}}_{\ell}$.}
	\begin{align}
	\label{Estimate resolvent modified}
		\|\tilde{R}_{\chi,\ell}(z)\|_{_{\mathcal{L}(\dot{\mathcal{E}}_{\ell},\dot{\mathcal{E}}_{\ell}^{-2})}}&\lesssim\langle z\rangle^{-2}\|\hat{R}_{\chi,\ell}(z)\|_{_{\mathcal{L}(\dot{\mathcal{E}}_{\ell},\dot{\mathcal{E}}_{\ell})}}.
	\end{align}
	We can show that
	\begin{align}
	\label{Null integral}
		\int_{-\infty+\mathrm{i}\nu}^{+\infty+\mathrm{i}\nu}\frac{B_j}{(z-\mathrm{i}(\nu+1))^{j+1}}\mathrm{d}z&=0
	\end{align}
	for all $j\in\mathbb{N}$ using a contour deformation (integrate first along the square with apexes $\pm C_0+\mathrm{i}\nu$, $\pm C_0-\mathrm{i}\mu$ and let $\mu\to-\infty$ then $C_0\to+\infty$). We thus obtain
	\begin{align}
	\label{Contour deformation}
		\chi\mathrm{e}^{-\mathrm{i}t\hat{K}_{\ell}}\chi u=\frac{1}{2\pi\mathrm{i}}\int_{-\infty+\mathrm{i}\nu}^{+\infty+\mathrm{i}\nu}\mathrm{e}^{-\mathrm{i}zt}\tilde{R}_{\chi,\ell}(z)u\mathrm{d}z
	\end{align}
	and the integral absolutely converges in $\mathcal{L}(\dot{\mathcal{E}}_\ell,\dot{\mathcal{E}}_\ell^{-2})$.
	\begin{figure}[!h]
		\centering
		\captionsetup{justification=centering,margin=1.8cm}
		\begin{center}
			\begin{tikzpicture}[scale=0.8]
			\filldraw[draw=black,dashed,fill=gray!15](-4,1)--(4,1)--(4,-2.552122)--plot[domain=4:1.75](\x,{-ln(1+(\x-1.75)^2)-0.75})--(-1.75,-0.75)--plot[domain=1.75:4](-\x,{-ln(1+(\x-1.75)^2)-0.75})--(-4,1);
			\draw[->][thick](0,-2)--(0,2)node[above]{$\mathrm{i}\mathbb{R}$};
			\draw[->][thick](-5,0)--(5,0)node[right]{$\mathbb{R}$};
			
			\draw[thick](-4,-0.1)--(-4,0.1)node[anchor=south east]{$-K$};
			\draw[thick](-1.75,-0.1)--(-1.75,0.1)node[anchor=south]{$-R\ell$};
			\draw[thick](4,-0.1)--(4,0.1)node[anchor=south west]{$K$};
			\draw[thick](1.75,-0.1)--(1.75,0.1)node[anchor=south]{$R\ell$};
			\draw[thick](-0.1,-0.75)--(0.1,-0.75)node[anchor=north east]{$-\mathrm{i}\mu$};
			\draw[thick](-0.1,1)--(0.1,1)node[anchor=south west]{$\mathrm{i}\nu$};
			
			\draw[thick](0.5,-0.75)node[anchor=south west]{$\Gamma_{1}$};
			\draw[thick](2.5,-1.0)node[anchor=west]{$\Gamma_{2}$};
			\draw[thick](4,-0.65)node[anchor=north west]{$\Gamma_{3}$};
			\draw[thick](-2.5,-1.0)node[anchor=east]{$\Gamma_{4}$};
			\draw[thick](-4,-0.65)node[anchor=north east]{$\Gamma_{5}$};

			\draw[->](-2,-2.5)--(-2.8,-1.55);
			\draw[->](2,-2.5)--(2.8,-1.55);
			\draw(0,-2.5)node[fill=white]{$\left\{\Im z=-\ln\left\langle\left|\Re z\right|-R\ell\right\rangle-\mu\right\}$};
			\end{tikzpicture}
			\caption{\label{Integration contour expansion}The contour used for the derivation of the resonance expansion.}
		\end{center}
	\end{figure}
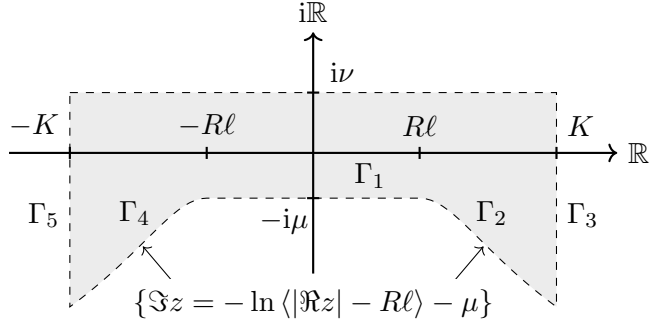
	We then integrate $\mathrm{e}^{-\mathrm{i}zt}\tilde{R}_{\chi,\ell}\,u$ over the (positively oriented) contour described in Figure \ref{Integration contour expansion} defined for $K,\mu>0$. Setting
	\begin{align*}
		I_{j}&:=\frac{1}{2\pi\mathrm{i}}\int_{\Gamma_{j}}\mathrm{e}^{-\mathrm{i}zt}\tilde{R}_{\chi,\ell}(z)u\mathrm{d}z,
	\end{align*}
	one obtains by the residue theorem:
	\begin{align}
	\label{Residue}
		\frac{1}{2\pi\mathrm{i}}\int_{-K+\mathrm{i}\nu}^{K+\mathrm{i}\nu}\mathrm{e}^{-\mathrm{i}zt}\tilde{R}_{\chi,\ell}(z)u\mathrm{d}z&=\sum_{\substack{z_j\in\mathrm{Res}(p_{\ell})\\\Im z_j>-\mu}}\sum_{k=0}^{m(z_j)}\mathrm{e}^{-\mathrm{i}z_jt}t^k\Pi_{j,k}^\chi u+\sum_{1\leq j\leq 5}I_{j}.
	\end{align}
	Using the estimate \eqref{Estimate on the cut-off inverse of the quadratic pencil in the zone IV formula} in the Zone IV as well as Proposition \ref{Estimee pour la resolvente tronquee de Lell dans Emodell} and \eqref{Estimate resolvent modified} above, we compute for $t$ large enough:
	\begin{align*}
		\|I_{3}\|_{_{\dot{\mathcal{E}}_{\ell}^{-2}}}&\lesssim\int_{K-\mathrm{i}\ln\langle K\rangle}^{K+\mathrm{i}\nu}\|\mathrm{e}^{-\mathrm{i}\lambda t}\tilde{R}_{\chi,\ell}(\lambda)u\|_{_{\dot{\mathcal{E}}_{\ell}^{-2}}}\mathrm{d}\lambda\\
		&\lesssim\|u\|_{_{\dot{\mathcal{E}}_{\ell}}}\int_{\ln\langle K\rangle}^{\nu}\langle z\rangle^{-2}\mathrm{e}^{\lambda t+C|\lambda|}\mathrm{d}\lambda\\
		&\lesssim\langle K\rangle^{-2}\frac{\mathrm{e}^{\nu t}}{t}\|u\|_{_{\dot{\mathcal{E}}_{\ell}}}.
	\end{align*}
	We let $K\to+\infty$: the integrals $I_{3}$ and $I_{5}$ then vanish. We still denote by $I_{2}$ and $I_{4}$ the integrals over $\Gamma_{2}$ and $\Gamma_{4}$ which now go to infinity. As for \eqref{Null integral}, we can show that
	\begin{align}
	\label{Null integral 2}
		\int_{\Gamma_{4}\cup\Gamma_{1}\cup\Gamma_{2}}\frac{B_j}{(z-\mathrm{i}(\nu+1))^{j+1}}\mathrm{d}z&=0.
	\end{align}
	Now, using \eqref{final bound in zone III 001} for $I_{1}$ and \eqref{Estimate on the cut-off inverse of the quadratic pencil in the zone IV formula} for $I_{2}$ and $I_{4}$, we get for $t$ large enough:
	\begin{align*}
		\|I_{1}\|_{_{\dot{\mathcal{E}}_{\ell}}}&\lesssim\int_{-R\ell}^{R\ell}\|\mathrm{e}^{-\mu t}\hat{R}_{\chi,\ell}(r-\mathrm{i}\mu)u\|_{_{\dot{\mathcal{E}}_{\ell}}}\mathrm{d}r\\
		&\lesssim\mathrm{e}^{-\mu t}\|u\|_{_{\dot{\mathcal{E}}_{\ell}}}\int_{-R\ell}^{R\ell}\langle r\rangle^{C}\mathrm{d}r\\
		&\lesssim\mathrm{e}^{-\mu t}\ell^{C+1}\|u\|_{_{\dot{\mathcal{E}}_{\ell}}},
	\end{align*}
	\begin{align*}
		\|I_{2}\|_{_{\dot{\mathcal{E}}_{\ell}}}&\lesssim\int_{0}^{+\infty}\left\|\mathrm{e}^{-\mathrm{i}\big(R\ell+\lambda-\mathrm{i}(\mu+\ln\langle\lambda\rangle)\big)t}\hat{R}_{\chi,\ell}\big(R\ell+\lambda-\mathrm{i}(\mu+\ln\langle\lambda\rangle)\big)u\right\|_{_{\dot{\mathcal{E}}_{\ell}}}\mathrm{d}\lambda\\
		&\lesssim\mathrm{e}^{-\mu t}\|u\|_{_{\dot{\mathcal{E}}_{\ell}}}\int_{0}^{+\infty}\mathrm{e}^{-\ln\langle\lambda\rangle t+C(\ln\langle\lambda\rangle+\mu)}\mathrm{d}\lambda\\
		&\lesssim\mathrm{e}^{-\mu t}\|u\|_{_{\dot{\mathcal{E}}_{\ell}}}.
	\end{align*}
	All these estimates hold in $\dot{\mathcal{E}}_{\ell}$, hence we have established part (i) of Theorem \ref{Theoreme principal} with $N=(C+1)/2$.
	
	Let us turn to part (ii). For $\mu<\varepsilon_{0}$ with $\varepsilon_{0}$ as in part 2. of Theorem \ref{No resonances in a strip near the real axis for s small}, we know that there is no resonance in formula \ref{Residue}. If $\ell<\mathrm{e}^{\varepsilon' t}$ for some $\varepsilon'>0$, then
	\begin{align*}
		\|I_{1}\|_{_{\dot{\mathcal{E}}_{\ell}}}&\lesssim\mathrm{e}^{((C+1)\varepsilon'-\mu)t}\|u\|_{_{\dot{\mathcal{E}}_{\ell}}}.
	\end{align*}
	Otherwise, if $\ell\geq\mathrm{e}^{\varepsilon' t}$, then the exponential decay of the local energy as well as the hypotheses on $g$ imply together:
	\begin{align*}
		\|\chi\mathrm{e}^{-\mathrm{i}t\hat{K}_{\ell}}\chi u\|_{_{\dot{\mathcal{E}}_{\ell}}}&\lesssim 1\lesssim\frac{g(\ell(\ell+1))}{g(\mathrm{e}^{2\varepsilon' t})}.
	\end{align*}
	It remains to take $\varepsilon'$ small enough and $\varepsilon:=\min\{2\varepsilon',\mu-(C+1)\varepsilon'\}$ to conclude the proof.
%
%
%
%
\pagebreak
%
%
%
%
\begin{appendices}
	\section{Analytic extension of the coordinate $r$}
	\label{Proof of analytic extension}
	In this appendix, we prove Proposition \ref{Analytic extension of r} which is analogous to \cite[Prop IV.2]{BaMo93}. Let $r\in\left]r_{-},r_{+}\right[$. By equation \eqref{Regge-Wheeler x(r)}, we have
	\begin{align*}
	\exp\left(-\frac{\Lambda}{3A_\pm r_\pm^2}x\right)&=\prod_{\alpha\in I}\left|\frac{r-r_\alpha}{\mathfrak{r}-r_\alpha}\right|^{\frac{A_\alpha r_\alpha^2}{A_\pm r_\pm^2}}.
	\end{align*}
	Call the left-hand side $z$ and the right-hand side $g_\pm(r)$. Observe that $g_\pm(r_\pm)=0$. Since $r\mapsto x(r)$ is increasing and analytic, we can apply the Lagrange's inversion theorem (see for example \cite[§2.2]{DeBr81} and references therein) to write
	\begin{align}
	\label{Lagrange inversion series}
	r&=r_\pm+\sum_{\ell=1}^{+\infty}\frac{z^\ell}{\ell!}\left[\frac{\mathrm{d}^{\ell-1}}{\mathrm{d}r^{\ell-1}}\left(\frac{r-r_\pm}{g_\pm(r)}\right)^{\ell}\,\right]_{r=r_\pm}.
	\end{align}
	Let us introduce Kronecker's symbol
	\begin{align*}
	\delta_{\alpha,\pm}:=\begin{cases}
	1&\text{if $\alpha=\pm$}\\
	0&\text{otherwise}
	\end{cases}
	\end{align*}
	and the notation
	\begin{align*}
	B_{\pm,\alpha}:=\frac{A_\alpha r_\alpha^2}{A_\pm r_\pm^2}-\delta_{\alpha,\pm}.
	\end{align*}
	Observe that $B_{-,-}=B_{+,+}=0$. We then have
	\begin{align*}
	\frac{\mathrm{d}^{\ell-1}}{\mathrm{d}r^{\ell-1}}\left(\frac{r-r_\pm}{g_\pm(r)}\right)^{\ell}&=\left(\prod_{\alpha\in I\setminus\{\pm\}}\left|\mathfrak{r}-r_\alpha\right|^{\ell B_{\pm,\alpha}}\right) \frac{\mathrm{d}^{\ell-1}}{\mathrm{d}r^{\ell-1}}\left(\prod_{\alpha\in I\setminus\{\pm\}}\left|r-r_\alpha\right|^{-\ell B_{\pm,\alpha}}\right).
	\end{align*}
	We now fix $\pm=+$ (the conclusion will not be changed if we choose $\pm=-$). Then
	\begin{align*}
	\frac{\mathrm{d}^{\ell-1}}{\mathrm{d}r^{\ell-1}}\left(\prod_{\alpha\in I\setminus\{+\}}\left(r-r_\alpha\right)^{-\ell B_{+,\alpha}}\right)=\sum_{0\leq k_2\leq k_1\leq\ell}C_{\ell,k_1,k_2}\left(\frac{\mathrm{d}^{\ell-k_1}}{\mathrm{d}r^{\ell-k_1}}(r-r_n)^{-\ell B_{+,n}}\right)\times\nonumber\\
	\times\left(\frac{\mathrm{d}^{k_1-k_2}}{\mathrm{d}r^{k_1-k_2}}(r-r_c)^{-\ell B_{+,c}}\right)\left(\frac{\mathrm{d}^{k_2}}{\mathrm{d}r^{k_2}}(r-r_-)^{-\ell B_{+,-}}\right)
	\end{align*}
	where
	\begin{align*}
	C_{\ell,k_1,k_2}&=\begin{pmatrix}
	\ell\\ k_1
	\end{pmatrix}\begin{pmatrix}
	k_1\\ k_2
	\end{pmatrix}.
	\end{align*}
	Direct computation shows that
	\begin{align*}
	\frac{\mathrm{d}^{p}}{\mathrm{d}r^{p}}(r-r_\alpha)^{-\ell B_{+,\alpha}}&=(-1)^{p}(\ell B_{+,\alpha})(\ell B_{+,\alpha}+1)\ldots(\ell B_{+,\alpha}+p-1)(r-r_\alpha)^{-\ell B_{+,\alpha}-p}.
	\end{align*}
	If we let
	\begin{align*}
	K&:=\prod_{\alpha\in I\setminus\{+\}}\left(\mathfrak{r}-r_\alpha\right)^{B_{+,\alpha}},&B_{+}:=\max_{\alpha\in I\setminus\{+\}}\{|B_{+,\alpha}|\},
	\end{align*}
	then it follows that
	\begin{align*}
	\frac{\mathrm{d}^{\ell-1}}{\mathrm{d}r^{\ell-1}}\left(\frac{r-r_+}{g_+(r)}\right)^{\ell}&=K^{\ell}\sum_{0\leq k_2\leq k_1\leq\ell}C_{\ell,k_1,k_2}(-1)^{\ell}\times\nonumber\\
	&\times(\ell B_{+,n})(\ell B_{+,n}+1)\ldots(\ell B_{+,n}+(\ell-k_1)-1)(r-r_n)^{-\ell B_{+,n}-(\ell-k_1)}\times\nonumber\\
	&\times(\ell B_{+,c})(\ell B_{+,c}+1)\ldots(\ell B_{+,c}+(k_1-k_2)-1)(r-r_c)^{-\ell B_{+,c}-(k_1-k_2)}\times\nonumber\\
	&\times(\ell B_{+,-})(\ell B_{+,-}+1)\ldots(\ell B_{+,-}+k_2-1)(r-r_\alpha)^{-\ell B_{+,-}-k_2}
	\end{align*}
	and thus
	\begin{align*}
	\left|\frac{\mathrm{d}^{\ell-1}}{\mathrm{d}r^{\ell-1}}\left(\frac{r-r_+}{g_+(r)}\right)^{\ell}\right|&\leq K^{\ell}\ell^{\ell}(B_{+}+1)^{\ell}\left(\prod_{\alpha\in I\setminus\{+\}}(r_+-r_\alpha)^{-B_{+,\alpha}}\right)^{\ell}\times\nonumber\\
	&\times\sum_{0\leq k_2\leq k_1\leq\ell}C_{\ell,k_1,k_2}(r_+-r_n)^{-(\ell-k_1)}(r_+-r_c)^{-(k_1-k_2)}(r_+-r_-)^{-k_2}\nonumber\\
	&=K^{\ell}\ell^{\ell}(B_{+}+1)^{\ell}\left(\prod_{\alpha\in I\setminus\{+\}}(r_+-r_\alpha)^{-B_{+,\alpha}}\right)^{\ell}\left(\sum_{\alpha\in I\setminus\{+\}}(r_+-r_\alpha)^{-1}\right)^{\ell}\\
	&=\left(K(B_{+}+1)\prod_{\alpha\in I\setminus\{+\}}(r_+-r_\alpha)^{-B_{+,\alpha}}\sum_{\alpha\in I\setminus\{+\}}(r_+-r_\alpha)^{-1}\right)^{\ell}\ell^{\ell}\\
	&=:\tilde{K}^{\ell}\ell^{\ell}.
	\end{align*}	
	Therefore, the convergence of the original series is absolute for $z\in\mathbb{C}$ if
	\begin{align*}
	\frac{(|z|\ell\tilde{K})^{\ell}}{\ell!}&<\ell^{-(1+\varepsilon)}
	\end{align*}
	for any $\varepsilon>0$. Using Stirling approximation $\ell!\sim\sqrt{2\pi}\ell^{\ell+1/2}$ for large values of $\ell$, we see that it is sufficient to have
	\begin{align*}
	\tilde{K}|z|&<\frac{\mathrm{e}^{-(1/2+\varepsilon)\ln\ell/\ell}}{\sqrt{2\pi}^{\ell}}<1.
	\end{align*}
	This condition is fulfilled if
	\begin{align*}
	\Re x&>\frac{3A_+ r_+^2}{\Lambda}\ln\tilde{K}.
	\end{align*}
%
%
%
%
\section{Localization of high frequency resonances}
\label{Microlocalization of high frequency resonances section}

We provide in this section an asymptotic approximation of resonances near the maximal energy $W_0(0)=\max_{x\in\mathbb{R}}\{W_0(x)\}$ as $h\to0$. This a generalization of the main Theorem in \cite{SaZw97} to the case $Q\neq0$. More precisely, we show that the resonances associated to the meromorphic extension of $p(z,s)^{-1}$ are close to the ones associated with the extension of $(P-z^2)^{-1}$, provided that $Q$ is sufficiently small. This is a direct consequence of the fact that the extra term $hsV$ in the semiclassical quadratic pencil is $\mathcal{O}(hs)$.

As in the paragraph \ref{Estimates in the zone III}, we set $h:=(\ell(\ell+1))^{-1/2}$ with $\ell>0$ and consider  $z\in\left[\ell/R,R\ell \right]+\mathrm{i}\left[-C_{0},C_{0}\right]$. We then define the spectral parameter $\lambda:=h^2z^2$ and also $\tilde{P}_h$ the semiclassical operator associated to $P_\ell$. Recall also that $\mathfrak{r}=\frac{3M}{2}\left(1+\sqrt{1-\frac{8Q^2}{9M^2}}\right)$ is the radius of the photon sphere and $W_0(0)=F(\mathfrak{r})/\mathfrak{r}^2$ with our definition of the Regge-Wheeler coordinate $x$ (see \eqref{Regge-Wheeler x(r)}).
\begin{theorem}
	\label{Microlocalization of high frequency resonances, microlocal version}
	Let
	\begin{align*}
		\Gamma_0(h)&:=\left\{W_0(0)+h\left(2\sqrt{W_0(0)}sV(0)+\mathrm{i}^{-1}\sqrt{W_0''(0)/2}\left(k+\frac{1}{2}\right)\right)\mid k\in\mathbb{N}\right\}.
	\end{align*}
	For all	$C_0>0$ such that $\partial D(W_0(0),C_0h)\cap\Gamma_0(h)=\emptyset$, there is a bijection $b\equiv b(h)$ from $\Gamma_0(h)$ onto the set of resonances of $\tilde{P}_h$ in $D(W_0(0),C_0h)$ (counted with their natural multiplicity) such that
	\begin{align*}
		b(h)(\mu)-\mu&=o_{h\to0}(h)&\mathrm{uniformly\ for\ }\mu\in\Gamma_0(h).
	\end{align*}
\end{theorem}
\begin{proof}
	This is a direct application of the results of Sá Barreto-Zworski \cite{SaZw97} which are based on the work of Sj\"{o}strand \cite{Sj87} (see Theorem 0.1 therein), the latter dealing with resonances generated by non-degenerate critical points when the trapping set is reduced to a single point (the difference for us is $W_0(0)\neq0$).
	
	We recall that in the zone III the symbol of the semiclassical quadratic pencil is the function $(x,\xi)\mapsto\xi^2+W_0(x)+h^2W_1(x)-(\sqrt{\lambda}-hsV(x))^2=:p(x,\xi)-\lambda$. We also recall the hypothesis in \cite{Sj87} for the case of a Schr\"{o}dinger operator of the form (0.1) in the reference:
	\begin{itemize}
	 \item The trapping set is reduced to the point $\{(0,0)\}$ ((0.3) in \cite{Sj87}),
	 \item $0$ is a non-degenerate critical point ((0.4) in \cite{Sj87}, which implies in the Schr\"{o}dinger case the more general assumptions (0.7) and (0.9) in the reference).
	\end{itemize}
	Although the symbol $p$ depends on $\lambda$, its principal part $p_0$ and subprincipal part $p_{-1}$ do not: indeed, for $\lambda\in D\left(W_0(0),C_0 h\right)$ with $C_0>0$, we can write when $h\ll1$
	\begin{align*}
	  p(x,\xi)&=\underbrace{\xi^2+W_0(x)}_{=p_0(x,\xi)}+h\underbrace{2\sqrt{W_0(0)}sV(x)}_{=p_{-1}(x,\xi)}+\text{ lower order terms in $h$}.
	\end{align*}
	This is enough to apply \cite[Thm. 0.1]{Sj87}: using formula (0.14) in the reference, we get the result for the set
	\begin{align*}
	    \left\{p_0(0,0)+h\left(p_{-1}(0,0)+\mathrm{i}^{-1}\sqrt{W_0''(0)/2}\left(k+\frac{1}{2}\right)\right)\mid k\in\mathbb{N}\right\}
	\end{align*}
	which is $\Gamma_0(h)$.
\end{proof}
Approximation of high frequency resonances $\Gamma_0(h)\ni z^2=\lambda/h^2$ is obtained as in \cite{SaZw97}, by taking the square root of any element of $\Gamma_0(h)$ and using Taylor expansion for $0<h\ll1$ (corresponding to $\ell\gg0$) as well as symmetry with respect to the imaginary axis (for the choice of the sign of the square root). In our setting, we obtain the set $\Gamma$ of Theorem \ref{Microlocalization of high frequency resonances}.
\begin{remark}
	\begin{enumerate}
		\item Let $\Gamma_{\mathrm{DSS}}$ be the set of pseudo-poles in the De Sitter-Schwarzschild case (see the Theorem at the end of \cite{SaZw97}). Then $\Gamma_{\mathrm{DSS}}$ is the limit of $\Gamma$ as $Q\to0$ in the sense of the sets, i.e. for all $z\in\Gamma$, there exists $z_0\in\Gamma_{\mathrm{DSS}}$ such that $z\to z_0$ as $Q\to0$.
		\item The pseudo-poles in the charged case are shifted with respect to the uncharged case. If the charges of the Klein-Gordon field and the black hole have the same sign (that is if $qQ>0$), then all the pseudo-poles go to infinity with a real part which never vanishes. However, if the charges have opposite sign ($qQ<0$), then all the pseudo-poles real part cancels precisely when $qQ=-(k+1/2)\sqrt{F(\mathfrak{r})}$, $k\in\mathbb{N}\setminus\{0\}$, before going to infinity. Notice that no pseudo-pole goes to $\mathbb{C}^+$ as $|s|\to+\infty$. 
		\item We can provide a physical interpretation of the set of pseudo-poles. First observe that $\sqrt{F(\mathfrak{r})}/\mathfrak{r}$ is nothing but the inverse of the impact parameter $b=|E/L|$ of trapped null geodesics with energy $E$ and angular momentum $L$. Theorem \ref{Microlocalization of high frequency resonances} shows that resonances near the real line in the zone III are $qQ$-dependent multiples of this quantity: they thus correspond to impact parameters of trapped photons with high energy and angular momentum. 
		\item Observe that in Newtonian mechanics, the electromagnetism and gravitation do not interact with chargeless and massless photons. As a consequence, photons are not deviated and only ones with impact parameter $|b|\leq r_-$ can ``fall`` in the black hole. Hence, high frequency resonances in zone III are expected to be multiple of $r_-^{-1}$. As $r_-\to0$, all resonances go to infinity: the trajectory are now classical straight lines as there is no obstacle anymore.
	\end{enumerate}
	\begin{figure}[!h]
		\centering
		\captionsetup{justification=centering,margin=1.8cm}
		\begin{center}
			\begin{tikzpicture}[scale=1]
			\draw(-3,0);
			\draw [cyan,directed,domain=45:30,variable=\t,smooth,samples=100]
			plot ({\t r}: {0.9455*exp((\t-30)/15)});
			\draw[cyan](1.35,2.185)--(1.1,2.39);
			\draw[gray,dashed](0.85,2.6)--(3.45,0.48);
			
			\draw[](0.8,1.25)node[anchor=east]{$b$};
						
			\filldraw[fill=black] (0,0) circle (18pt);
			\draw[thick,dotted,gray,<->](0,0)--(1.45,2.05);
			\draw[blue,thick] (0,0) circle (27pt);
			\draw[cyan,directed](5,1.5)--(12,1.5);
			\draw[cyan,directed](5,0.36)--(9,0.36);
			
			\draw[thick,dashed,gray](9,0.636)--(11.1,0.636);
			\draw[thick,dashed,gray](9,0)--(11.1,0);
			\draw[thick,dotted,gray,<->](10.95,0.05)--(10.95,0.586);
			
			\draw[](10.95,0.314)node[anchor=west]{$r_-$};
			
			\filldraw[fill=black] (9,0) circle (18pt);
			\end{tikzpicture}
			\caption{\label{Impact parameter}On the left: a relativistic trapped null geodesic. On the right: classical null geodesic trajectories.}
		\end{center}
	\end{figure}
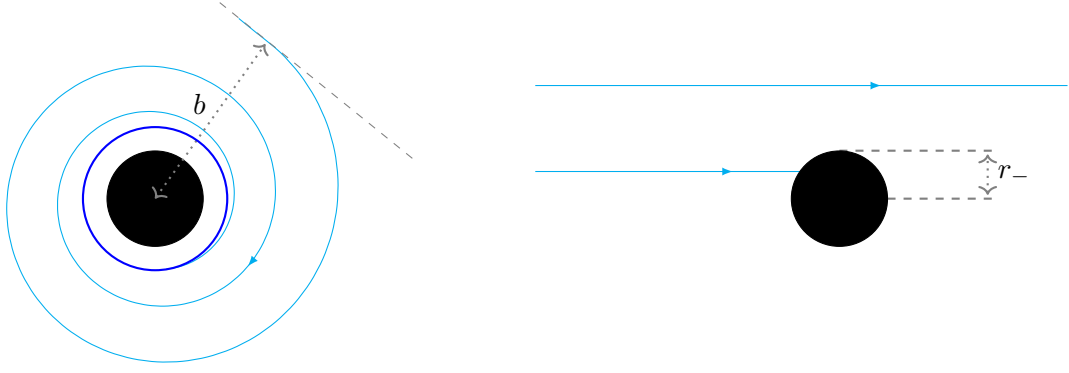
\end{remark}
%
%
%
%
%
%
\section{Abstract Semiclassical Limiting Absorption Principle for a class of Generalized Resolvents}
\label{Limiting absorption principle for the quadratic pencil}

We show in this section an abstract semiclassical limiting absorption principle for perturbed resolvents.
\paragraph{Abstract setting.}
Let $\big(\mathcal{H},\langle\cdot,\cdot\rangle\big)$ be a Hilbert space, $J:=[a,b]\subset\mathbb{R}$, $J_\mu^+:=\{\omega\in\mathbb{C}^+\mid\Re\omega\in J,\,\Im\omega<\mu\}$ for some $\mu>0$ fixed and $h_0>0$. The norm associated to $\langle\cdot,\cdot\rangle$ will be denoted by $\|\cdot\|$. We consider families of self-adjoint operators $P\equiv P(h)$ and $\mathcal{A}\equiv\mathcal{A}(h)$ acting on $\mathcal{H}$ for $0<h<h_0$. We set
\begin{align*}
	L^{\infty}_{\ell\mathrm{oc}}(P):=\big\{A:\mathcal{H}\to\mathcal{H}\text{ linear}\mid\forall\chi\in\mathcal{C}^{\infty}_{\mathrm{c}}(\mathbb{R},\mathbb{R}),\,\forall u\in\mathscr{D}(P),\,\|\chi(P)Au\|<+\infty\big\}
\end{align*}
and $\|.\|_{_{P}}$ will be the operator norm on $\mathcal{B}(\mathscr{D}(P),\mathcal{H})$. We also define the local version of the operator \nolinebreak$P$:
\begin{align*}
	P_{\tau}&:=\tau(P)P\qquad\qquad\forall\tau\in\mathcal{C}^\infty_{\mathrm{c}}(\mathbb{R},\mathbb{R}).
\end{align*}
Let then $f:\mathbb{C}\times L^{\infty}_{\ell\mathrm{oc}}(P)\to L^{\infty}_{\ell\mathrm{oc}}(P)$ satisfying the following continuity type relation near $0_{_{L^{\infty}_{\ell\mathrm{oc}}(P)}}$: there exist $\delta_{_{J,\mu}}:\mathbb{R}_+\to\mathbb{R}$ satisfying $\delta_{_{J,\mu}}(r)\to0$ as $r\to0$ and $\varepsilon_{_{J,\mu}}:L^{\infty}_{\ell\mathrm{oc}}(P)\to L^{\infty}_{\ell\mathrm{oc}}(P)$ such that, for all $(z,A)\in J^+_\mu\times L^{\infty}_{\ell\mathrm{oc}}(P)$ with $\|A\|_{_{P}}$ small,
\begin{align}
\label{Assumption H}
f(z,A)&=z+\delta_{_{J,\mu}}(\|A\|_{_{P}})\,\varepsilon_{_{J,\mu}}(A).\tag{C}
\end{align}
We make the following assumptions:
\begin{align}
\label{Assumption P}
&(P-f(z,hA))^{-1}\text{ exists for all $z\in J_\mu^+$ and $A\in L^{\infty}_{\ell\mathrm{oc}}(P)$ if $h\leq h_0$}\tag{I}\\
\label{Assumption I}
&P\in\mathcal{C}^2(\mathcal{A})\tag{P}\\
\label{Assumption M}
&\mathds{1}_{I}(P)[P,i\mathcal{A}]\mathds{1}_{I}(P)\geq ch\mathds{1}_{I}(P)\qquad\text{\!for some $c>0$ and $J\Subset I:=\left]\alpha,\beta\right[\subset\mathbb{R}$}\tag{M}\\
\label{Assumption A}
&\mathrm{ad}^k_{\chi(P)}(\varepsilon_{_{J,\mu}}(A))\in h^k\mathcal{B}(\mathscr{D}(\mathcal{A}))\qquad\text{for all $k\in\{0,1\}$, $\chi\in\mathcal{C}^{\infty}_{\mathrm{c}}(\mathbb{R},\mathbb{R})$ and $\|A\|_{_P}<c'$ for $c'>0$}.\tag{A}
\end{align}
Recall that $P\in\mathcal{C}^{2}\left(\mathcal{A}\right)$ means for all $z\in\mathbb{C}\setminus\sigma(P)$ that the map
\begin{align*}
\mathbb{R}\ni t\mapsto\mathrm{e}^{\mathrm{i}t\mathcal{A}}(P-z)^{-1}\mathrm{e}^{-\mathrm{i}t\mathcal{A}}
\end{align*}
is $\mathcal{C}^2$ for the strong topology of $L^2$. Recall also that for all linear operators $L_1,L_2$ acting on $\mathcal{H}$, $\mathrm{ad}^{0}_{L_1}(L_2):=L_2$ and $\mathrm{ad}^{k+1}_{L_1}(L_2):=[L_1,\mathrm{ad}^k_{L_1}(L_2)]$. Our goal is to show the following result:
\begin{theorem}
	\label{ASALP}
	Assume hypotheses \eqref{Assumption H}, \eqref{Assumption P}, \eqref{Assumption I}, \eqref{Assumption M} and \eqref{Assumption A}. Then for all $\sigma>1/2$,
	\begin{align}
	\label{ASALP equation}
	\sup_{z\in J^+_\mu}\|\langle\mathcal{A}\rangle^{-\sigma}(P-f(z,hB))^{-1}\langle\mathcal{A}\rangle^{-\sigma}\|&\lesssim h^{-1}.
	\end{align}
\end{theorem}
In the sequel, we will write $R(z,hB):=(P-f(z,hB))^{-1}$ and call it the generalized resolvent (of $P$). Also, since $J$ and $\mu$ are now fixed, we will simply write  $J,\,\delta$ and $\varepsilon$ instead of $J_{\mu},\,\delta_{_{J,\mu}}$ and $\varepsilon_{_{J,\mu}}$.
\paragraph{Preliminary results.}
The purpose of this paragraph is to show preliminary results used to prove Theorem \ref{ASALP}. We first prove an adapted version of \cite[Lem. 2.1]{Ger08} to our situation.
\begin{lemma}
	\label{Lemma 2.1 adapted}
	Let $0\leq\sigma\leq1$, $z\in J^+$ and let $\chi\in\mathcal{C}^{\infty}_{\mathrm{c}}\left(\mathbb{R},\mathbb{R}\right)$. If $h$ is small enough, then $R(z,hB)$ and $\chi(P)$ are bounded on $\mathscr{D}(\langle\mathcal{A}\rangle^{\sigma})$.
\end{lemma}
\begin{proof}
	The result is true for $(P-z)^{-1}$ and $\chi(P)$ by \cite[Lem. 2.1]{Ger08}. Let us show that $R(z,hB)\mathscr{D}(\langle\mathcal{A}\rangle^\sigma)\subset\mathscr{D}(\langle\mathcal{A}\rangle^\sigma)$:
	\begin{align*}
	\|\langle\mathcal{A}\rangle^{\sigma}R(z,hB)\langle\mathcal{A}\rangle^{-\sigma}\|&\leq\|\langle\mathcal{A}\rangle^{\sigma}(P-z)^{-1}\langle\mathcal{A}\rangle^{-\sigma}\|+\|\langle\mathcal{A}\rangle^{\sigma}(R(z,hB)-(P-z)^{-1})\langle\mathcal{A}\rangle^{-\sigma}\|\\
	&\lesssim 1+\|\langle\mathcal{A}\rangle^{\sigma}R(z,hB)(z-f(z,hB))(P-z)^{-1}\langle\mathcal{A}\rangle^{-\sigma}\|
	\end{align*}
	and (using that $\varepsilon(hB)\in\mathcal{B}(\mathscr{D}(\mathcal{A}))$ by Assumption \eqref{Assumption A} for $k=0$)
	\begin{align*}
	&\|\langle\mathcal{A}\rangle^{\sigma}R(z,hB)(z-f(z,hB))(P-z)^{-1}\langle\mathcal{A}\rangle^{-\sigma}\|\\
	&\leq\|\langle\mathcal{A}\rangle^{\sigma}R(z,hB)\langle\mathcal{A}\rangle^{-\sigma}\|\|\langle\mathcal{A}\rangle^{\sigma}(z-f(z,hB))\langle\mathcal{A}\rangle^{-\sigma}\|\|\langle\mathcal{A}\rangle^{\sigma}(P-z)^{-1}\langle\mathcal{A}\rangle^{-\sigma}\|\\
	&\lesssim\delta(h\|B\|_{_{P}})\|\langle\mathcal{A}\rangle^{\sigma}\varepsilon(hB)\langle\mathcal{A}\rangle^{-\sigma}\|\|\langle\mathcal{A}\rangle^{\sigma}R(z,hB)\langle\mathcal{A}\rangle^{-\sigma}\|.
	\end{align*}
	We then use the uniformity in assumption \eqref{Assumption A} for $k=0$ to write for $h$ very small
	\begin{align*}
		\delta(h\|B\|_{_{P}})\|\langle\mathcal{A}\rangle^{\sigma}\varepsilon(hB)\langle\mathcal{A}\rangle^{-\sigma}\|\|\langle\mathcal{A}\rangle^{\sigma}R(z,hB)\langle\mathcal{A}\rangle^{-\sigma}\|&<\frac{1}{2}\|\langle\mathcal{A}\rangle^{\sigma}R(z,hB)\langle\mathcal{A}\rangle^{-\sigma}\|.
	\end{align*}
	The proof is complete.
\end{proof}
\begin{corollary}
	\label{Corollary Lemma 2.1 adapted}
	Let $0\leq\sigma\leq1$, $z\in J^+$ and $\tau,\chi\in\mathcal{C}^{\infty}_{\mathrm{c}}\left(\mathbb{R},[0,1]\right)$ such that $\chi\equiv1$ on $I$ and $\tau\chi=\chi$. If $h$ is small enough, then $(P_\tau-f(z,hB))\chi(P)$, $(P_\tau-f(z,hB))\chi(P)(P+\mathrm{i})^{-1}$ and $(P-f(z,hB))(P+\mathrm{i})^{-1}$ preserve $\mathscr{D}(\langle\mathcal{A}\rangle^{\sigma})$.
\end{corollary}
\begin{proof}
	We have
	\begin{align*}
	\langle\mathcal{A}\rangle^{\sigma}(P_\tau-f(z,hB))\chi(P)\langle\mathcal{A}\rangle^{-\sigma}&=\langle\mathcal{A}\rangle^{\sigma}(P_\tau-z)\chi(P)\langle\mathcal{A}\rangle^{-\sigma}\\
	&+\langle\mathcal{A}\rangle^{\sigma}(z-f(z,hB))\langle\mathcal{A}\rangle^{-\sigma}\langle\mathcal{A}\rangle^{\sigma}\chi(P)\langle\mathcal{A}\rangle^{-\sigma}
	\end{align*}
	which is bounded by assumption \eqref{Assumption A} for $k=0$, Lemma \ref{Lemma 2.1 adapted} and the fact that $P_\tau\chi(P)=\varphi(P)$ with $\varphi\in\mathcal{C}^{\infty}_{\mathrm{c}}(\mathbb{R},\mathbb{R})$ by functional calculus. Next, \cite[Lem. 2.1]{Ger08} implies that $(P+\mathrm{i})^{-1}$ preserves $\mathscr{D}(\mathcal{A})$, so we can write
	\begin{align*}
	\langle\mathcal{A}\rangle^{\sigma}(P_\tau-f(z,hB))\chi(P)(P+\mathrm{i})^{-1}\langle\mathcal{A}\rangle^{-\sigma}&=\langle\mathcal{A}\rangle^{\sigma}(P_\tau-f(z,hB))\chi(P)\langle\mathcal{A}\rangle^{-\sigma}\langle\mathcal{A}\rangle^{\sigma}(P+\mathrm{i})^{-1}\langle\mathcal{A}\rangle^{-\sigma}
	\end{align*}
	which is clearly bounded thanks to the above computation. Finally,
	\begin{align*}
	\langle\mathcal{A}\rangle^{\sigma}(P-f(z,hB))(P+\mathrm{i})^{-1}\langle\mathcal{A}\rangle^{-\sigma}&=\langle\mathcal{A}\rangle^{\sigma}(P+\mathrm{i}-\mathrm{i}-z+z-f(z,hB))(P+\mathrm{i})^{-1}\langle\mathcal{A}\rangle^{-\sigma}\\
	&=\mathrm{Id}-(\mathrm{i}+z)\langle\mathcal{A}\rangle^{\sigma}(P+\mathrm{i})^{-1}\langle\mathcal{A}\rangle^{-\sigma}\\
	&+\langle\mathcal{A}\rangle^{\sigma}(z-f(z,hB))\langle\mathcal{A}\rangle^{-\sigma}\langle\mathcal{A}\rangle^{\sigma}(P+\mathrm{i})^{-1}\langle\mathcal{A}\rangle^{-\sigma}
	\end{align*}
	and we again use \cite[Lem. 2.1]{Ger08} and assumption \eqref{Assumption A} for $k=0$.
\end{proof}
The next result is an adaptation of \cite[Lem. 3.1]{Ger08} to our setting.
\begin{lemma}
	\label{Preliminary estimates for the limiting absorption principle in the zone IB}
	Let $0<\sigma\leq1$ and let $\tau,\chi\in\mathcal{C}^{\infty}_{\mathrm{c}}\left(\mathbb{R},[0,1]\right)$ such that $\chi\equiv1$ on $I$ and $\tau\chi=\chi$. Consider the following three statements:\\[1.5mm]
	
	\noindent(i) \ \ \ \ \ \ \ \ \ \ \ \ \ \ \ \ \ \ \ \ \ \ \ \ \ \ \ \ \ \ \ \ \ \ $\displaystyle\sup_{z\in J^{+}}\|\langle\mathcal{A}\rangle^{-\sigma}R(z,hB)\langle\mathcal{A}\rangle^{-\sigma}\|\lesssim h^{-1};$\\[1.5mm]
	
	\noindent(ii) For all $z\in J^{+}$ and all $u\in(P+\mathrm{i})^{-1}\mathscr{D}(\langle\mathcal{A}\rangle^{\sigma})$,
	\begin{align*}
	\|\langle\mathcal{A}\rangle^{-\sigma}u\|&\lesssim h^{-1}\|(P-f(z,hB))u\|+h^{-1}\|\langle\mathcal{A}\rangle^{\sigma}(P-f(z,hB))\chi(P)u\|;
	\end{align*}
	\noindent(iii) For all $z\in J^{+}$ and all $u\in\mathscr{D}(\langle\mathcal{A}\rangle^{\sigma})$,
	\begin{align*}
	\|\langle\mathcal{A}\rangle^{-\sigma}\chi(P)u\|&\lesssim h^{-1}\|\langle\mathcal{A}\rangle^{\sigma}(P_\tau-f(z,hB))\chi(P)u\|.
	\end{align*}
	If $h$ is sufficiently small, then (iii) implies (ii) and (ii) implies (i).
\end{lemma}

\begin{proof}
	First of all, observe that $(i)$ makes sense by Lemma \eqref{Lemma 2.1 adapted}, and $(ii)$, $(iii)$ make sense by Corollary \ref{Corollary Lemma 2.1 adapted} and because $P\chi(P)=P_{\tau}\chi(P)$.
	\begin{itemize}
		\item We show that $(ii)$ implies $(i)$. Let $u\in \mathcal{H}$ and let $v:=R(z,hB)\langle A\rangle^{-\sigma}u$. Then
		\begin{align*}
		w&:=u-\langle\mathcal{A}\rangle^{\sigma}(f(z,hB)-\mathrm{i})R(z,hB)\langle\mathcal{A}\rangle^{-\sigma}u\in\mathcal{H}.
		\end{align*}
		This makes sense if $h$ is small enough because $R(z,hB)$ preserves $\mathscr{D}(\langle\mathcal{A}\rangle^\sigma)$ by Lemma \ref{Lemma 2.1 adapted} and because
		\begin{align*}
		\langle\mathcal{A}\rangle^{\sigma}(f(z,hB)-\mathrm{i})\langle\mathcal{A}\rangle^{-\sigma}&=\langle\mathcal{A}\rangle^{\sigma}(f(z,hB)-z)\langle\mathcal{A}\rangle^{-\sigma}+(z-\mathrm{i})
		\end{align*}
		is bounded by assumption \eqref{Assumption A} for $k=0$. Next, using the resolvent identity $(P+\mathrm{i})^{-1}-R(z,hB)=(P+\mathrm{i})^{-1}(f(z,hB)-\mathrm{i})R(z,hB)$, we see that
		\begin{align*}
		(P+\mathrm{i})^{-1}\langle A\rangle^{-\sigma}w&=\big((P+\mathrm{i})^{-1}-(P+\mathrm{i})^{-1}(f(z,hB)-\mathrm{i})\,R(z,hB)\big)\langle\mathcal{A}\rangle^{-\sigma}u\\
		&=R(z,hB)\langle\mathcal{A}\rangle^{-\sigma}u\\
		&=v
		\end{align*}
		so that $v\in(P+\mathrm{i})^{-1}\mathscr{D}(\langle\mathcal{A}\rangle^{\sigma})$. Hence, applying (ii) to $v$ yields
		\begin{align*}
		\|\langle\mathcal{A}\rangle^{-\sigma}R(z,hB)\langle\mathcal{A}\rangle^{-\sigma}u\|&=\|\langle\mathcal{A}\rangle^{-\sigma}v\|\nonumber\\
		&\lesssim h^{-1}\|\langle\mathcal{A}\rangle^{-\sigma}u\|+h^{-1}\|\langle\mathcal{A}\rangle^{\sigma}(P-f(z,hB))\chi(P)R(z,hB)\langle\mathcal{A}\rangle^{-\sigma}u\|\nonumber\\
		&\lesssim h^{-1}\|\langle\mathcal{A}\rangle^{-\sigma}u\|+h^{-1}\|\langle\mathcal{A}\rangle^{\sigma}[P-f(z,hB),\chi(P)]R(z,hB)\langle\mathcal{A}\rangle^{-\sigma}u\|\nonumber\\
		&+h^{-1}\|\langle\mathcal{A}\rangle^{\sigma}\chi(P)\langle\mathcal{A}\rangle^{-\sigma}u\|.
		\end{align*}
		By assumption \eqref{Assumption A} for $k=1$ and Lemma \ref{Lemma 2.1 adapted}, we have
		\begin{align*}
		&\|\langle\mathcal{A}\rangle^{\sigma}[P-f(z,hB),\chi(P)]R(z,hB)\langle\mathcal{A}\rangle^{-\sigma}u\|\\
		&=\|\langle\mathcal{A}\rangle^{\sigma}[z-f(z,hB),\chi(P)]R(z,hB)\langle\mathcal{A}\rangle^{-\sigma}u\|\\
		&\leq\delta(h\|B\|_{_{P}})\|\langle\mathcal{A}\rangle^{\sigma}[\varepsilon(hB),\chi(P)]\langle\mathcal{A}\rangle^{-\sigma}\|\|\langle\mathcal{A}\rangle^{\sigma}R(z,hB)\langle\mathcal{A}\rangle^{-\sigma}u\|\nonumber\\
		&\lesssim h\delta(h\|B\|_{_{P}}).
		\end{align*}
		Therefore, $(i)$ follows from $(ii)$ if $h$ is small enough.
		\item We show that $(iii)$ implies $(ii)$. Let $\tilde{\chi}:=1-\chi$ and let $u\in(P+\mathrm{i})^{-1}\mathscr{D}\left(\langle\mathcal{A}\rangle^{\sigma}\right)$. We write
		\begin{align}
		\label{Control for (iii) implies (ii) in the limiting absorption principle in zone IB 001}
		\|\langle\mathcal{A}\rangle^{-\sigma}u\|&\leq\|\langle\mathcal{A}\rangle^{-\sigma}\chi(P)u\|+\|\langle\mathcal{A}\rangle^{-\sigma}\tilde{\chi}(P)u\|
		\end{align}
		and $(iii)$ implies that
		\begin{align*}
		\|\langle\mathcal{A}\rangle^{-\sigma}\chi(P)u\|&\lesssim h^{-1}\|\langle\mathcal{A}\rangle^{\sigma}(P-f(z,hB))\chi(P)u\|
		\end{align*}
		because $\tau\equiv1$ on $\mathrm{Supp\,}\chi$. In order to control the term involving $\tilde{\chi}(P)$ in \eqref{Control for (iii) implies (ii) in the limiting absorption principle in zone IB 001}, we write $\tilde{\chi}=\psi_{-}+\psi_{+}$ with $\psi_{\pm}\in\mathcal{C}^{\infty}\left(\mathbb{R},[0,1]\right)$ such that $\mathrm{Supp}\ \psi_{-}\subset\left]-\infty,\alpha\right]$ and $\mathrm{Supp}\ \psi_{+}\subset\left[\beta,+\infty\right[$. We also pick $\rho\in\mathcal{C}^\infty_{\mathrm{c}}(\mathbb{R},\mathbb{R})$ such that $\rho\psi_-=\psi_-$. Since $B\in L^{\infty}_{\ell\mathrm{oc}}(P)$, we have for any $v\in\mathscr{D}(P)$
		\begin{align}
		&\Re\big\langle\psi_{-}(P)^{2}(f(z,hB)-P)v,v\big\rangle\nonumber\\
		&=\Re\big\langle\psi_{-}(P)^{2}zv,v\big\rangle+\Re\big\langle\psi_{-}(P)^{2}\delta(h\|B\|_{_{P}})\varepsilon(hB)v,v\big\rangle-\Re\big\langle\psi_{-}(P)^{2}Pv,v\big\rangle\nonumber\\
		&\geq a\|\psi_{-}(P)v\|^2-\delta(h\|B\|_{_{P}})\|\rho(P)\varepsilon(hB)\|_{_{P}}\|\psi_{-}(P)v\|^2-\alpha\|\psi_{-}(P)^2v\|^2\nonumber\\
		&\geq c_-\|\psi_{-}(P)v\|^{2}
		\end{align}
		where $c_->0$ if $h$ is sufficiently small. Using Cauchy-Schwarz inequality, we get $\|\psi_{-}(P)(P-f(z,hB))v\|\geq c_-\|\psi_{-}(P)v\|$ and thus $\|\psi_{-}(P)R(z,hB)v\|\lesssim\|\psi_{-}(P)v\|$. Similarly, one can show $\|\psi_{+}(P)R(z,hB)v\|\lesssim\|\psi_{+}(P)v\|$. These inequalities and $\tilde{\chi}^{2}=(\psi_{-}+\psi_{+})^{2}=\psi_{-}^{2}+\psi_{+}^{2}$ then imply
		\begin{align*}
		\|\tilde{\chi}(P)R(z,hB)v\|\lesssim\|\tilde{\chi}(P)v\|
		\end{align*}
		which in turn implies for $u\in\mathscr{D}(P)$
		\begin{align*}
		\|\langle\mathcal{A}\rangle^{-\sigma}\tilde{\chi}(P)u\|&\lesssim\|\tilde{\chi}(P)u\|\\
		&=\|\tilde{\chi}(P)R(z,hB)(P-f(z,hB))u\|\\
		&\lesssim\|\tilde{\chi}(P)(P-f(z,hB))u\|\\
		&\lesssim\|(P-f(z,hB))u\|.
		\end{align*}
	\end{itemize}
\end{proof}
\begin{proof}[Proof of Theorem \ref{ASALP}.] We show that the regularity \eqref{Assumption I} and the Mourre estimate \eqref{Assumption M} are enough to establish \eqref{ASALP equation}. As pointed out at the beginning of \cite{Ger08}, the key point is the following energy estimate\nolinebreak: for any self-adjoint operators $H$ acting on $\mathcal{H}$, $u\in\mathscr{D}\left(H\right)$, $\tau\in\mathcal{C}^{\infty}_{\mathrm{c}}\left(\mathbb{R},[0,1]\right)$ and $P_{\tau}:=\tau(P)P$, we have
\begin{align}
\label{Energy estimate for limiting absorption principle}
2\Im\big\langle Hu,(P_{\tau}-f(z,hB))u\big\rangle&=\big\langle u,[P_{\tau},\mathrm{i}H]u\big\rangle-2\Im\big\langle u,f(z,hB)Hu\big\rangle
\end{align}
where the commutator must be understood as a quadratic form on $\mathscr{D}(H)$.

We follow the proof of \cite[Thm. 1]{Ger08}. Let $\tau,\chi\in\mathcal{C}^{\infty}_{\mathrm{c}}\left(\mathbb{R},[0,1]\right)$ such that $\chi\equiv1$ on $I$ and $\tau\chi=\chi$ and let
\begin{align*}
F\left(\xi\right)&:=-\int_{\xi}^{+\infty}g(\zeta)^{2}\mathrm{d}\zeta
\end{align*}
with $g\in\mathcal{C}^{\infty}\left(\mathbb{R},[0,1]\right)$ satisfying $g\left(\xi\right)=0$ for $\xi\geq2$ and $g\left(\xi\right)=1$ for $\xi\leq1$. By Lemma \ref{Preliminary estimates for the limiting absorption principle in the zone IB}, it is sufficient to prove the following estimate: for any $z\in J^+$ and $u\in\mathscr{D}(\langle\mathcal{A}\rangle^{\sigma})$,
\begin{align*}
\|\langle\mathcal{A}\rangle^{-\sigma}\chi(P)u\|&\lesssim h^{-1}\|\langle\mathcal{A}\rangle^{\sigma}(P_\tau-f(z,hB))\chi(P)u\|.
\end{align*}
As $P\in\mathcal{C}^{2}\left(\mathcal{A}\right)$, $P$ and $\mathcal{A}$ are self-adjoint and satisfy the Mourre estimate \eqref{Assumption M} on $I$, we can apply the estimate $(3.30)$ in the proof of \cite[Thm. 1]{Ger08}:
\begin{align}
\label{Estimate of Christian for the limiting absorption principle in the zone IB}
\chi(P)[P_\tau,\mathrm{i}F(\mathcal{A})]\chi(P)\gtrsim h\chi(P)\langle\mathcal{A}\rangle^{-2\sigma}\chi(P).
\end{align}
Now we apply the identity \eqref{Energy estimate for limiting absorption principle} with $H=F(\mathcal{A})$: for all $u\in\mathscr{D}(\mathcal{A})$,
\begin{align*}
2\Im\big\langle F(\mathcal{A})u,(P_{\tau}-f(z,hB))u\big\rangle&=\big\langle u,[P_{\tau},\mathrm{i}F(\mathcal{A})]u\big\rangle+2\Im\big\langle f(z,hB)u,F(\mathcal{A})u\big\rangle.
\end{align*}
Since $F<0$ is bounded and $\Im z>0$, we can write for all $h$ sufficiently small
\begin{align*}
&2\Im\big\langle F(\mathcal{A})u,(P_\tau-f(z,hB))u\big\rangle\\
&=\big\langle u,[P_{\tau},\mathrm{i}F(\mathcal{A})]u\big\rangle-2(\Im z)\big\langle u,F(\mathcal{A})u\big\rangle-2\delta(h\|B\|_{_{P}})\Im\big\langle u,\varepsilon(hB)F(\mathcal{A})u\big\rangle\\
&>\big\langle u,[P_\tau,\mathrm{i}F(\mathcal{A})]u\big\rangle-2\delta(h\|B\|_{_{P}})\|\varepsilon(hB)u\|\|F(\mathcal{A})u\|
\end{align*}
where we used that that $\varepsilon(hB)\in\mathcal{B}(\mathscr{D}(\mathcal{A}))$ by Assumption \eqref{Assumption A}. It thus follows
\begin{align}
\label{Last preliminary inequality for the limiting absorption principle in the zone IB}
2\Im\big\langle F(\mathcal{A})u,(P_\tau-f(z,hB))u\big\rangle&\geq\big\langle u,[P_\tau,\mathrm{i}F(\mathcal{A})]u\big\rangle.
\end{align}
Plugging the estimate \eqref{Estimate of Christian for the limiting absorption principle in the zone IB} into inequality \eqref{Last preliminary inequality for the limiting absorption principle in the zone IB} and putting $\chi(P)u$ instead of $u$ yield
\begin{align*}
\|\langle\mathcal{A}\rangle^{-\sigma}\chi(P)u\|^{2}&=\langle u,\chi(P)\langle\mathcal{A}\rangle^{-2\sigma}\chi(P)u\rangle\nonumber\\
&\lesssim h^{-1}\big\langle u,\chi(P)[P_\tau,\mathrm{i}F(\mathcal{A})]\chi(P)u\big\rangle\nonumber\\
&\leq h^{-1}\big|\big\langle F(\mathcal{A})\chi(P)u,(P_\tau-f(z,hB))\chi(P)u\big\rangle\big|.
\end{align*}
Using again the boundedness of $F$, we get
\begin{align*}
\|\langle\mathcal{A}\rangle^{-\sigma}\chi(P)u\|^{2}&\lesssim h^{-1}\|\langle\mathcal{A}\rangle^{-\sigma}\chi(P)u\|\|\langle\mathcal{A}\rangle^{\sigma}(P_\tau-f(z,hB))\chi(P)u\|
\end{align*}
which establishes the point (iii) and thus the point (i) in Lemma \ref{Preliminary estimates for the limiting absorption principle in the zone IB}.
\end{proof}
\end{appendices}
%
%
%
%
%
%
%
%
%
%
%
%
%
%
%
%
%
%
%
%
\small

\begin{thebibliography}{99}
	\bibitem[AIK10]{AlIoKl2010}
	Alexakis, S., Ionescu, A.,  Klainerman, S.: \textit{Uniqueness  of  smooth  stationary  black holes  in  vacuum: small  perturbations  of  the  Kerr  spaces},  Comm.   Math. Phys. {\bf 299}, 89-127 (2010)
	\bibitem[ABR08]{AlBoRa08}
	Alexandrova, I., Bony, J.-F., Ramond, T.: \textit{Semiclassical scattering amplitude at the maximum point of
	the potential}, Asymptotic Analysis \textbf{58}(1–2), 57–125 (2008)
	\bibitem[B04]{Ba}
	Bachelot, A.: \textit{Superradiance and scattering of the charged Klein–Gordon field by a step-like electrostatic potential}, J. Math. Pures et Appl., \textbf{83}, 1179-1239 (2004)
	\bibitem[BM93]{BaMo93}
	Bachelot, A., Motet-Bachelot, A.: \textit{Les r\'{e}sonances d’un trou noir de Schwarzschild}, Ann. Inst. H. Poincar\'{e}, Phys. Théor. \textbf{59}(1), 3–68 (1993)
	%
	%
	\bibitem[BeHa20]{BeHa20}
	Besset, N., H\"{a}fner, D.: \textit{Existence of exponentially growing finite energy solutions for the charged Klein-Gordon equation on the De Sitter-Kerr-Newman metric}, arXiv:2004.02483 (2020)
	%
	\bibitem[BoHa08]{BoHa08}
	Bony, J.-F., H\"{a}fner, D.: \textit{Decay and non-decay of the local energy for the wave equation on the De Sitter-Schwarzschild metric}, Comm. Math. Phys., \textbf{282}(3), 697-719 (2008)
	%
	\bibitem[BM04]{BoMi04}
	Bony, J.-F., Michel, L.: \textit{Microlocalization of resonant states and estimates of the residue of the scattering amplitude}, Commun. Math. Phys. \textbf{246}(2), 375–402 (2004)
	%
	\bibitem[C92]{Ch92}
	Chandrasekhar, S.: \textit{The mathematical theory of black holes}, International Series of Monographs on Physics, Vol. 69, The Clarendon Press, Oxford: Oxford University Press (1992)
	%
	\bibitem[DeB81]{DeBr81}
	De Bruijin, N. G.: \textit{Asymptotic Methods in Analysis}, Dover, New York (1981)
	%
	%
	\bibitem[D11]{DyQNM}
	Dyatlov, S.: \textit{Quasi-normal modes and exponential energy decay for the Kerr-de Sitter black hole}, Commun. Math. Phys. \textbf{306}, 119–163 (2011)
	%
	\bibitem[DZ19]{DyZw}
	Dyatlov, S., Zworski, M.: \textit{Mathematical theory of scattering resonances},	AMS Graduate Studies in Mathematics \textbf{200}, to appear in September 2019.
	%
	\bibitem[GGH13]{GeoGerHa13}
	Georgescu, V., G\'{e}rard, C.,  H\"{a}fner, D.: \textit{Boundary values of resolvents of self-adjoint operators in Krein spaces}, J. Fun. Ana. \textbf{265}, 3245-3304 (2013)
	\bibitem[GGH17]{GeoGerHA17}
	Georgescu, V., G\'{e}rard, C.,  H\"{a}fner, D.: \textit{Asymptotic completeness for superradiant Klein–Gordon equations and applications to the De Sitter–Kerr metric}, J. Eur. Math. Soc. \textbf{19}, 2171–2244 (2017)
	\bibitem[Ge08]{Ger08}
	G\'{e}rard, C.: \textit{A proof of the abstract limiting absorption principle by energy estimates}, J. Fun. Ana., 2707-2724 (2008)
	%
	\bibitem[Gi19]{Gi19}
	Giorgi, E.: \textit{The linear stability of Reissner-Nordstr\"{o}m spacetime: the full subextremal range $|Q|<M$}, arXiv:1910.05630 (2019)
	%
	\bibitem[Gu04]{Gu04}
	Guillarmou, C.: \textit{Meromorphic Properties of the Resolvent on Asymptotically Hyperbolic Manifolds}, Duke Math. J. \textbf{129}(1), 1-37 (2005)
	\bibitem[Ha01]{Ha01}
	H\"{a}fner, D.: \textit{Compl\'etude asymptotique pour l'\'equation des ondes dans une classe d'espaces-temps stationnaires et asymptotiquement plats}, Ann. Institut Fourier, \textbf{51}(3), 779-833 (2001)
	\bibitem[Ha09]{Ha09}
	H\"afner, D.: \textit{Creation of fermions by rotating charged black holes}, M\'{e}moires de la SMF \textbf{117}, 158 pp. (2009)
	\bibitem[Hi17]{Hi17}
	Hintz, P.: \textit{Uniqueness of Kerr-Newman-de Sitter black holes with small angular momenta}, arXiv:1702.05239 (2017)
	\bibitem[Hi18]{HinlsKN}
	Hintz, P.: \textit{Nonlinear stability of the Kerr-Newman-de Sitter family of charged black holes}, Annals of PDE {\bf 4}(11) (2018)
	\bibitem[HV16]{HiVa16}
	Hintz, P.,  Vasy, A.: \textit{The global non-linear stability of the Kerr-de Sitter family of black holes}, Acta Math., \textbf{220}, 1-206 (2018)
	%
	\bibitem[Ho94]{HoIII}
	H\"{o}rmander, L.: \textit{The Analysis of Linear Partial Differential Operators III}, Srpinger (1994)
	%
	\bibitem[KS18]{KlSz17}
	Klainerman, S., Szeftel, J.: \textit{Global Nonlinear Stability of Schwarzschild Spacetime under Polarized Perturbations}, arXiv:1711.07597v2 (2018)
	%
	\bibitem[L53]{Le53}
	Leray, J.: \textit{Hyperbolic differential equations}, The Institute for Advanced Study, Princeton, N. J., 240 pages (1953).
	%
	\bibitem[Ma02]{Ma}
	Martinez, A.: \textit{Resonance free domains for non globally analytic potentials}, Ann. Henri Poincar\'e \textbf{4}, 739-756 (2002)
	\bibitem[MM87]{MaMe87}
	Mazzeo, R., Melrose, R.: Meromorphic extension of the resolvent on complete spaces with asymptotically
	constant negative curvature. J. Funct. Anal. \textbf{75}(2), 260–310 (1987)
	\bibitem[Mok17]{Mo17}
	Mokdad, M.: \textit{Reissner-Nordstrøm-de Sitter Manifold: Photon Sphere and	Maximal Analytic Extension}, Class. Quantum Grav. \textbf{34}, 175014 (2017)
	%
	\bibitem[Mou80]{Mou80}
	Mourre, E.: \textit{Absence of singular continuous spectrum for certain selfadjoint operators}, Commun. Math. Phys. \textbf{78}(3), 391–408 (1980)
	%
	\bibitem[ON83]{ON}
	O'Neill, B.: \textit{Semi-riemannian geometry}, Academic Press Inc. (1983)
	%
	\bibitem[R69]{Ra69}
	Ralston, J.: \textit{Solutions of the wave equation with localized energy}, Commun. Pure Appl. Math. \textbf{22},
	807–823 (1969)
	%
	\bibitem[RS78]{ReSi4}
	Reed, M., Simon, B.: \textit{Methods of modern mathematical physics, Volume IV}, New York: Academic Press (1978)
	\bibitem[SZ97]{SaZw97}
	Sá Barreto, A., Zworski, M.: \textit{Distribution of resonances for spherical black holes}, Math. Res.	Lett. \textbf{4}(1), 103–121 (1997)
	\bibitem[SR14]{SR14}
	Shlapentokh-Rothman, Y.: \textit{Exponentially growing finite energy solutions for the Klein-Gordon equation on sub-extremal Kerr spacetimes}, Comm. Math. Phys. {\bf 329}, 859-891 (2014)
	%
	\bibitem[S87]{Sj87}
	Sj\"{o}strand, J.: \textit{Semiclassical resonances generated by non-degenerate critical points}, Lecture Notes in Math., Vol.1256, pp. 402-429 (2006)
	%
	\bibitem[S96]{Sj96}
	Sj\"{o}strand, J.: \textit{A trace formula and review of some estimates for resonances}, Microlocal analysis
	and spectral theory (Lucca, 1996), NATO Adv. Sci. Inst. Ser. C Math. Phys. Sci., Vol.490, Dordrecht: Kluwer Acad. Publ., pp. 377–437 (1997)
	\bibitem[SZ91]{SjZw91}
	Sj\"{o}strand, J., Zworski, M.: \textit{Complex scaling and the distribution of scattering poles}, J. Amer. Math. Soc., \textbf{4}(4), 729-769 (1991)
	\bibitem[TZ98]{TaZw98}
	Tang, S.-H., Zworski, M.: \textit{From quasimodes to resonances}, Math. Res. Lett. \textbf{5}(3), 261–272 (1998)
	\bibitem[V13]{Va13}
	Vasy, A.: \textit{Microlocal analysis of asymptotically hyperbolic and Kerr-de Sitter spaces} (with an appendix by Semyon Dyatlov), Inventiones Math. {\bf 194}, 381-513 (2013)
	\bibitem[W89]{Wh89}
	Whiting, B.: \textit{Mode stability of the Kerr black hole}, J. Math. Phys. {\bf 30}, 1301-1305 (1989)
	\bibitem[Zh05]{FuZh}
	Zhang, F.: \textit{The Schur complement and its applications}, Numerical Methods and Algorithms, Vol. 4, Springer (2005)	
	\bibitem[Zw99]{Zworski_paper}
	Zworski, M.: \textit{Dimension of the limit set and the density of resonances for convex co-compact hyperbolic surfaces}, Invent. Math., \textbf{136}(2), 353-409 (1999)	
	\bibitem[Zw12]{Zw}
	Zworski, M.: \textit{Semiclassical analysis}, Amer. Math. Soc. (2012)	
\end{thebibliography}
\end{document}